\documentclass[11pt]{article}

\usepackage[margin=1in]{geometry}
\usepackage{amsmath,amssymb,amsthm,mathtools}
\usepackage{bm}
\usepackage{enumitem}
\usepackage{booktabs}
\usepackage{array}
\usepackage{xcolor}
\usepackage[colorlinks=true, linkcolor=blue!70!black, citecolor=green!50!black, urlcolor=blue!70!black]{hyperref}
\usepackage[numbers,sort&compress]{natbib}
\usepackage{cleveref}

\numberwithin{equation}{section}
\allowdisplaybreaks[4]

\newtheorem{theorem}{Theorem}[section]
\newtheorem{lemma}[theorem]{Lemma}

\theoremstyle{definition}
\newtheorem{definition}[theorem]{Definition}

\makeatletter
\theoremstyle{plain}
\newtheorem*{rep@theorem}{\rep@title}
\newcommand{\newreptheorem}[2]{\newenvironment{rep#1}[1]{\def\rep@title{#2~\ref{##1}}\begin{rep@theorem}}{\end{rep@theorem}}}
\makeatother
\newreptheorem{theorem}{Theorem}
\newreptheorem{lemma}{Lemma}

\newcommand{\R}{\mathbb{R}}
\newcommand{\Z}{\mathbb{Z}}
\newcommand{\E}{\mathbb{E}}
\let\Pr\relax
\DeclareMathOperator{\Pr}{Pr}
\newcommand{\I}{\mathbb{I}}
\newcommand{\OPT}{\mathrm{OPT}}
\newcommand{\STR}{\mathrm{STR}}
\newcommand{\FSD}{\succeq_{\mathrm{FSD}}}
\newcommand{\RDOM}{\succeq_r}
\newcommand{\gain}{\mathrm{gain}}
\newcommand{\loss}{\mathrm{loss}}

\title{\vspace{-1.5cm}\bf The Power of Recruiting the Smaller Side: Two Additional Traders Suffice in Two-Sided Markets\thanks{The main result was entirely obtained by Cogentic, an agentic framework for mathematical discovery, using an internal version of Gemini as the base model. The authors contextualized the findings and verified the proofs. The full exposition here is due to the authors aided by different AI models.

The Cogentic harness was built by Yang Cai, Vineet Gupta, Yanchen Jiang, Christopher Liaw, Aranyak Mehta, Grigoris Velegkas, and Di Wang. The problem was proposed and the exposition was led by Yang Cai, Christopher Liaw, Aranyak Mehta, and Mingfei Zhao.

The following authors have additional affiliations beyond Google Research: Yang Cai (Yale University) and Vineet Gupta (Google DeepMind).}}

\author{
\parbox{0.95\textwidth}{\centering
Yang Cai,
Vineet Gupta,
Yanchen Jiang,
Christopher Liaw,\\[0.3em]
Aranyak Mehta,
Grigoris Velegkas,
Di~Wang,
Mingfei Zhao
}\\[1.2em]
\normalsize Google Research\\[0.3em]
\small\texttt{\string{caiy, vineet, yanchenjiang, cvliaw, aranyak, gvelegkas, wadi, mingfei\string}@google.com}
}
\date{}

\begin{document}

\hypersetup{pageanchor=false}
\maketitle

\begin{abstract}
We study Bulow--Klemperer-style competition complexity in two-sided double auctions with $m$ unit-demand buyers drawn i.i.d.\ from $F_B$ and $n$ unit-supply sellers drawn i.i.d.\ from $F_S$.
When $m \ge n$ and buyer valuations first-order stochastically dominate seller costs ($F_B \FSD F_S$), we prove that recruiting just \emph{two} additional sellers enables Seller Trade Reduction ($\STR$), a prior-independent mechanism, to achieve expected Gains From Trade (GFT) at least the first-best GFT of the original market.
When the buyer side is the smaller side of the market ($m\le n$), an analogous result holds for Buyer Trade Reduction with two additional buyers.
This resolves open questions of Babaioff, Goldner, and Gonczarowski~\cite{BabaioffGG20} and Cai, Liaw, Mehta, and Zhao~\cite{CaiLMZ24}. We complement our upper bound by showing that this uniform bound is optimal: already for $m = n = 1$, no prior-independent mechanism (deterministic or randomized) that is dominant-strategy incentive-compatible (DSIC), individually rational (IR), and weakly budget-balanced (WBB) can match the first-best GFT by recruiting only one additional seller.

We further extend our competition complexity characterization to \emph{$r$-dominance} ($F_B \RDOM F_S$), where $r \in (0, 1]$ is a parameter and $1 - F_B(x) \ge r\bigl(1 - F_S(x)\bigr)$ for all $x \in \R$ (recovering $F_B \FSD F_S$ at $r = 1$). When the seller side is the smaller side of the market ($m \ge n \ge 1$), we prove that $c = \lceil 1/r \rceil + 1$ additional sellers suffice for the GFT of $\STR$ to match or exceed the first-best GFT, and that \emph{every} prior-independent DSIC, IR, and WBB mechanism requires at least $\max(2, \lceil 1/r \rceil - 1)$ additional sellers even when $m = n = 1$.
This pins down the number of additional sellers needed to within two traders. Finally, we show that this $\Theta(1/r)$ complexity cannot be asymptotically improved even by recruiting on \emph{both} sides of the market: for $r < 1/64$, if $m = n = \lceil \frac{2}{5r} \rceil$, every prior-independent DSIC, IR, and WBB mechanism requires $c \ge \lceil \frac{1}{64 r} \rceil = \Omega(1/r)$ additional traders per side.
\end{abstract}

\thispagestyle{empty}
\newpage
\thispagestyle{empty}
\tableofcontents
\newpage

\hypersetup{pageanchor=true}
\setcounter{page}{1}
\section{Introduction}\label{sec:intro}

Maximizing the Gains From Trade (GFT) in two-sided markets is a fundamental problem in mechanism design with widespread applications ranging from financial exchanges and the FCC Incentive Auction to online ride-sharing and lodging marketplaces. Unlike one-sided auctions—where the classic Vickrey auction~\cite{Vickrey61} achieves full ex-post efficiency (first-best welfare) and the Myerson auction~\cite{Myerson81} achieves optimal expected revenue—two-sided mechanism design must satisfy incentive compatibility and individual rationality on \emph{both} sides of the market while simultaneously maintaining budget balance (running no deficit). The seminal impossibility theorem of \citet{MyersonSatterthwaite83} demonstrates that these three desiderata are fundamentally in conflict: even in bilateral trade with a single buyer and a single seller ($m = n = 1$), no Bayesian incentive-compatible (BIC), individually rational (IR), and weakly budget-balanced (WBB) mechanism can achieve full efficiency. Furthermore, the second-best mechanism that maximizes expected GFT subject to BIC, IR, and WBB is intricate and depends heavily on the underlying value and cost distributions~\cite{MyersonSatterthwaite83}.

A central paradigm for bypassing the complexity and prior-dependence of optimal mechanisms is the \emph{competition complexity} framework pioneered by \citet{BulowKlemperer96}. In single-item one-sided auctions, \citet{BulowKlemperer96} showed that the expected revenue of a simple, prior-independent second-price auction with $n + 1$ i.i.d.\ regular bidders exceeds that of the Myerson optimal auction with $n$ bidders. Rather than tailoring a complex mechanism to the prior, a market designer can achieve the optimal benchmark simply by recruiting one additional participant and running a prior-independent mechanism.

\citet{BabaioffGG20} initiated the study of Bulow--Klemperer-style questions in two-sided double auction markets, where $m$ unit-demand buyers draw values i.i.d.\ from a distribution $F_B$ and $n$ unit-supply sellers draw costs i.i.d.\ from a distribution $F_S$. They studied two asymmetric variants of McAfee's prior-independent \emph{Trade Reduction} mechanism~\cite{McAfee92}—\emph{Buyer Trade Reduction} ($\mathrm{BTR}$) and \emph{Seller Trade Reduction} ($\STR$).
In the $\STR$ mechanism, we first compute the efficient trade size $r$ (i.e.~the size of a GFT-maximizing matching); if $r = 0$, no trade occurs, and if $r \ge 1$, we check whether there exist at least $r+1$ sellers and that the $r$-th highest buyer bid is at least the $(r+1)$-th lowest seller bid.
If so then the top $r$ buyers and bottom $r$ sellers trade and pay (resp.~get paid) the bid of the $(r+1)$-th seller.
Otherwise, we \emph{reduce} the final trade and the remaining $r-1$ buyers pay the $r$-th buyer's bid and the remaining $r-1$ sellers get paid the $r$-th seller's bid (with no trade if $r = 1$).
$\mathrm{BTR}$ is defined analogously.
When buyer and seller distributions are identical ($F_B = F_S$), \citet{BabaioffGG20} proved that recruiting just $1$ additional trader suffices for $\mathrm{BTR}$ (or $\STR$) to match the first-best Gains From Trade $\OPT(m, n)$ of the original market, with no regularity assumptions whatsoever. \citet{BabaioffGG20} then consider the setting where the buyer's distribution first-order stochastically dominates the seller's distribution ($F_B \FSD F_S$, i.e., $F_B(x) \le F_S(x)$ for all $x \in \R$). Studying recruitment on a \emph{single side} of the market, they showed that for $m \ge n$, recruiting $n(m + 4\sqrt{m})$ additional buyers suffices for $\mathrm{BTR}$ (and symmetrically $m(n + 4\sqrt{n})$ additional sellers suffice for $\STR$), while on the larger side, even when $n = 1$, recruiting $\lfloor \log_2 m \rfloor$ additional buyers is insufficient for any anonymous, deterministic prior-independent mechanism. For bilateral trade ($m = n = 1$), they showed that $4$ additional traders suffice while $1$ additional trader is insufficient for any anonymous, deterministic prior-independent mechanism, leaving open the exact threshold in $\{2, 3, 4\}$.

Subsequently, \citet{CaiLMZ24} showed that the $\Omega(\log m)$ lower bound of~\cite{BabaioffGG20} stems from restricting recruitment to the larger side of the market.
They proved that under $F_B \FSD F_S$ and $m \ge n$, recruiting a constant number of traders $c = O(1)$ on \emph{both} sides of the market enables $\STR$ to achieve expected GFT at least the first-best GFT.
However, the constant appearing in their analysis is at least $20{,}000$. In the conclusion of their paper~\cite[Section~5]{CaiLMZ24}, they posed the natural open question of whether recruiting $O(1)$ traders \emph{only on the smaller side} of the market suffices: when $m \ge n$ buyers face $n$ sellers, how many additional sellers are needed to match the first-best GFT?

\subsection{Our Results}\label{subsec:results}

In this paper, we resolve this question and establish the exact competition complexity of recruiting on the smaller side of a two-sided market. When $m \ge n$ and $F_B \FSD F_S$, we prove that recruiting \emph{just two} additional sellers enables Seller Trade Reduction to achieve expected Gains From Trade at least the first-best optimal GFT of the original $(m, n)$ market.
For notation, we use $\STR(m, n)$ and $\OPT(m, n)$ to denote the expected GFT of $\STR$ and the first-best GFT in a market with $m$ buyers and $n$ sellers, respectively; throughout the paper, we assume that the benchmark first-best GFT is finite ($\OPT(m, n) < \infty$).

\begin{theorem}[Two Traders Suffice]\label{thm:main}
Let $m, n \in \Z_{\ge 1}$ with $m \ge n$. Let $F_B$ and $F_S$ be cumulative distribution functions on $\R$ with $\OPT(m, n) < \infty$ such that $F_B$ first-order stochastically dominates $F_S$ ($F_B \FSD F_S$, i.e., $F_B(x) \le F_S(x)$ for all $x \in \R$). Then recruiting two additional sellers drawn independently from $F_S$ suffices for Seller Trade Reduction to match or exceed the first-best GFT of the original market:
\begin{equation}\label{eq:main-thm}
    \STR(m, n + 2) \ge \OPT(m, n).
\end{equation}
\end{theorem}

Via a standard reduction (i.e.~negating values and costs then swapping the roles of buyers and sellers; see~\cite[Proposition~A.1]{BabaioffGG20}), Theorem~\ref{thm:main} symmetrically implies that when $n \ge m$ and $F_B \FSD F_S$, recruiting \emph{two additional buyers} enables Buyer Trade Reduction to match the first-best GFT.
In particular, Theorem~\ref{thm:main} significantly strengthens the two-sided recruitment guarantee of~\cite{CaiLMZ24}, which recruits a large constant number of traders on both sides, to recruiting only two additional sellers and no additional buyers, and resolves the bilateral trade ($m = n = 1$) open problem of~\cite[Section~4.3]{BabaioffGG20} by improving their upper bound of $4$ additional traders down to $2$. Since social welfare equals Gains From Trade plus the total initial endowment of all sellers, $\STR(m, n + 2) \ge \OPT(m, n)$ also immediately implies that the expected social welfare of $\STR$ in the $(m, n + 2)$ market matches or exceeds the first-best social welfare of the original $(m, n)$ market whenever seller costs are non-negative.

We complement Theorem~\ref{thm:main} by showing that this uniform bound of two cannot be improved: one additional seller does not suffice already in bilateral trade. \citet[Theorem~4.3]{BabaioffGG20} had already proved that recruiting $1$ additional trader fails for \emph{anonymous, deterministic} prior-independent mechanisms when $m = n = 1$. Building on the structure of their hard instance, we extend this impossibility result to hold for \emph{every} prior-independent DSIC, IR, and WBB mechanism—deterministic or randomized.
\begin{theorem}[One Trader is Insufficient]\label{thm:lower-bound}
Let $M$ be any prior-independent mechanism (deterministic or randomized) for $m = 1$ buyer and $k = 2$ sellers that is dominant-strategy incentive-compatible (DSIC), individually rational (IR), and weakly budget-balanced (WBB). Then there exist distributions $F_B \FSD F_S$ on $[0, 1]$ with $\OPT(1, 1) > 0$ such that
\begin{equation}\label{eq:impossibility-m1-n1}
    \mathrm{GFT}_M(1, 2) < \OPT(1, 1).
\end{equation}
\end{theorem}

Here $\mathrm{GFT}_M(1, 2)$ denotes the expected Gains From Trade of $M$ under truthful reporting; formal definitions of general mechanisms, $\mathrm{GFT}_M$, and the DSIC, IR, WBB, and prior-independence axioms are given in Subsection~\ref{subsec:mechanisms}.

\paragraph{Extension to $r$-dominance.}
Beyond first-order stochastic dominance, for a parameter $r \in (0, 1]$, we prove competition complexity bounds on the smaller side of the market that depend on $r$ when buyer valuations \emph{$r$-dominate} seller costs ($F_B \RDOM F_S$, Definition~\ref{def:r-dominance}), i.e.,
\[
    1 - F_B(x) \ge r \bigl(1 - F_S(x)\bigr) \qquad \forall x \in \R,
\]
so that at every price threshold $x$, a buyer's probability of exceeding $x$ is at least an $r$-fraction of a seller's probability of exceeding $x$ (with $r = 1$ recovering $F_B \FSD F_S$).\footnote{Our notion of $r$-dominance is closely related to the \emph{approximate stochastic dominance} condition introduced by \citet{BrustleCDEFV25} to study the competition complexity of prophet inequalities, which requires the tail probability of an online algorithm's accepted value to be at least an $r$-fraction of the tail probability of the offline maximum at every threshold $x$.} Our first result here proves that recruiting $c = \lceil 1/r \rceil + 1$ additional sellers suffices for $\STR$ to match or exceed $\OPT(m, n)$.

\begin{theorem}[Recruiting the Smaller Side under $r$-Dominance]\label{thm:rdom-upper}
Let $r \in (0, 1]$ and let $m, n \in \Z_{\ge 1}$ with $m \ge n$. Let $F_B$ and $F_S$ be cumulative distribution functions on $\R$ satisfying $F_B \RDOM F_S$ and $\OPT(m, n) < \infty$. Then recruiting $c = \lceil 1/r \rceil + 1$ additional sellers drawn independently from $F_S$ suffices for Seller Trade Reduction to match or exceed the first-best GFT of the original market:
\begin{equation}\label{eq:rdom-upper-all-mn}
    \STR\bigl(m,\; n + \lceil 1/r \rceil + 1\bigr) \ge \OPT(m, n).
\end{equation}
\end{theorem}
Note that at $r = 1$, $\lceil 1/r \rceil + 1 = 2$, recovering Theorem~\ref{thm:main}.
Next, we prove that \emph{every} prior-independent DSIC, IR, and WBB mechanism requires at least $\max\bigl(2, \lceil 1/r \rceil - 1\bigr)$ additional sellers even when $m = n = 1$.
This leaves a small gap between the upper bound and lower bound of just $2$ traders.

\begin{theorem}[Mechanism-Independent Lower Bound for Recruiting the Smaller Side]\label{thm:rdom-mech-lb}
Let $r \in (0, 1]$ and let $c \in \Z_{\ge 1}$ satisfy $c < \max\bigl(2, \lceil 1/r \rceil - 1\bigr)$. Let $M$ be any prior-independent mechanism (deterministic or randomized) for $m = 1$ buyer and $k = 1 + c$ sellers that is DSIC, IR, and WBB. Then there exist distributions $F_B \RDOM F_S$ on $[0, 1]$ with $\OPT(1, 1) > 0$ such that
\begin{equation}\label{eq:rdom-mech-lb}
    \mathrm{GFT}_M(1, 1 + c) < \OPT(1, 1).
\end{equation}
\end{theorem}

\paragraph{Two-sided recruitment still requires $\Omega(1/r)$ traders.}
Because Theorem~\ref{thm:rdom-upper} holds for all $m \ge n \ge 1$ and places all $c = \lceil 1/r \rceil + 1$ additional traders on the smaller side, recruiting $c = \lceil 1/r \rceil + 1$ additional traders on \emph{both} sides of the market immediately achieves expected GFT $\STR(m + c, n + c) \ge \OPT(m + c, n) \ge \OPT(m, n)$ across all $m \ge n \ge 1$. A natural question is whether recruiting on both sides can asymptotically beat $\Theta(1/r)$ across all market sizes $m \ge n \ge 1$. Our final result answers this question in the negative: across \emph{all} prior-independent DSIC, IR, and WBB mechanisms (deterministic or randomized, anonymous or non-anonymous), for $r < 1/64$ two-sided recruitment still requires $c \ge \lceil \frac{1}{64 r} \rceil = \Omega(1/r)$ additional traders per side. This holds even when the market is balanced with $m = n = \lceil \frac{2}{5r} \rceil$ traders on both sides.

\begin{theorem}[Mechanism-Independent $\Omega(1/r)$ Lower Bound for Two-Sided Recruitment]\label{thm:general-two-sided-mech-lb}
Let $r \in (0, 1)$ and let $c \in \Z_{\ge 1}$ satisfy $c < \lceil \frac{1}{64 r} \rceil$, and set $n := \lceil \frac{2}{5 r} \rceil$. Then for every prior-independent mechanism $M$ (deterministic or randomized) for $n + c$ buyers and $n + c$ sellers that is DSIC, IR, and WBB, there exist distributions $F_B \RDOM F_S$ on $[0, 1]$ with $\OPT(n, n) > 0$ such that
\begin{equation}\label{eq:general-two-sided-mech-lb}
    \mathrm{GFT}_M(n + c, n + c) < \OPT(n, n).
\end{equation}
\end{theorem}

Table~\ref{tab:comparison} summarizes how Theorems~\ref{thm:main} and~\ref{thm:lower-bound} (and Theorems~\ref{thm:rdom-upper}--\ref{thm:general-two-sided-mech-lb} under $r$-dominance) compare with prior competition complexity bounds for beating $\OPT(m, n)$ in two-sided markets with $m \ge n$.

\begin{table}[ht!]
\centering
\scriptsize
\renewcommand{\arraystretch}{1.3}
\resizebox{\textwidth}{!}{\begin{tabular}{@{}llcccc@{}}
\toprule
& & \multicolumn{2}{c}{\textbf{Upper Bound}} & \multicolumn{2}{c}{\textbf{Lower Bound}} \\
\cmidrule(lr){3-4} \cmidrule(lr){5-6}
\textbf{Setting ($m \ge n$)} & \textbf{Augmented Side} & \textbf{Prev.\ Work} & \textbf{This Work} & \textbf{Prev.\ Work} & \textbf{This Work} \\
\midrule
$F_B = F_S$ & Either side & $1$~\cite{BabaioffGG20} & --- & $1$~\cite{MyersonSatterthwaite83} & --- \\
$F_B \FSD F_S$ ($m = n = 1$) & Either side & $4$~\cite{BabaioffGG20} & $\mathbf{2}$ (Thm.~\ref{thm:main}) & $2$ (det.)~\cite{BabaioffGG20} & $\mathbf{2}$ (Thm.~\ref{thm:lower-bound}) \\
$F_B \FSD F_S$ ($m \ge n$) & Buyers (larger) & $n(m + 4\sqrt{m})$~\cite{BabaioffGG20} & --- & $\lfloor \log_2 m \rfloor + 1$ (det.)~\cite{BabaioffGG20} & --- \\
$F_B \FSD F_S$ ($m \ge n$) & Both sides (total) & $O(1)$~\cite{CaiLMZ24} & $\mathbf{2}$ (Thm.~\ref{thm:main}) & $2$ (det.)~\cite{BabaioffGG20} & $\mathbf{2}$ (Thm.~\ref{thm:lower-bound}) \\
$F_B \FSD F_S$ ($m \ge n$) & \textbf{Sellers (smaller)} & $m(n + 4\sqrt{n})$~\cite{BabaioffGG20} & $\mathbf{2}$ \textbf{(Thm.~\ref{thm:main})} & $2$ (det.)~\cite{BabaioffGG20} & $\mathbf{2}$ \textbf{(Thm.~\ref{thm:lower-bound})} \\
\midrule
$F_B \RDOM F_S$ ($m \ge n$) & Both sides (per side) & --- & {\boldmath$\lceil 1/r\rceil + 1$} (Thm.~\ref{thm:rdom-upper}) & --- & {\boldmath$\lceil\tfrac{1}{64r}\rceil$} (Thm.~\ref{thm:general-two-sided-mech-lb}, \cite{MyersonSatterthwaite83}) \\
$F_B \RDOM F_S$ ($m \ge n$) & \textbf{Sellers (smaller)} & --- & {\boldmath$\lceil 1/r\rceil + 1$} \textbf{(Thm.~\ref{thm:rdom-upper})} & --- & {\boldmath$\max(2,\lceil 1/r\rceil - 1)$} \textbf{(Thm.~\ref{thm:rdom-mech-lb})} \\
\bottomrule
\end{tabular}}
\caption{Additional traders needed for a prior-independent DSIC, IR, and WBB mechanism to achieve expected GFT at least $\OPT(m, n)$ when $m \ge n$. ``Both sides (total)'' counts total recruits across both sides ($m + c_B$ buyers, $n + c_S$ sellers), while ``Both sides (per side)'' counts recruits $c$ added to each side. Upper bounds hold for all $(m, n)$ in the stated regime and lower bounds hold for worst-case $(m, n)$ in that regime (at $n = 1$ for every $m \ge 1$ when augmenting buyers); ``(det.)'' restricts to anonymous deterministic mechanisms, whereas unmarked lower bounds apply to all (possibly randomized, non-anonymous) mechanisms.}
\label{tab:comparison}
\end{table}

\subsection{Proof Overview}\label{subsec:proof-overview}

In this subsection, we give a high-level overview of our proofs. Subsection~\ref{subsec:overview-upper} outlines the upper bound of Theorem~\ref{thm:main} (Section~\ref{sec:upper-bound}), Subsection~\ref{subsec:overview-lower} outlines the matching lower bound of Theorem~\ref{thm:lower-bound} (Section~\ref{sec:lower-bound}), and Subsection~\ref{subsec:overview-rdom} summarizes how Section~\ref{sec:r-dominance} and Appendix~\ref{app:rdom-lower} build on these techniques to establish Theorems~\ref{thm:rdom-upper}--\ref{thm:general-two-sided-mech-lb} under $r$-dominance.

\subsubsection{Upper Bound}\label{subsec:overview-upper}

To motivate our approach, it is instructive to compare it with the techniques of \citet{BabaioffGG20} and \citet{CaiLMZ24}. \citet{BabaioffGG20} recruit on the \emph{buyer} side; their analysis bounds the gap to the augmented optimum and reduces the $n$-seller market to $n$ single-seller markets at a linear overhead in $n$. \citet{CaiLMZ24} instead recruit on \emph{both} sides. Because $\STR$ drops at most one pair from the efficient matching, a natural intuition is to recruit enough traders so that the efficient trade size $r$ increases by at least $1$. However, as observed in~\cite{CaiLMZ24}, adding eligible traders does not deterministically increase $r$; to overcome this, Cai et al.~\cite{CaiLMZ24} coupled all $m + n + 2c$ uniform quantiles and partitioned the extreme quantiles into buckets of size $\Theta(n)$ to construct a ``good'' event $\mathcal{E}_1$ (where $r$ increases by at least $2$) whose probability dominates the ``bad'' event $\mathcal{E}_2$ (where $\OPT > \STR$)—an argument that inherently requires a large constant $c$ on both sides of the market.

In this paper, we obtain the sharp bound of \emph{just two} additional sellers ($\STR(m, n + 2) \ge \OPT(m, n)$) by bypassing global trade-size events and quantile bucketing altogether.

\paragraph{Rewriting and decomposing expected GFT (Subsection~\ref{sec:layer-cake}).}
A key starting point of our analysis is rewriting the realized GFT of both $\STR(m, n + 2)$ and $\OPT(m, n)$ as integrals of indicator functions, rather than comparing their surpluses directly. Notice that in $\STR(m, n+2)$, the $i$-th highest buyer $b^{(i)}$ trades with the $i$-th lowest seller $s_{n+2}^{(i)}$ if and only if $b^{(i)}$ also exceeds the $(i+1)$-th seller's cost $s_{n+2}^{(i+1)}$, generating surplus $(b^{(i)} - s_{n+2}^{(i)}) \I[b^{(i)} \ge s_{n+2}^{(i+1)}]$ as compared to $(b^{(i)} - s_n^{(i)})^+$ in $\OPT(m, n)$. Rewriting each surplus term using $(b - s)^+ = \int_{-\infty}^{\infty} \I[s \le x \le b] \, dx$ (and noting that $b^{(i)} \ge s_{n+2}^{(i+1)}$ implies $b^{(i)} \ge s_{n+2}^{(i)}$), we have
$$\STR(m, n+2) = \sum_{i=1}^{\min\{m, n+1\}}\int_{-\infty}^{\infty}\mathbb{E}\left[\I[s_{n+2}^{(i)} \leq x \text{ and } s_{n+2}^{(i+1)}\leq b^{(i)}]\cdot\I[b^{(i)} \geq x]\right]\, dx,$$
$$\OPT(m, n) = \sum_{i=1}^n \int_{-\infty}^{\infty}\mathbb{E}\left[\I[s_n^{(i)} \leq x]\cdot\I[b^{(i)} \geq x]\right]\, dx.$$

To compare these two expressions, since $m \ge n$ implies $\min\{m, n+1\} \ge n$, we keep only the first $n$ summands in $\STR(m, n+2)$ and subtract the two sums term by term. Conditioning on a buyer valuation $b^{(i)} \ge x$ (and using the independence of buyers and sellers), the difference between the $i$-th integrands reduces to comparing the seller event probabilities
\[
\begin{aligned}
&\Pr\bigl[s_{n+2}^{(i)} \le x \text{ and } s_{n+2}^{(i+1)} \le b^{(i)}\bigr] - \Pr\bigl[s_n^{(i)} \le x\bigr] \\
&\qquad = \underbrace{\Bigl(\Pr\bigl[s_{n+2}^{(i)} \le x\bigr] - \Pr\bigl[s_n^{(i)} \le x\bigr]\Bigr)}_{\eqcolon\,\gain(x)} \;-\; \underbrace{\Pr\bigl[s_{n+2}^{(i)} \le x \text{ and } s_{n+2}^{(i+1)} > b^{(i)}\bigr]}_{\eqcolon\,\loss(x, b^{(i)})},
\end{aligned}
\]
where we suppress the subscript $i$ on $\gain$ and $\loss$ for brevity. This identity makes explicit the two competing forces when moving from $\OPT(m, n)$ to $\STR(m, n + 2)$: a \emph{recruitment gain} $\gain(x)$ (drawing $n + 2$ sellers instead of $n$ increases the probability that at least $i$ sellers have cost at most $x$) against a \emph{trade-reduction loss} $\loss(x, b^{(i)})$ (even when $s_{n+2}^{(i)} \le x$, $\STR$ cancels the $i$-th trade if $s_{n+2}^{(i+1)} > b^{(i)}$).
Taking expectations over $b^{(i)}$ gives
\[
\STR(m, n+2) - \OPT(m, n) \geq \sum_{i=1}^n \int_{-\infty}^{\infty} \underbrace{\mathbb{E}_{b^{(i)}}\left[ (\gain(x) - \loss(x, b^{(i)})) \cdot \I[b^{(i)} \geq x] \right]}_{\eqcolon I_i(x)} \, dx.
\]
Thus, to prove Theorem~\ref{thm:main}, it suffices to show that $I_i(x) \geq 0$ for all $x \in \mathbb{R}$ and $i \in \{1, \dots, n\}$.
Although this decomposition requires only elementary algebra, it introduces a new and useful perspective.
Rather than dealing directly with the cardinal surplus differences $b - s$ in value space, it reduces the GFT comparison to reasoning solely about binary events (whose indicators take values in $\{0, 1\}$) at each fixed threshold $x$. Because $I_i(x)$ depends only on the probabilities of these binary events, this enables a tight analysis in quantile space by explicitly evaluating and comparing the exact probabilities at every threshold $x$.

\paragraph{Balancing the gain and loss in the quantile space (Subsections~\ref{sec:warmup} and~\ref{sec:generalization}).}
We start from the simplest setting where $F_B = F_S = F$ is continuous and $m = n$ (Subsection~\ref{sec:warmup}).
Note that $I_i(x)$ contains an expectation of the $i$-th highest value buyer. So it is technically more convenient to work in quantile space via $w=1-F(b^{(i)})$. Since for $b^{(i)} \ge x$ computing the probability directly gives the exact formula $\loss(x, b^{(i)}) = \binom{n+2}{i} p^i (1 - F(b^{(i)}))^{n+2-i}$ (with $p = F(x) = F_S(x)$), abusing notation slightly by redefining $\loss(x, w) \coloneqq \binom{n+2}{i} p^i w^{n+2-i}$ in quantile coordinates allows us to write $I_i(x) = \int_{0}^{1-p} (\gain(x) - \loss(x, w)) f_{i, n}(w) \, dw$, where $f_{i, n}$ is the p.d.f.~of the $i$-th smallest order statistic of $n$ i.i.d.~uniform random variables in $(0, 1)$. We show that $\gain(x) - \loss(x, w)$ is non-increasing in $w$. 
So if $\gain(x) - \loss(x, 1-p) \geq 0$ then clearly $I_i(x) \geq 0$. Otherwise, we extend the integral range from $0$ to $1$: $I_i(x) \geq \int_{0}^{1} (\gain(x) - \loss(x, \min\{w, 1-p\})) f_{i, n}(w) \, dw$. Since $\int_0^1 f_{i, n}(w)\,dw = 1$, this right-hand side equals $\gain(x) - \E_{w \sim f_{i, n}}[\loss(x, \min\{w, 1 - p\})]$. To evaluate $\gain(x) = \Pr[s_{n+2}^{(i)} \le x] - \Pr[s_n^{(i)} \le x]$, imagine adding the two new sellers one by one: seller $n + 1$ brings the number of sellers with cost at most $x$ from $i - 1$ to $i$ with probability $\binom{n}{i-1} p^i (1 - p)^{n+1-i}$, and seller $n + 2$ similarly adds $\binom{n+1}{i-1} p^i (1 - p)^{n+2-i}$. After some algebraic manipulation to rewrite these two binomial coefficients in terms of $\binom{n+2}{i}$ and $\E_{w \sim f_{i, n}}[w] = \frac{i}{n+1}$, we obtain
\[
\gain(x) = \binom{n+2}{i} p^i \left[ \frac{i}{n+2}(1 - p)^{n+2-i} + \frac{n + 2 - i}{n + 2}\,\E_{w \sim f_{i, n}}[w]\,(1 - p)^{n+1-i} \right].
\]
Meanwhile, $\E_{w \sim f_{i, n}}[\loss(x, \min\{w, 1 - p\})] = \binom{n+2}{i} p^i \E_{w \sim f_{i, n}}[\min\{w, 1 - p\}^{n+2-i}]$. Since $\min\{w, 1 - p\}^{n+2-i}$ is bounded above by $(1 - p)^{n+2-i}$ and also by $w(1 - p)^{n+1-i}$, taking expectations and combining these two bounds with weights $\frac{i}{n+2}$ and $\frac{n+2-i}{n+2}$ (which sum to $1$) shows that $\gain(x) - \E_{w \sim f_{i, n}}[\loss(x, \min\{w, 1 - p\})] \ge 0$.

Although this identical-distribution setting is already covered in Babaioff et al.~\cite{BabaioffGG20}---who in fact establish the tighter bound of a single additional trader when $F_B = F_S$---our analysis here has the advantage that it readily extends to the setting where
(i) $F_B$ stochastically dominates $F_S$ and (ii) $m \geq n$, with minor changes.
The former extension can be done by observing that for $b^{(i)} \ge x$, $F_B \FSD F_S$ implies $1 - F_S(b^{(i)}) \le \min\{w, 1 - p\}$ (with $w = 1 - F_B(b^{(i)})$ in the continuous case), so a nearly identical integral appears (just with a different upper limit $R = 1 - F_B(x^-)$ and the truncated loss $\loss(x, \min\{w, 1-p\})$) while the latter is a straightforward algebraic exercise by showing the per-density GFT difference $\frac{\int_{0}^{R} (\gain(x) - \loss(x, \min\{w, 1-p\})) f_{i, m}(w)\,dw}{\int_{0}^{R} f_{i, m}(w)\,dw}$ (for $R > 0$) is non-decreasing as $m$ increases from $m=n$ to any $m>n$.

\subsubsection{Lower Bound}\label{subsec:overview-lower}
While recruiting a single trader suffices when $F_B = F_S$~\cite{BabaioffGG20}, \citet[Theorem~4.3]{BabaioffGG20} showed that under $F_B \FSD F_S$ with $m = n = 1$, recruiting one additional trader fails for any \emph{anonymous, deterministic} prior-independent mechanism. In Section~\ref{sec:lower-bound}, we build on their two-point distribution and proof to extend this impossibility to \emph{every} prior-independent DSIC, IR, and WBB mechanism $M$, deterministic or randomized (Theorem~\ref{thm:lower-bound}).
We also show that we may assume that the mechanism is symmetric among the sellers without loss of generality since any asymmetric mechanism can be made symmetric by randomizing the seller identities.
Write $(b, s_1, s_2)$ for the buyer's valuation and the two sellers' costs, and $x_1(b, s_1, s_2), x_2(b, s_1, s_2) \in [0, 1]$ for the two sellers' expected trade probabilities. Fixing $\rho \in (0, 1)$ and setting $b_\tau \coloneqq 1 - \rho(1 - \tau) \in (\tau, 1)$ for $\tau \in (0, 1)$, we will show that for some $\tau \in (0, 1)$, Seller~$1$'s trade probability is arbitrarily close to $0$ on $(b_\tau, \tau, 1)$.
By seller symmetry, an analogous statement holds for Seller~$2$ when the valuation and costs are $(b_\tau, 1, \tau)$.
For any $\tau, p \in (0, 1)$, let $F_S^{(\tau)}$ place probability mass $p$ at $\tau$ and $1 - p$ at $1$, and let $F_B^{(\tau)}$ place mass $p$ at $b_\tau$ and $1 - p$ at $1$, so that $F_B^{(\tau)} \FSD F_S^{(\tau)}$.
With some calculation, we have $\OPT(1, 1) \geq p \cdot (b_\tau - \tau) = p(1 - \rho)(1 - \tau)$.

Now consider $M$ at a value of $\tau$ where $x_1(b_\tau, \tau, 1)$ is arbitrarily close to $0$.
If both sellers have cost $1$, then the GFT is at most $0$.
If the buyer's value is $b_\tau$ and one seller's cost is $\tau$ while the other's is $1$, then trade with the cost-$\tau$ seller happens with probability arbitrarily close to $0$ (while trading with the cost-$1$ seller yields negative GFT), contributing arbitrarily little.
On the other hand, we may have trade when the buyer's value is $b_\tau$ and both sellers' costs are $\tau$; this event contributes at most $p^3(1 - \tau)$ to the GFT.
The final case is if the buyer's value is $1$ (probability $1 - p$) and at least one seller's cost is $\tau$ (probability $2p(1 - p) + p^2$); this event contributes at most $(1 - p)\bigl(2p(1 - p) + p^2\bigr)(1 - \tau)$ to the GFT.
Summing up, up to an arbitrarily small multiple of $1 - \tau$, the GFT of $M$ is at most $\bigl(p^3 + (1 - p)(2p(1 - p) + p^2)\bigr)(1 - \tau) = \bigl(p^2 + 2p(1 - p)^2\bigr)(1 - \tau)$, which equals $\frac{7}{8}p(1 - \tau) < \OPT(1, 1)$ at $p = 3/4$ and $\rho \in (0, 1/8)$.

The only remaining step is to show that $x_1(b_\tau, \tau, 1)$ is arbitrarily close to $0$ for some $\tau \in (0, 1)$ (Lemma~\ref{lem:sim-inf-zero}). We prove this in Section~\ref{sec:lower-bound} by analyzing the realized expected GFT on the profile $(1, \tau, 1)$: it is not hard to see that the realized expected GFT is $(1 - \tau)x_1(1, \tau, 1)$, whereas by WBB and IR it is at least the sum of the buyer's and Seller~$1$'s expected utilities, which we lower bound via Myerson's payment identity and IR. Comparing the two sides across $\tau \in (0, 1)$ using DSIC monotonicity shows that $x_1(b_\tau, \tau, 1)$ cannot remain bounded away from $0$ for all $\tau \in (0, 1)$.

\subsubsection{Extension to \texorpdfstring{$r$}{r}-Dominance}\label{subsec:overview-rdom}

\begingroup\emergencystretch=1em
\paragraph{Upper bound under $r$-dominance (Section~\ref{sec:r-dominance}).}
To see what choice of $c \ge 2$ ensures $\STR(m, n + c) \ge \OPT(m, n)$ under $F_B \RDOM F_S$ (Theorem~\ref{thm:rdom-upper}), we follow nearly the same path as in Subsection~\ref{subsec:overview-upper} with $n + c$ sellers in place of $n + 2$. Writing $\gain_i(x) \coloneqq \Pr[s_{n+c}^{(i)} \le x] - \Pr[s_n^{(i)} \le x]$ and $\loss_i(x, w) \coloneqq \binom{n+c}{i} p^i w^{n+c-i}$ in quantile coordinates with $p = F_S(x)$, under $F_B \RDOM F_S$ we replace the upper bound $\min\{w, 1 - p\}$ on $1 - F_S(b^{(i)})$ for $b^{(i)} \ge x$ with $\min\{w/r, 1 - p\}$. Defining a similar (to the above) integrand $I_i(x)$ for $n + c$ sellers, the same monotonicity arguments in $m \ge n$ and $w$ yield $I_i(x) \ge \Pr[b^{(i)} \ge x]\,\widetilde{I}_i(x)$, where
\[
\widetilde{I}_i(x) \coloneqq \gain_i(x) - \E_{w \sim f_{i, n}}\bigl[\loss_i(x, \min\{w/r, 1 - p\})\bigr].
\]
Unlike the $r = 1$ case, however, when $r < 1$ and $n$ is large, both $I_i(x)$ and $\widetilde{I}_i(x)$ can be negative at larger indices $i$, so we cannot argue term by term. Instead, since $b^{(1)} \ge \dots \ge b^{(n)}$ makes the weights $\Pr[b^{(i)} \ge x]$ non-increasing in $i$, summation by parts reduces $\sum_{i=1}^n I_i(x) \ge 0$ to showing that every prefix sum $\sum_{i=1}^j \widetilde{I}_i(x)$ is non-negative ($1 \le j \le n$), allowing the surplus at smaller $i$ to offset any deficit at larger $i$.
\par\endgroup

To show $\sum_{i=1}^j \widetilde{I}_i(x) \ge 0$ (focusing on $p \in (0, 1)$, as both terms in $\widetilde{I}_i(x)$ vanish at $p \in \{0, 1\}$), we bound $\sum_{i=1}^j \gain_i(x)$ and $\sum_{i=1}^j \E_{w \sim f_{i, n}}[\loss_i(x, \min\{w/r, 1 - p\})]$ in terms of $\Pr[s_{n+c-1}^{(j)} > x]$ and $\Pr[s_{n+c-2}^{(j)} > x]$, where $s_k^{(j)}$ denotes the $j$-th lowest cost among $k \ge n$ sellers. Adding the $c$ new sellers one by one ($k = n, \dots, n + c - 1$) as in Subsection~\ref{subsec:overview-upper} and summing over $i = 1, \dots, j$ expresses each new seller's contribution as $p \Pr[s_k^{(j)} > x]$. Since having fewer sellers makes $s_k^{(j)} > x$ more likely, $\Pr[s_k^{(j)} > x] \ge \Pr[s_{n+c-2}^{(j)} > x]$ for $n \le k \le n + c - 2$, giving
\[
\sum_{i=1}^j \gain_i(x) = \sum_{k=n}^{n+c-1} p \Pr\bigl[s_k^{(j)} > x\bigr] \ge p \Pr\bigl[s_{n+c-1}^{(j)} > x\bigr] + (c - 1) \, p \Pr\bigl[s_{n+c-2}^{(j)} > x\bigr].
\]
On the loss side, bounding $\min\{w/r, 1 - p\}^{n+c-i}$ by $(w/r)(1 - p)^{n+c-1-i}$ as in Subsection~\ref{subsec:overview-upper} is too loose when $r$ is small. A tighter bound on $\E_{w \sim f_{i, n}}[\min\{w/r, 1 - p\}^{n+c-i}]$, summed over $i = 1, \dots, j$, yields
\[
\sum_{i=1}^j \E_{w \sim f_{i, n}}\bigl[\loss_i(x, \min\{w/r, 1 - p\})\bigr] \le p \Pr\bigl[s_{n+c-1}^{(j)} > x\bigr] + \frac{(n + c - 1)(1 - p)}{n r(1 - p) + 1} \, p \Pr\bigl[s_{n+c-2}^{(j)} > x\bigr].
\]
Comparing these two bounds cancels the first term $p \Pr[s_{n+c-1}^{(j)} > x]$, so $\sum_{i=1}^j \widetilde{I}_i(x) \ge 0$ whenever $\frac{(n + c - 1)(1 - p)}{n r(1 - p) + 1} \le c - 1$. This inequality holds when $c \ge 1/r + 1$.

\paragraph{Lower bounds under $r$-dominance (Appendix~\ref{app:rdom-lower}).}
For our $r$-dominance lower bounds (Theorems~\ref{thm:rdom-mech-lb} and~\ref{thm:general-two-sided-mech-lb}), we follow a similar path to Subsection~\ref{subsec:overview-lower}. When recruiting $c \ge 2$ additional sellers at $m = n = 1$ (Theorem~\ref{thm:rdom-mech-lb}; the case $c = 1$ follows from Theorem~\ref{thm:lower-bound}), $F_B \RDOM F_S$ allows the buyer distribution to place mass $r(1 - p)$ at value $1$ instead of $1 - p$, shifting the remaining mass to $b_\tau$. For small $p$, a single low-cost seller appears with probability roughly $(c + 1)p$, and the same argument as in Subsection~\ref{subsec:overview-lower} makes trade at $b = b_\tau$ negligible for some $\tau$. This leaves expected GFT roughly at most $(c + 1)r \cdot \OPT(1, 1) < \OPT(1, 1)$ whenever $c < \lceil 1/r \rceil - 1$.

We now turn to our lower bound when recruiting from both sides of the market (Theorem~\ref{thm:general-two-sided-mech-lb}). Let $r \in (0, 1/64)$ be fixed and let $n = \Theta(1/r)$. We set the low-cost seller probability $p_S$ to be very small and modify the buyer distribution so that each buyer has value $0$ with probability $1 - 2r$ and value $b_\tau$ or $1$ with probability roughly $r$ each. For small $p_S$ and $b_\tau$ close to $1$, having multiple cost-$\tau$ sellers is negligible, so $\OPT(n, n)$ is at least roughly $n p_S \bigl(1 - (1 - 2r)^n\bigr)(1 - \tau)$. Now consider any mechanism $M$ with $k = n + c$ traders on each side. Up to $O(p_S^2)$ terms, $M$ achieves positive GFT (at most $1 - \tau$) only when a unique seller has cost $\tau$ (probability $\approx k p_S$) and at least one buyer has positive value (probability $1 - (1 - 2r)^k$). However, if one buyer has value $b_\tau$ and the other $k - 1$ buyers have value $0$ (probability $\approx k r (1 - 2r)^{k-1}$), the argument of Subsection~\ref{subsec:overview-lower} makes trade negligible for some $\tau$. Hence the remaining buyer profiles with positive value have probability roughly $1 - (1 - 2r)^k - k r (1 - 2r)^{k-1}$, bounding $M$'s expected GFT by roughly $k p_S \bigl[1 - (1 - 2r)^k - k r (1 - 2r)^{k-1}\bigr](1 - \tau)$. Some calculations give that $k(1 - (1 - 2r)^k) - n(1 - (1 - 2r)^n) = O(c)$ while $k^2 r (1 - 2r)^{k-1} = \Theta(1/r)$. Thus, if $c \ll 1/r$, we have $\mathrm{GFT}_M(k, k) < \OPT(n, n)$.

\subsection{Related Work}\label{subsec:related-work}

\paragraph{Competition complexity in two-sided markets.}
Our work directly continues the line of research on Bulow--Klemperer-style competition complexity in two-sided markets initiated by \citet{BabaioffGG20} and later followed up by \citet{CaiLMZ24}; Table~\ref{tab:comparison} collects the resulting recruitment bounds. Two features distinguish this line of work from the approximation literature discussed below. First, the benchmark is the first-best GFT of the \emph{original} market, and the guarantee is exact rather than a constant-factor approximation. Second, the mechanisms are \emph{prior-independent}: $\mathrm{BTR}$ and $\STR$ are defined without reference to $F_B$ or $F_S$, and the prior enters only through the sampling of the agents. Babaioff et al.~\cite{BabaioffGG20} augment a single side of the market (obtaining bounds that grow with $m$ and $n$), and Cai et al.~\cite{CaiLMZ24} augment both sides. Theorems~\ref{thm:main} and~\ref{thm:lower-bound} settle the question raised in \cite[Section~4.3]{BabaioffGG20} about the optimal number of recruits one needs for $m = n = 1$ (they showed the answer is in $\{2, 3, 4\}$) and the question in \cite[Section~5]{CaiLMZ24} of what is achievable by augmenting only the smaller side.

\paragraph{Approximating gains from trade and welfare in two-sided markets.}
Following Myerson and Satterthwaite's impossibility result~\cite{MyersonSatterthwaite83} and McAfee's Trade Reduction mechanism \cite{McAfee92}, a rich literature has studied approximately efficient mechanisms for bilateral trade and double auctions. These results differ along two axes worth separating: the objective (gains from trade versus social welfare) and the benchmark (the first-best optimum versus the second-best optimum, i.e., the best BIC, IR, and WBB mechanism). In bilateral trade ($m = n = 1$), McAfee~\cite{McAfee08} showed that posting any fixed price between the seller's median and the buyer's median achieves $\frac{1}{2}\OPT(1, 1)$ whenever $\mathrm{med}(F_B) \ge \mathrm{med}(F_S)$, a condition implied by $F_B \FSD F_S$. Subsequent works established constant-factor GFT and welfare approximations using prior-dependent posted-price or randomized-offerer mechanisms~\cite{BlumrosenM16, ColiniBaldeschiGKLT17, BrustleCWZ17, BlumrosenD21}. In a major breakthrough, Deng, Mao, Sivan, and Wang \cite{DengMSW22} proved that a simple mechanism, which with probability $\frac{1}{2}$ lets the seller make a take-it-or-leave-it offer to the buyer and otherwise lets the buyer make a take-it-or-leave-it offer to the seller, achieves a universal constant fraction of the \emph{first-best} GFT for arbitrary distributions, thereby bounding the gap between first-best and second-best GFT by a constant. The constant was subsequently improved by Fei~\cite{Fei22} and, most recently, by Liu, Qin, Ren, and Wang~\cite{LiuQRW26}, who established the tight bound that the second-best GFT is always at least $\frac{1}{2}$ of the first-best GFT. Kang, Pernice, and Vondr{\'a}k~\cite{KangPV22} and Cai and Wu~\cite{CaiW23} studied fixed-price approximation guarantees in bilateral trade under different information assumptions; for social welfare, Cai and Wu~\cite{CaiW23} obtained improved bounds under full prior information and a tight $\frac{2}{3}$ approximation when the designer knows only the mean of a single agent's value. Recent works have also studied bilateral trade with interactive communication~\cite{MaoPLW22}, online and no-regret learning~\cite{BabaioffFN24, BernasconiCCF24}, and strategic intermediation~\cite{HajiaghayiHPS25}. 

Beyond bilateral trade, approximation guarantees have been developed for welfare in double auctions and combinatorial two-sided markets~\cite{BabaioffN04, DuettingRT14, ColiniBaldeschiKLT16, ColiniBaldeschiGKLRT20, BraunK23} and for GFT in single-dimensional markets~\cite{BabaioffCGZ18, BrustleCWZ17, SegalHaleviHA18a, SegalHaleviHA18b}. 
A recent work by Babaioff, Rubinstein, Tan, and Wang~\cite{BabaioffRTW26} proves a constant-factor approximation to the first-best GFT in any matching market. For multi-dimensional two-sided markets, Cai, Goldner, Ma, and Zhao~\cite{CaiGMZ21} obtained the first logarithmic approximation guarantees for GFT in a setting with one buyer and multiple unit-supply sellers. A recent work by Rubinstein, Tan, and Zhou~\cite{RubinsteinTZ26} obtained constant-factor guarantees to the first-best GFT in more general settings where both sides of the market are multi-dimensional. In contrast to these multiplicative approximations, which either lose a factor of $\OPT$ or rely on knowledge of the prior, our result achieves $100\%$ of the original market's first-best GFT via a single prior-independent mechanism by recruiting two sellers.

\paragraph{Bulow--Klemperer-style results in one-sided markets.}
In one-sided revenue maximization, Bulow and Klemperer's theorem~\cite{BulowKlemperer96} has been extended across a wide array of single- and multi-dimensional auction environments. For single-dimensional settings, Hartline and Roughgarden~\cite{HartlineR09} and Dughmi, Roughgarden, and Sundararajan~\cite{DughmiRS09} extended competition complexity to non-identical regular bidders and matroid feasibility constraints; Sivan and Syrgkanis~\cite{SivanS13} treated irregular distributions; Dhangwatnotai, Roughgarden, and Yan~\cite{DhangwatnotaiRY15} and Fu, Liaw, and Randhawa~\cite{FuLR19} connected single-sample and duplicate-bidder mechanisms to optimal revenue; and Liu and Psomas~\cite{LiuP18} studied dynamic auctions. For multi-dimensional settings, competition complexity bounds have been established for unit-demand and additive bidders over independent items~\cite{RoughgardenTCY12, EdenFFTW17, FeldmanFR18, BeyhaghiW19, CaiS21}.

\subsection{Paper Organization}\label{subsec:organization}

Section~\ref{sec:preliminaries} introduces our model and preliminaries. Section~\ref{sec:upper-bound} proves our main upper bound, Theorem~\ref{thm:main}, Section~\ref{sec:lower-bound} proves our matching lower bound, Theorem~\ref{thm:lower-bound}, and Section~\ref{sec:r-dominance} extends our upper bound analysis to $r$-dominance (Theorem~\ref{thm:rdom-upper}). Omitted proofs from Sections~\ref{sec:upper-bound}--\ref{sec:r-dominance}, together with the proofs of our $r$-dominance lower bounds (Theorems~\ref{thm:rdom-mech-lb} and~\ref{thm:general-two-sided-mech-lb}), appear in Appendices~\ref{app:upper-bound}--\ref{app:r-dominance}.

\section{Model and Preliminaries}\label{sec:preliminaries}

We consider a two-sided market with $m$ buyers and $k$ sellers ($m, k \in \Z_{\ge 1}$). Independently across all agents, each buyer's value $b_j$ is drawn i.i.d.\ from a distribution $D_B$ with CDF $F_B$ on $\R$, and each seller's cost $s_\ell$ is drawn i.i.d.\ from a distribution $D_S$ with CDF $F_S$ on $\R$. Denote $\mathbf{b} = (b_1, \dots, b_m)$ and $\mathbf{s}_k = (s_1, \dots, s_k)$. Here we write the subscript $k$ explicitly to distinguish the profile of the market before and after recruitment. By standard probabilistic convention, $F_B$ and $F_S$ are right-continuous and non-decreasing on $\R$, with left limits $F(x^-) := \lim_{z \uparrow x} F(z)$. We make no continuity or support assumptions on $F_B$ or $F_S$. We denote by $n$ the number of sellers in the original market.

\begin{definition}[Order Statistics and Quantile Functions]\label{def:setting}
Given a valuation profile $\mathbf{b} \in \R^m$, we sort buyers in non-increasing order:
\[
    b^{(1)} \ge b^{(2)} \ge \dots \ge b^{(m)}.
\]
Given a seller cost profile $\mathbf{s}_k \in \R^k$, we sort sellers in non-decreasing order:
\[
    s_k^{(1)} \le s_k^{(2)} \le \dots \le s_k^{(k)},
\]
and define the sentinel cost $s_k^{(k+1)} := +\infty$. For any CDF $F$, its left-continuous generalized inverse (quantile function) $F^{-1} : (0, 1) \to \R$ is defined by
\[
    F^{-1}(u) := \inf\bigl\{ z \in \R : F(z) \ge u \bigr\} \qquad \text{for } u \in (0, 1).
\]
We denote the buyer quantile function by $F_B^{-1}(u) := \inf\{z \in \R : F_B(z) \ge u\}$. By right-continuity of $F_B$, we have that $F_B^{-1}(u) \le z$ if and only if $u \le F_B(z)$ for all $u \in (0, 1)$ and $z \in \R$; in particular, setting $z = F_B^{-1}(u)$ yields $F_B(F_B^{-1}(u)) \ge u$.
\end{definition}

\begin{definition}[First-Order Stochastic Dominance and $r$-Dominance]\label{def:r-dominance}
We say that the buyer distribution $F_B$ \emph{first-order stochastically dominates} the seller distribution $F_S$, written $F_B \FSD F_S$, if $F_B(x) \le F_S(x)$ for all $x \in \R$, i.e., $\Pr_{b \sim F_B}[b > x] \ge \Pr_{s \sim F_S}[s > x]$ for every threshold $x \in \R$. More generally, for a parameter $r \in (0, 1]$, we say that $F_B$ \emph{$r$-dominates} $F_S$, written $F_B \RDOM F_S$, if
\begin{equation}\label{eq:r-dominance-def}
    1 - F_B(x) \ge r \bigl(1 - F_S(x)\bigr) \qquad \forall x \in \R,
\end{equation}
recovering $F_B \FSD F_S$ at $r = 1$. Since $1 - F_S(x) \ge 0$,~\eqref{eq:r-dominance-def} is monotone in $r$: whenever $0 < r' \le r \le 1$, $F_B \RDOM F_S$ implies $F_B \succeq_{r'} F_S$.
\end{definition}

\subsection{Mechanisms and Incentives}
\label{subsec:mechanisms}

A mechanism $M$ in a two-sided market specifies, for any reported profile of buyer valuations $\mathbf{b} = (b_1, \dots, b_m)$ and seller costs $\mathbf{s}_k = (s_1, \dots, s_k)$, a (possibly randomized) allocation rule and payment rule:
\begin{itemize}[leftmargin=2em, topsep=2pt, itemsep=2pt]
    \item \textbf{Allocation rule:} For each buyer $j \in \{1, \dots, m\}$ and seller $\ell \in \{1, \dots, k\}$, Borel-measurable expected allocation probabilities $x_j^B(\mathbf{b}, \mathbf{s}_k) \in [0, 1]$ and $x_\ell^S(\mathbf{b}, \mathbf{s}_k) \in [0, 1]$ (over any internal randomness of the mechanism), subject to the feasibility constraint $\sum_{j=1}^m x_j^B(\mathbf{b}, \mathbf{s}_k) = \sum_{\ell=1}^k x_\ell^S(\mathbf{b}, \mathbf{s}_k)$.
    \item \textbf{Payment rule:} Expected payments $p_j^B(\mathbf{b}, \mathbf{s}_k) \in \mathbb{R}$ paid by buyer $j$, and $p_\ell^S(\mathbf{b}, \mathbf{s}_k) \in \mathbb{R}$ paid to seller $\ell$.
\end{itemize}
Assuming quasi-linear utilities, buyer $j$'s expected utility is $b_j\cdot x_j^B(\mathbf{b}, \mathbf{s}_k) - p_j^B(\mathbf{b}, \mathbf{s}_k)$, while seller $\ell$'s expected utility is $p_\ell^S(\mathbf{b}, \mathbf{s}_k) - s_\ell\cdot x_\ell^S(\mathbf{b}, \mathbf{s}_k)$.
A mechanism is \emph{Dominant-Strategy Incentive Compatible} (DSIC) if every agent maximizes their expected utility when reporting truthfully, regardless of what the other agents report. We say that a mechanism is \emph{Individually Rational} (IR) if every agent obtains non-negative expected utility when reporting truthfully, no matter what the other agents report. A mechanism is \emph{Weakly Budget-Balanced} (WBB) if it never runs a deficit (in expectation over the mechanism's internal randomness) at any profile, i.e., $\sum_{j=1}^m p_j^B(\mathbf{b}, \mathbf{s}_k) \ge \sum_{\ell=1}^k p_\ell^S(\mathbf{b}, \mathbf{s}_k)$ for all realization profiles $(\mathbf{b}, \mathbf{s}_k)$. Under truthful reporting, the realized expected Gains From Trade at profile $(\mathbf{b}, \mathbf{s}_k)$ is $\mathrm{gft}_M(\mathbf{b}, \mathbf{s}_k) := \sum_{j=1}^m b_j \, x_j^B(\mathbf{b}, \mathbf{s}_k) - \sum_{\ell=1}^k s_\ell \, x_\ell^S(\mathbf{b}, \mathbf{s}_k)$, and we denote its expectation over $(\mathbf{b}, \mathbf{s}_k) \sim F_B^m \times F_S^k$ by $\mathrm{GFT}_M(m, k) = \mathrm{GFT}_M(m, k; F_B, F_S) := \E_{\mathbf{b} \sim F_B^m, \, \mathbf{s}_k \sim F_S^k}[\mathrm{gft}_M(\mathbf{b}, \mathbf{s}_k)]$. Finally, a mechanism is \emph{prior-independent} if its allocation and payment rules do not depend on the underlying distributions $F_B$ and $F_S$.

\begin{definition}[First-Best Gains From Trade ($\OPT$)]\label{def:opt}
For a market with $m$ buyers and $n$ sellers ($m, n \in \Z_{\ge 1}$), the efficient trade size is $r(\mathbf{b}, \mathbf{s}_n) := \max\bigl(\{0\} \cup \{i \in \{1, \dots, \min(m, n)\} : b^{(i)} \ge s_n^{(i)}\}\bigr)$. When $m \ge n \ge 1$, the ex-post first-best Gains From Trade is
\[
    \OPT(\mathbf{b}, \mathbf{s}_n) := \sum_{i=1}^{r(\mathbf{b}, \mathbf{s}_n)} \bigl(b^{(i)} - s_n^{(i)}\bigr) = \sum_{i=1}^n \bigl(b^{(i)} - s_n^{(i)}\bigr)^+,
\]
where $z^+ := \max(z, 0)$. We denote its expectation by $\OPT(m, n) := \E_{\mathbf{b} \sim F_B^m, \, \mathbf{s}_n \sim F_S^n}[\OPT(\mathbf{b}, \mathbf{s}_n)]$ (also written $\OPT(m, n; F_B, F_S)$ when emphasizing the distributions); throughout the paper, we assume $\OPT(m, n) < \infty$.
\end{definition}

\begin{definition}[Seller Trade Reduction ($\STR$)]
\label{def:str}
Given $\mathbf{b} \in \R^m$ and $\mathbf{s}_k \in \R^k$, let $r' := r(\mathbf{b}, \mathbf{s}_k)$ be the efficient trade size. If $r' = 0$, no trade occurs ($t(\mathbf{b}, \mathbf{s}_k) := 0$). When $r' \ge 1$, if $b^{(r')} \ge s_k^{(r'+1)}$ (the \emph{unreduced} case), the top $r'$ pairs trade at price $s_k^{(r'+1)}$. Otherwise (the \emph{reduced} case), the top $r' - 1$ pairs trade, with buyers paying $b^{(r')}$ and sellers receiving $s_k^{(r')}$ (no trade if $r' = 1$). 

Denoting the number of traded pairs for $r' \ge 1$ by $t(\mathbf{b}, \mathbf{s}_k) := r' - \I[b^{(r')} < s_k^{(r'+1)}]$ (and $t(\mathbf{b}, \mathbf{s}_k) := 0$ when $r' = 0$), the ex-post GFT of $\STR$ is
\[
\STR(\mathbf{b}, \mathbf{s}_k) := \sum_{i=1}^{t(\mathbf{b}, \mathbf{s}_k)} \bigl(b^{(i)} - s_k^{(i)}\bigr),
\]
and we denote its expectation by $\STR(m, k) := \E_{\mathbf{b} \sim F_B^m, \, \mathbf{s}_k \sim F_S^k}[\STR(\mathbf{b}, \mathbf{s}_k)]$.
\end{definition}

\begin{lemma}[\cite{BabaioffGG20}]
Seller Trade Reduction is DSIC, IR, WBB, and prior-independent.
\end{lemma}

\section{Upper Bound: Recruiting Two Sellers Suffices}\label{sec:upper-bound}

In this section, we prove our main upper bound, Theorem~\ref{thm:main}, in three steps. Subsection~\ref{sec:layer-cake} sets up a general decomposition comparing $\STR(m, n + 2)$ and $\OPT(m, n)$ without any distributional assumptions on $F_B$ or $F_S$; Subsection~\ref{sec:warmup} proves the bound in the warm-up setting where $F_B = F_S$ is continuous and $m = n$, capturing the core cancellation between recruitment gain and trade-reduction loss; and Subsection~\ref{sec:generalization} shows that general distributions satisfying $F_B \FSD F_S$ and general market sizes $m \ge n$ reduce directly to the warm-up analysis of Subsection~\ref{sec:warmup}.

\subsection{Decomposing Expected Gains From Trade}\label{sec:layer-cake}

We begin by decomposing the GFT of both $\STR(m, n + 2)$ and $\OPT(m, n)$ across buyer-seller ranks $i = 1, \dots, n$, and rewriting each rank's surplus via $(b - s)^+ = \int_{-\infty}^{\infty} \I[s \le x \le b] \, dx$ so that the difference $\STR(m, n + 2) - \OPT(m, n)$ is bounded below by a sum over $i \in \{1, \dots, n\}$ of integrals over $x \in \R$. The following lemma is the main takeaway of this subsection: it shows that proving Theorem~\ref{thm:main} reduces to checking that for each rank $i \in \{1, \dots, n\}$ and each $x \in \R$, the integrand $I_i(x)$ is non-negative---an expression that cleanly separates the recruitment gain from the trade-reduction loss.

\begin{lemma}\label{prop:layer-cake-main}
Let $m \ge n \ge 1$ and let $F_B, F_S$ be any distributions on $\R$ with $\OPT(m, n) < \infty$. Then
\begin{equation}\label{eq:surplus-diff-sum}
    \STR(m, n + 2) - \OPT(m, n) \ge \sum_{i=1}^n \int_{-\infty}^{\infty} I_i(x) \, dx,
\end{equation}
where for each rank $i \in \{1, \dots, n\}$ and threshold $x \in \R$ with $p := F_S(x) \in [0, 1]$,
\begin{equation}\label{eq:I-i-def}
    I_i(x) := p^i \, \E_{b^{(i)}}\left[ \left( A_{n,i}(p) - \binom{n+2}{i} \bigl(1 - F_S(b^{(i)})\bigr)^{n+2-i} \right) \I\bigl[b^{(i)} \ge x\bigr] \right],
\end{equation}
and $A_{n,i}(p)$ is defined as
\begin{equation}\label{eq:A-def}
    A_{n,i}(p) := \binom{n}{i-1} (1 - p)^{n-i+1} (2 - p) + \binom{n}{i-2} (1 - p)^{n-i+2},
\end{equation}
with the convention $\binom{n}{-1} := 0$.
\end{lemma}

We establish Lemma~\ref{prop:layer-cake-main} via four auxiliary lemmas whose proofs are deferred to Appendix~\ref{app:layer-cake}. First, from Definition~\ref{def:str}, it is not hard to observe that $\STR(\mathbf{b}, \mathbf{s}_k)$ can be written as
\[
    \STR(\mathbf{b}, \mathbf{s}_k) = \sum_{i=1}^{\min(m, k-1)} \bigl(b^{(i)} - s_k^{(i)}\bigr) \I\bigl[b^{(i)} \ge s_k^{(i+1)}\bigr],
\]
since each pair $(b^{(i)}, s_k^{(i)})$ trades in $\STR$ if and only if $b^{(i)} \ge s_k^{(i+1)}$. When $m \ge n$ and $k \ge n + 1$, we have $\min(m, k - 1) \ge n$, and dropping the non-negative summands for $i > n$ yields the following lower bound.

\begin{lemma}\label{lem:str-indicator}
For any $m \ge n \ge 1$, any $k \ge n + 1$, and any realization $(\mathbf{b}, \mathbf{s}_k) \in \R^m \times \R^k$,
\begin{equation}\label{eq:str-rank-lower-bound}
    \STR(\mathbf{b}, \mathbf{s}_k) \ge \sum_{i=1}^n \bigl(b^{(i)} - s_k^{(i)}\bigr) \I\bigl[b^{(i)} \ge s_k^{(i+1)}\bigr].
\end{equation}
\end{lemma}

Next, for each rank $i \in \{1, \dots, n\}$, we compare the rank-$i$ term in~\eqref{eq:str-rank-lower-bound} against $(b^{(i)} - s_n^{(i)})^+$ in $\OPT(m, n)$. Using $b - s = \int_{-\infty}^{\infty} \I[s \le x \le b]\,dx$ whenever $b \ge s$ (and noting that $b^{(i)} \ge s_k^{(i+1)}$ implies $b^{(i)} \ge s_k^{(i)}$), we can write the threshold-$x$ integrands of $(b^{(i)} - s_k^{(i)})\I[b^{(i)} \ge s_k^{(i+1)}]$ and $(b^{(i)} - s_n^{(i)})^+$ as
\[
\begin{aligned}
    \I\bigl[s_k^{(i)} \le x \le b^{(i)}\bigr] \cdot \I\bigl[b^{(i)} \ge s_k^{(i+1)}\bigr] &= \I\bigl[s_k^{(i)} \le x \text{ and } s_k^{(i+1)} \le b^{(i)}\bigr] \cdot \I\bigl[b^{(i)} \ge x\bigr], \\
    \I\bigl[s_n^{(i)} \le x \le b^{(i)}\bigr] &= \I\bigl[s_n^{(i)} \le x\bigr] \cdot \I\bigl[b^{(i)} \ge x\bigr].
\end{aligned}
\]
Integrating over $x \in \R$, conditioning on $b^{(i)}$, and taking expectations over the independent seller profiles $\mathbf{s}_k$ and $\mathbf{s}_n$ turns the difference between the two seller indicators into $\Delta_{n,k,i}(x, b^{(i)})$ in~\eqref{eq:delta-def} below.

\begin{lemma}\label{lem:layer-cake-indicators}
Let $m \ge n \ge i \ge 1$ and $k \ge n + 1$. If $\OPT(m, n) < \infty$, then
\begin{equation}\label{eq:layer-cake-diff}
    \E\Bigl[ \bigl(b^{(i)} - s_k^{(i)}\bigr) \I\bigl[b^{(i)} \ge s_k^{(i+1)}\bigr] - \bigl(b^{(i)} - s_n^{(i)}\bigr)^+ \Bigr] = \int_{-\infty}^{\infty} \E_{b^{(i)}}\Bigl[ \Delta_{n,k,i}\bigl(x, b^{(i)}\bigr) \I\bigl[b^{(i)} \ge x\bigr] \Bigr] dx,
\end{equation}
where
\begin{equation}\label{eq:delta-def}
    \Delta_{n,k,i}(x, b^{(i)}) := \Pr_{\mathbf{s}_k}\bigl[s_k^{(i)} \le x \textnormal{ and } s_k^{(i+1)} \le b^{(i)}\bigr] - \Pr_{\mathbf{s}_n}\bigl[s_n^{(i)} \le x\bigr].
\end{equation}
\end{lemma}

To evaluate $\Delta_{n,k,i}(x, b^{(i)})$ when $b^{(i)} \ge x$, let $p := F_S(x)$ and $q := F_S(b^{(i)})$. Then the first probability in~\eqref{eq:delta-def} satisfies
\[
\begin{aligned}
    \Pr_{\mathbf{s}_k}\bigl[s_k^{(i)} \le x \text{ and } s_k^{(i+1)} \le b^{(i)}\bigr]
    &= \Pr_{\mathbf{s}_k}\bigl[s_k^{(i)} \le x\bigr] - \Pr_{\mathbf{s}_k}\bigl[s_k^{(i)} \le x \text{ and } s_k^{(i+1)} > b^{(i)}\bigr] \\
    &= \Pr_{\mathbf{s}_k}\bigl[s_k^{(i)} \le x\bigr] - \binom{k}{i} p^i (1 - q)^{k-i},
\end{aligned}
\]
since when $b^{(i)} \ge x$, the subtracted event holds if and only if exactly $i$ of the $k$ sellers have cost at most $x$ and the remaining $k - i$ have cost strictly above $b^{(i)}$.

\begin{lemma}\label{lem:prob-diff}
Let $m \ge n \ge i \ge 1$, $k \ge n + 1$, and $x \in \R$, and set $p := F_S(x)$ and $q := F_S(b^{(i)})$. When $b^{(i)}\geq x$,
\begin{equation}\label{eq:prob-diff}
    \Delta_{n,k,i}(x, b^{(i)}) = \Bigl( \Pr_{\mathbf{s}_k}\bigl[s_k^{(i)} \le x\bigr] - \Pr_{\mathbf{s}_n}\bigl[s_n^{(i)} \le x\bigr] \Bigr) - \binom{k}{i} p^i (1 - q)^{k-i}.
\end{equation}
\end{lemma}

Finally, specializing to $k = n + 2$, a straightforward binomial calculation evaluates $\Pr_{\mathbf{s}_{n+2}}[s_{n+2}^{(i)} \le x] - \Pr_{\mathbf{s}_n}[s_n^{(i)} \le x]$ in terms of $A_{n,i}(p)$ from~\eqref{eq:A-def}.

\begin{lemma}\label{lem:cdf-diff-n+2}
For any $n \ge i \ge 1$ and threshold $x \in \R$ with $p := F_S(x) \in [0, 1]$,
\begin{equation}\label{eq:cdf-diff-n+2}
    \Pr_{\mathbf{s}_{n+2}}\bigl[s_{n+2}^{(i)} \le x\bigr] - \Pr_{\mathbf{s}_n}\bigl[s_n^{(i)} \le x\bigr] = p^i A_{n,i}(p).
\end{equation}
\end{lemma}

Combining Lemma~\ref{lem:str-indicator}, Lemma~\ref{lem:layer-cake-indicators}, Lemma~\ref{lem:prob-diff}, and Lemma~\ref{lem:cdf-diff-n+2} immediately proves Lemma~\ref{prop:layer-cake-main}.

\begin{proof}[Proof of Lemma~\ref{prop:layer-cake-main}]
Taking expectations in Lemma~\ref{lem:str-indicator} with $k = n + 2$ and subtracting the first-best benchmark $\OPT(m, n) = \sum_{i=1}^n \E[(b^{(i)} - s_n^{(i)})^+]$ (Definition~\ref{def:opt}) gives
\[
    \STR(m, n + 2) - \OPT(m, n) \ge \sum_{i=1}^n \E\Bigl[\bigl(b^{(i)} - s_{n+2}^{(i)}\bigr) \I\bigl[b^{(i)} \ge s_{n+2}^{(i+1)}\bigr] - \bigl(b^{(i)} - s_n^{(i)}\bigr)^+\Bigr].
\]
Applying Lemma~\ref{lem:layer-cake-indicators} and substituting Lemmas~\ref{lem:prob-diff} and~\ref{lem:cdf-diff-n+2} with $k = n + 2$ yields~\eqref{eq:surplus-diff-sum}--\eqref{eq:I-i-def}.
\end{proof}

\subsection{Warm-Up: Symmetric Continuous Markets (\texorpdfstring{$F_B = F_S, m = n$}{FB = FS, m = n})}\label{sec:warmup}

Before treating general distributions under first-order stochastic dominance, we present the proof of Theorem~\ref{thm:main} in the symmetric continuous setting where $F_B = F_S = F$ is continuous and $m = n$. More specifically, we will show that $I_i(x)\geq 0$ for any $x\in \R$ and $i\in \{1, \dots, n\}$, from which the theorem follows immediately by Lemma~\ref{prop:layer-cake-main}.
We first move to quantile space and express the quantity $I_i(x)$ from~\eqref{eq:I-i-def} as an integral over the order-statistic density.

\begin{lemma}\label{lem:warmup-identity}
Suppose $m = n \ge 1$ and $F_B = F_S = F$ is a continuous CDF on $\R$. Then for every rank $i \in \{1, \dots, n\}$ and threshold $x \in \R$ with $p := F(x) \in [0, 1]$,
\begin{equation}\label{eq:warmup-integral-identity}
    I_i(x) = p^i \int_0^{1-p} \left( A_{n,i}(p) - \binom{n+2}{i} w^{n+2-i} \right) f_{i,n}(w) \, dw,
\end{equation}
where $f_{i,n}(w) := n \binom{n-1}{i-1} w^{i-1}(1 - w)^{n-i}$ is the $\mathrm{Beta}(i, n - i + 1)$ density on $(0, 1)$.
\end{lemma}

\begin{proof}
Since $F$ is continuous, the quantiles $U_j := F(b_j)$ for $j \in \{1, \dots, n\}$ are i.i.d.\ $\mathrm{Uniform}(0, 1)$. Sorting buyers from highest to lowest valuation ($b^{(1)} \ge \dots \ge b^{(n)}$) preserves the order of their uniform quantiles ($U^{(1)} \ge \dots \ge U^{(n)}$), so $F(b^{(i)}) = U^{(i)}$ holds deterministically. Setting $w:= 1 - U^{(i)}$, which is the $i$-th smallest uniform order statistic out of $n$ draws and follows the $\mathrm{Beta}(i, n - i + 1)$ density $f_{i,n}(w)$~\cite{DavidNagaraja03}, we obtain:
\begin{enumerate}[label=(\alph*), leftmargin=2em, topsep=2pt, itemsep=2pt]
    \item The quantile $U^{(i)} = F(b^{(i)}) = \Pr_{s \sim F}[s \le b^{(i)}]$ is the probability that a random seller is willing to trade with buyer $b^{(i)}$, and its complement $w = 1 - U^{(i)} = 1 - F(b^{(i)}) = \Pr_{s \sim F}[s > b^{(i)}]$ is the probability that a random seller is unwilling to trade with buyer $b^{(i)}$ (so $(1 - F(b^{(i)}))^{n+2-i} = w^{n+2-i}$ is the probability that all $n + 2 - i$ remaining sellers are unwilling to trade with $b^{(i)}$).
    \item Since $F$ is non-decreasing and $U^{(i)}$ is continuous ($\Pr[U^{(i)} = p] = 0$), the active participation event $\{b^{(i)} \ge x\}$ is almost surely equivalent to requiring that a random seller is willing to trade with $b^{(i)}$ with probability at least $p = F(x)$, i.e., $\{F(b^{(i)}) \ge F(x)\} = \{U^{(i)} \ge p\} = \{w \le 1 - p\}$.
\end{enumerate}
Substituting~(a) and~(b) into~\eqref{eq:I-i-def} yields~\eqref{eq:warmup-integral-identity}.
\end{proof}

By Lemma~\ref{prop:layer-cake-main} and Lemma~\ref{lem:warmup-identity}, to complete the proof of Theorem~\ref{thm:main} for symmetric continuous markets ($F_B = F_S = F$ and $m = n$), it suffices to show that $I_i(x) \ge 0$ for every threshold $x \in \R$. In fact, we prove a slightly more general inequality—replacing $w$ in the integrand of~\eqref{eq:warmup-integral-identity} with $\min(1 - p, w)$ and allowing the upper limit of integration to be any $r \in [0, 1]$—both because this truncated integrand enables a direct reduction to an integral over the full interval $[0, 1]$, and because this exact generality will be reused in Subsection~\ref{sec:generalization} when we treat general distributions $F_B \FSD F_S$.

\begin{lemma}\label{lem:warmup-reduction}
For any integers $n \ge i \ge 1$ and any $p \in [0, 1]$, define the truncated integrand on $[0, 1]$ by
\begin{equation}\label{eq:psi-def}
    \psi_{n,i}(w; p) := A_{n,i}(p) - \binom{n+2}{i} \min(1 - p, w)^{n+2-i}.
\end{equation}
Then for every $r \in [0, 1]$,
\begin{equation}\label{eq:truncated-integral-nonneg}
    \int_0^r \psi_{n,i}(w; p) \, f_{i,n}(w) \, dw \ge 0.
\end{equation}
\end{lemma}

We defer the proof of Lemma~\ref{lem:warmup-reduction} to Appendix~\ref{app:warmup} and summarize its two steps here. First, because $n + 2 - i \ge 2$, the function $w \mapsto \psi_{n,i}(w; p)$ is non-increasing on $[0, 1]$. For any $r \in [0, 1]$, if $\psi_{n,i}(r; p) \ge 0$, then $\psi_{n,i}(w; p) \ge 0$ on $[0, r]$ and~\eqref{eq:truncated-integral-nonneg} holds immediately; if $\psi_{n,i}(r; p) < 0$, then $\psi_{n,i}(w; p) < 0$ on $[r, 1]$, so extending the upper limit from $r$ to $1$ only decreases the integral. Thus it suffices to verify~\eqref{eq:truncated-integral-nonneg} at $r = 1$.

Second, at $r = 1$, letting $W \sim \mathrm{Beta}(i, n - i + 1)$ (so $\E[W] = \frac{i}{n+1}$), the integral in~\eqref{eq:truncated-integral-nonneg} equals $A_{n,i}(p) - \binom{n+2}{i} \E[\min(1 - p, W)^{n+2-i}]$. Using Pascal's identity and $\E[W] = \frac{i}{n+1}$, we can rewrite $A_{n,i}(p)$ from~\eqref{eq:A-def} as
\[
    A_{n,i}(p) = \binom{n+2}{i} \left[ \frac{i}{n + 2} \, (1 - p)^{n+2-i} + \frac{n + 2 - i}{n + 2} \, \E[W] \, (1 - p)^{n+1-i} \right].
\]
Since $\min(1 - p, W)^{n+2-i}$ is bounded above by both $(1 - p)^{n+2-i}$ and $W (1 - p)^{n+1-i}$ (as $n + 2 - i \ge 2$), taking expectations and averaging with weights $\frac{i}{n+2}$ and $\frac{n+2-i}{n+2}$ gives $A_{n,i}(p) - \binom{n+2}{i} \E[\min(1 - p, W)^{n+2-i}] \ge 0$.

Combining Lemma~\ref{lem:warmup-reduction} with Lemmas~\ref{prop:layer-cake-main} and~\ref{lem:warmup-identity} immediately establishes Theorem~\ref{thm:main} in the symmetric continuous case ($F_B = F_S = F$ and $m = n$): for all $w \in [0, 1 - p]$, we have $\min(1 - p, w) = w$, so $\psi_{n,i}(w; p) = A_{n,i}(p) - \binom{n+2}{i} w^{n+2-i}$. Applying Lemma~\ref{lem:warmup-reduction} at $r = 1 - p \in [0, 1]$ and substituting into Lemma~\ref{lem:warmup-identity} yields
\[
    I_i(x) = p^i \int_0^{1-p} \psi_{n,i}(w; p) \, f_{i,n}(w) \, dw \ge 0 \qquad \forall x \in \R, \; i \in \{1, \dots, n\}.
\]
Integrating over $x \in (-\infty, \infty)$ and summing over $i = 1, \dots, n$ via Lemma~\ref{prop:layer-cake-main} gives $\STR(n, n + 2) \ge \OPT(n, n)$.

\subsection{Generalization to \texorpdfstring{$F_B \FSD F_S$}{FB >= FSD FS} and \texorpdfstring{$m \ge n$}{m >= n}}\label{sec:generalization}

We now extend Subsection~\ref{sec:warmup} to arbitrary distributions satisfying $F_B \FSD F_S$ and arbitrary market sizes $m \ge n$, deferring the formal proofs of Lemmas~\ref{lem:fsd-cdf-bound} and~\ref{lem:mlrp-reduction} to Appendix~\ref{app:generalization}. Recall that Lemma~\ref{prop:layer-cake-main} already holds for arbitrary $F_B, F_S$ and $m \ge n \ge 1$, so it suffices to show $I_i(x) \ge 0$ in~\eqref{eq:I-i-def} for every rank $i \in \{1, \dots, n\}$ and threshold $x \in \R$ with $p := F_S(x)$. Compared to the symmetric continuous warm-up ($F_B = F_S = F$, $m = n$), only two changes arise.

First, consider how $I_i(x)$ in~\eqref{eq:I-i-def} transforms in quantile space. To see the idea in the continuous case, set $w := 1 - F_B(b^{(i)})$ as in Lemma~\ref{lem:warmup-identity}: with $m$ buyers instead of $n$, $w$ has the $\mathrm{Beta}(i, m - i + 1)$ density $f_{i,m}(w)$ instead of $f_{i,n}(w)$, and the active-buyer event $\{b^{(i)} \ge x\}$ becomes $\{w \le 1 - F_B(x)\}$ instead of $\{w \le 1 - p\}$. Meanwhile, whenever $b^{(i)} \ge x$, we can upper bound $1 - F_S(b^{(i)})$ in~\eqref{eq:I-i-def} both by $1 - F_B(b^{(i)}) = w$ (since $F_B \FSD F_S$ means $F_B \le F_S$) and by $1 - F_S(x) = 1 - p$ (since $F_S$ is non-decreasing), giving $1 - F_S(b^{(i)}) \le \min(1 - p, w)$ and replacing the integrand of~\eqref{eq:warmup-integral-identity} with $\psi_{n,i}(w; p)$ from~\eqref{eq:psi-def}. For general $F_B$ (possibly with atoms or flat regions), coupling buyers via $b^{(i)} = F_B^{-1}(U^{(i)})$ for sorted uniform draws $U^{(1)} \ge \dots \ge U^{(m)}$ (Definition~\ref{def:setting}) yields the same bound with upper limit $1 - q_x$ for $q_x := F_B(x^-)$.

\begin{lemma}\label{lem:fsd-cdf-bound}
Suppose $m \ge n \ge 1$ and $F_B \FSD F_S$. Then for every rank $i \in \{1, \dots, n\}$ and every threshold $x \in \R$ with $p := F_S(x)$ and $q_x := F_B(x^-)$, the quantity $I_i(x)$ from~\eqref{eq:I-i-def} satisfies
\begin{equation}\label{eq:general-integral-lower-bound}
    I_i(x) \ge p^i \int_0^{1-q_x} \psi_{n,i}(w; p) \, f_{i,m}(w) \, dw,
\end{equation}
where $\psi_{n,i}(w; p) = A_{n,i}(p) - \binom{n+2}{i} \min(1 - p, w)^{n+2-i}$ is the truncated integrand from~\eqref{eq:psi-def} and $f_{i,m}(w) := m \binom{m-1}{i-1} w^{i-1}(1 - w)^{m-i}$ is the $\mathrm{Beta}(i, m - i + 1)$ density on $(0, 1)$.
\end{lemma}

Second, whereas Lemma~\ref{lem:warmup-reduction} showed $\int_0^r \psi_{n,i}(w; p) f_{i,n}(w)\,dw \ge 0$ for every $r \in [0, 1]$,~\eqref{eq:general-integral-lower-bound} integrates against $f_{i,m}$ for $m \ge n$. Because $f_{i,m}(w) / f_{i,n}(w) \propto (1 - w)^{m-n}$ is non-increasing in $w$, having $m \ge n$ buyers shifts probability mass toward smaller values of $w$, where the non-increasing integrand $\psi_{n,i}(w; p)$ is larger. Consequently, for every $r \in (0, 1]$, the conditional average of $\psi_{n,i}(w; p)$ on $[0, r]$ is at least as large under $f_{i,m}$ as under $f_{i,n}$:
\[
    \frac{\int_0^r \psi_{n,i}(w; p) \, f_{i,m}(w) \, dw}{\int_0^r f_{i,m}(w) \, dw} \ge \frac{\int_0^r \psi_{n,i}(w; p) \, f_{i,n}(w) \, dw}{\int_0^r f_{i,n}(w) \, dw} \ge 0,
\]
where the second inequality holds by Lemma~\ref{lem:warmup-reduction}.

\begin{lemma}\label{lem:mlrp-reduction}
Let $m \ge n \ge i \ge 1$ and $p \in [0, 1]$. Then for every $r \in [0, 1]$,
\begin{equation}\label{eq:general-integral-nonneg}
    \int_0^r \psi_{n,i}(w; p) \, f_{i,m}(w) \, dw \ge 0.
\end{equation}
\end{lemma}

Combining Lemma~\ref{prop:layer-cake-main} with Lemmas~\ref{lem:fsd-cdf-bound} and~\ref{lem:mlrp-reduction} immediately establishes Theorem~\ref{thm:main}.

\begin{proof}[Proof of Theorem~\ref{thm:main}]
Fix any market sizes $m \ge n \ge 1$ and distributions $F_B \FSD F_S$ with $\OPT(m, n) < \infty$. For every threshold $x \in \R$ and rank $i \in \{1, \dots, n\}$, applying Lemma~\ref{lem:fsd-cdf-bound} and Lemma~\ref{lem:mlrp-reduction} at $r = 1 - q_x \in [0, 1]$ yields
\[
    I_i(x) \ge p^i \int_0^{1-q_x} \psi_{n,i}(w; p) \, f_{i,m}(w) \, dw \ge 0.
\]
Integrating over $x \in (-\infty, \infty)$ and summing over $i = 1, \dots, n$ via Lemma~\ref{prop:layer-cake-main} gives $\STR(m, n + 2) - \OPT(m, n) \ge \sum_{i=1}^n \int_{-\infty}^{\infty} I_i(x)\,dx \ge 0$.
\end{proof}

\section{Tightness: Impossibility of Recruiting One Seller}\label{sec:lower-bound}

In this section, we prove Theorem~\ref{thm:lower-bound}, showing that when $m = n = 1$, recruiting a single additional seller cannot match $\OPT(1, 1)$ under $F_B \FSD F_S$ for any deterministic or randomized prior-independent DSIC, IR, and WBB mechanism $M$. Our proof builds closely on the lower bound of Babaioff, Goldner, and Gonczarowski~\cite[Theorem~4.3]{BabaioffGG20}, who proved this impossibility for \emph{deterministic, anonymous} prior-independent mechanisms. Their argument combines a discrete two-point distribution on which Trade Reduction fails to match $\OPT(1, 1)$~\cite[Proposition~4.1]{BabaioffGG20} with a profile-by-profile reduction~\cite[Lemma~4.2 and Lemma~D.1]{BabaioffGG20}: if a deterministic anonymous mechanism $M$ ever trades on a ``mixed'' profile where one seller has low cost $s_1 < b_1$ and the other seller's cost exceeds $b_1$, then determinism and DSIC monotonicity force $M$ to have zero trade whenever the buyer's valuation drops slightly below $s_1$. This allows them to construct a second two-point distribution on which $M$ achieves an arbitrarily small fraction of $\OPT(1, 1)$.

Our construction inherits the two-point structure of~\cite{BabaioffGG20}, in which each seller has either a low cost $\tau \in (0, 1)$ or cost $1$, and the buyer has either an intermediate valuation $b_\tau \in (\tau, 1)$ or valuation $1$. On the mixed seller profile $(s_1, s_2) = (\tau, 1)$, the high-cost seller at $1$ provides no competitive pressure below $1$. Whereas~\cite[Lemma~D.1]{BabaioffGG20} tested only two such distributions and used determinism to force zero trade on one of them, a randomized mechanism can trade with fractional probability across multiple thresholds. To handle arbitrary randomized and non-anonymous mechanisms, we vary $\tau$ continuously over $(0, 1)$ and apply Myerson's payment identity to show that Seller~1's trade probability at $(b_\tau, \tau, 1)$ must have infimum $0$ over $\tau \in (0, 1)$. All omitted proofs in this section are deferred to Appendix~\ref{app:lower-bound}.

Let $M = (x_B, x_1, x_2, p_B, p_1, p_2)$ be any prior-independent DSIC, IR, and WBB mechanism for $m = 1$ buyer and $k = 2$ sellers. Since all test distributions below are supported on $[0, 1]$, restricting $M$ to reports in $[0, 1]^3$ preserves DSIC, IR, and WBB. Moreover, by randomly permuting the two sellers' roles via an independent fair coin before running $M$, we may assume without loss of generality that $M$ is \emph{seller-symmetric}, meaning $x_B(b, s_1, s_2) = x_B(b, s_2, s_1)$ and $x_1(b, s_1, s_2) = x_2(b, s_2, s_1)$ for all $(b, s_1, s_2) \in [0, 1]^3$. By Lemma~\ref{lem:seller-sym-reduction} in Appendix~\ref{app:lower-bound}, this symmetrization preserves DSIC, IR, WBB, and $\mathrm{GFT}_M(1, 2; F_B, F_S)$ for all $(F_B, F_S)$. By Subsection~\ref{subsec:mechanisms}, feasibility requires
\begin{equation}\label{eq:feasibility}
    x_B(b, s_1, s_2) = x_1(b, s_1, s_2) + x_2(b, s_1, s_2) \quad \text{and} \quad x_B, x_1, x_2 \in [0, 1] \qquad \forall (b, s_1, s_2) \in [0, 1]^3,
\end{equation}
so the realized expected Gains From Trade at $(b, s_1, s_2)$ can be written as
\begin{equation}\label{eq:gft-profile}
\begin{aligned}
    \mathrm{gft}_M(b, s_1, s_2)
    &= b \, x_B(b, s_1, s_2) - s_1 \, x_1(b, s_1, s_2) - s_2 \, x_2(b, s_1, s_2) \\
    &= (b - s_1) x_1(b, s_1, s_2) + (b - s_2) x_2(b, s_1, s_2).
\end{aligned}
\end{equation}

On any mixed seller profile $(s_1, s_2) = (\tau, 1)$ with $0 < \tau < b \le 1$, since $b - 1 \le 0$ and $x_2(b, \tau, 1) \ge 0$,~\eqref{eq:gft-profile} immediately gives
\begin{equation}\label{eq:mixed-gft-ub}
    \mathrm{gft}_M(b, \tau, 1) = (b - \tau) x_1(b, \tau, 1) - (1 - b) x_2(b, \tau, 1) \le (b - \tau) x_1(b, \tau, 1),
\end{equation}
with equality at $b = 1$. In particular, for any $\rho \in (0, 1)$ and buyer valuation $b_\tau := 1 - \rho(1 - \tau) \in (\tau, 1)$, the mechanism's expected GFT on $(b_\tau, \tau, 1)$ is bounded by $(b_\tau - \tau) x_1(b_\tau, \tau, 1)$, which becomes arbitrarily small whenever $x_1(b_\tau, \tau, 1)$ is close to $0$. Our main technical lemma shows that no DSIC, IR, and WBB mechanism can keep $x_1(b_\tau, \tau, 1)$ bounded away from $0$ across all $\tau \in (0, 1)$.

\begin{lemma}\label{lem:sim-inf-zero}
Let $M$ be a DSIC, IR, and WBB mechanism on $[0, 1]^3$ satisfying~\eqref{eq:feasibility}. For every $\rho \in (0, 1)$,
\begin{equation}\label{eq:sim-inf-zero}
    \inf_{\tau \in (0, 1)} x_1\bigl(1 - \rho(1 - \tau), \, \tau, \, 1\bigr) = 0.
\end{equation}
\end{lemma}

\begin{proof}
At any profile $(b, s_1, s_2) \in [0, 1]^3$, let $u_B(b, s_1, s_2) := b \, x_B(b, s_1, s_2) - p_B(b, s_1, s_2)$ and, for $j \in \{1, 2\}$, $u_j(b, s_1, s_2) := p_j(b, s_1, s_2) - s_j \, x_j(b, s_1, s_2)$ denote the agents' expected utilities. Suppressing the argument $(b, s_1, s_2)$,~\eqref{eq:gft-profile} and WBB ($p_B - p_1 - p_2 \ge 0$) give
\[
    u_B + u_1 + u_2 = \mathrm{gft}_M(b, s_1, s_2) - (p_B - p_1 - p_2) \le \mathrm{gft}_M(b, s_1, s_2).
\]
Fix any $\tau \in (0, 1)$ and consider the profile $(1, \tau, 1)$, where~\eqref{eq:mixed-gft-ub} gives $\mathrm{gft}_M(1, \tau, 1) = (1 - \tau) x_1(1, \tau, 1)$. Since $u_2(1, \tau, 1) \ge 0$ by IR, it follows that $u_B(1, \tau, 1) + u_1(1, \tau, 1) \le (1 - \tau) x_1(1, \tau, 1)$. On the other hand, by DSIC, $t \mapsto x_B(t, \tau, 1)$ is non-decreasing and $u \mapsto x_1(1, u, 1)$ is non-increasing on $[0, 1]$, and Myerson's payment identity gives
\[
    u_B(1, \tau, 1) = u_B(\tau, \tau, 1) + \int_\tau^1 x_B(t, \tau, 1) \, dt, \qquad
    u_1(1, \tau, 1) = u_1(1, 1, 1) + \int_\tau^1 x_1(1, u, 1) \, du.
\]
Since $u_B(\tau, \tau, 1) \ge 0$ and $u_1(1, 1, 1) \ge 0$ by IR, combining these bounds yields
\begin{equation}\label{eq:sim-reflected-env}
    \int_\tau^1 x_B(t, \tau, 1) \, dt + \int_\tau^1 x_1(1, u, 1) \, du \le (1 - \tau) x_1(1, \tau, 1) \qquad \forall \tau \in (0, 1).
\end{equation}

Now fix $\rho \in (0, 1)$, set $b_\tau := 1 - \rho(1 - \tau) \in (\tau, 1)$, and let $\mu := \inf_{\tau \in (0, 1)} x_1(b_\tau, \tau, 1) \in [0, 1]$. For every $\tau \in (0, 1)$ and $t \in [b_\tau, 1] \subseteq [\tau, 1]$, DSIC monotonicity of $x_B(\cdot, \tau, 1)$ and $x_B = x_1 + x_2 \ge x_1$ from~\eqref{eq:feasibility} give $x_B(t, \tau, 1) \ge x_B(b_\tau, \tau, 1) \ge x_1(b_\tau, \tau, 1) \ge \mu$, so
\[
    \int_\tau^1 x_B(t, \tau, 1) \, dt \ge \int_{b_\tau}^1 \mu \, dt = (1 - b_\tau)\mu = \rho \mu (1 - \tau).
\]
Meanwhile, let $L := \inf_{u \in (0, 1)} x_1(1, u, 1) \in [0, 1]$ (equivalently, $L = \lim_{u \uparrow 1} x_1(1, u, 1)$ since $x_1(1, \cdot, 1)$ is non-increasing). Since $x_1(1, u, 1) \ge L$ for all $u \in (\tau, 1)$ and the single point $u = 1$ has measure zero, we have $\int_\tau^1 x_1(1, u, 1) \, du \ge (1 - \tau) L$. Substituting both lower bounds into~\eqref{eq:sim-reflected-env} and dividing by $1 - \tau > 0$ gives
\[
    \rho \mu + L \le x_1(1, \tau, 1) \qquad \forall \tau \in (0, 1).
\]
Taking the infimum over $\tau \in (0, 1)$ on the right-hand side yields $\rho \mu + L \le L$, and since $\rho > 0$ and $\mu \ge 0$, we conclude $\mu = 0$.
\end{proof}

We now define the family of distributions used to prove the lower bound. Fix any $p, \rho \in (0, 1)$. For each $\tau \in (0, 1)$, recall $b_\tau = 1 - \rho(1 - \tau) \in (\tau, 1)$ and define the two-point distributions $(F_B^{(\tau)}, F_S^{(\tau)})$ on $[\tau, 1]$ as follows: $F_S^{(\tau)}$ places mass $p$ at $\tau$ and $1 - p$ at $1$, while $F_B^{(\tau)}$ places mass $p$ at $b_\tau$ and $1 - p$ at $1$. Since $b_\tau > \tau$, we have $1 - F_B^{(\tau)}(z) = 1 \ge 1 - F_S^{(\tau)}(z)$ for $z < b_\tau$, $1 - F_B^{(\tau)}(z) = 1 - p = 1 - F_S^{(\tau)}(z)$ for $b_\tau \le z < 1$, and $1 - F_B^{(\tau)}(z) = 0 = 1 - F_S^{(\tau)}(z)$ for $z \ge 1$, so $F_B^{(\tau)} \FSD F_S^{(\tau)}$.

Evaluating $\OPT(1, 1)$ and bounding $\mathrm{GFT}_M(1, 2)$ on $(F_B^{(\tau)}, F_S^{(\tau)})$ in Lemmas~\ref{lem:sim-opt-bound} and~\ref{lem:sim-gft-bound} below are straightforward expectation calculations over the seller realizations in $\{\tau, 1\}$, which we defer to Appendix~\ref{app:lower-bound}. In $\OPT(1, 1)$, trade occurs with probability $p$ when the single seller has cost $\tau$, yielding GFT $1 - \tau$ when $b = 1$ and $(1 - \rho)(1 - \tau)$ when $b = b_\tau$. In $\mathrm{GFT}_M(1, 2)$, the seller profile $(1, 1)$ yields expected GFT at most $0$, while $(\tau, \tau)$ occurs with probability $p^2$ and yields expected GFT at most $1 - \tau$. Finally, the mixed profiles $(\tau, 1)$ and $(1, \tau)$ have combined probability $2p(1 - p)$; by seller symmetry and~\eqref{eq:mixed-gft-ub}, each yields expected GFT at most $1 - \tau$ when $b = 1$ and at most $(1 - \tau) x_1(b_\tau, \tau, 1)$ when $b = b_\tau$.

\begin{lemma}\label{lem:sim-opt-bound}
For any $p, \rho, \tau \in (0, 1)$,
\begin{equation}\label{eq:sim-opt-exact}
    \OPT\bigl(1, 1; F_B^{(\tau)}, F_S^{(\tau)}\bigr) = p(1 - p\rho)(1 - \tau) \ge p(1 - \rho)(1 - \tau) > 0.
\end{equation}
\end{lemma}

\begin{lemma}\label{lem:sim-gft-bound}
Let $M$ be any seller-symmetric DSIC, IR, and WBB mechanism satisfying~\eqref{eq:feasibility}. For any $p, \rho, \tau \in (0, 1)$,
\begin{equation}\label{eq:sim-gft-upper}
    \mathrm{GFT}_M\bigl(1, 2; F_B^{(\tau)}, F_S^{(\tau)}\bigr) \le p^2(1 - \tau) + 2p(1 - p)(1 - \tau)\Bigl[(1 - p) + p \, x_1\bigl(b_\tau, \tau, 1\bigr)\Bigr].
\end{equation}
\end{lemma}

\begin{proof}[Proof of Theorem~\ref{thm:lower-bound}]
Fix $p, \rho \in (0, 1)$. Dividing~\eqref{eq:sim-gft-upper} from Lemma~\ref{lem:sim-gft-bound} by $p(1 - \tau) > 0$, taking the infimum over $\tau \in (0, 1)$, and applying~\eqref{eq:sim-inf-zero} from Lemma~\ref{lem:sim-inf-zero} gives
\begin{equation}\label{eq:sim-inf-ub}
    \inf_{\tau \in (0, 1)} \frac{\mathrm{GFT}_M\bigl(1, 2; F_B^{(\tau)}, F_S^{(\tau)}\bigr)}{p(1 - \tau)} \le p + 2(1 - p)^2,
\end{equation}
whereas Lemma~\ref{lem:sim-opt-bound} gives $\frac{\OPT(1, 1; F_B^{(\tau)}, F_S^{(\tau)})}{p(1 - \tau)} \ge 1 - \rho$ for all $\tau \in (0, 1)$. Setting $p = 3/4$ gives $p + 2(1 - p)^2 = 7/8$; taking any $\rho \in (0, 1/8)$ so that $7/8 < 1 - \rho$ then yields $\tau \in (0, 1)$ with $\mathrm{GFT}_M\bigl(1, 2; F_B^{(\tau)}, F_S^{(\tau)}\bigr) < \OPT\bigl(1, 1; F_B^{(\tau)}, F_S^{(\tau)}\bigr)$.
\end{proof}

\section{Extension to \texorpdfstring{$r$}{r}-Dominance}\label{sec:r-dominance}\label{subsec:rdom-all-mn}

In this section, we extend our upper bound analysis from first-order stochastic dominance ($r = 1$) to $r$-dominance ($F_B \RDOM F_S$, Definition~\ref{def:r-dominance}), proving that $c = \lceil 1/r \rceil + 1$ additional sellers on the smaller side suffice (Theorem~\ref{thm:rdom-upper}). Because the argument builds closely on Section~\ref{sec:upper-bound}, we focus here on the main ideas and key differences, deferring omitted lemma proofs to Appendix~\ref{app:rdom-Pj-proof}. The proofs of our $r$-dominance lower bounds (Theorems~\ref{thm:rdom-mech-lb} and~\ref{thm:general-two-sided-mech-lb}) are quite similar to Section~\ref{sec:lower-bound}, so we defer them to Appendix~\ref{app:rdom-lower}.

Fix $r \in (0, 1]$, market sizes $m \ge n \ge 1$, and an integer $c \ge 2$ representing the number of recruited sellers, so that the augmented market has $n + c$ sellers; keeping $c$ as a parameter makes explicit how the requirement $c \ge 1/r + 1$ in Theorem~\ref{thm:rdom-upper} arises. Following the notation of Subsection~\ref{subsec:overview-rdom} (and suppressing the fixed parameters $n, c, r, F_S$, with $c = 2$ and $r = 1$ recovering the quantities of Section~\ref{sec:upper-bound}), for each $i \in \{1, \dots, n\}$, threshold $x \in \R$ with $p := F_S(x) \in [0, 1]$, and $w \ge 0$, define the recruitment gain and quantile-space trade-reduction loss
\begin{equation}\label{eq:rdom-gain-loss-def}
    \gain_i(x) := \Pr\bigl[s_{n+c}^{(i)} \le x\bigr] - \Pr\bigl[s_n^{(i)} \le x\bigr], \qquad
    \loss_i(x, w) := \binom{n+c}{i} p^i w^{n+c-i},
\end{equation}
and the integrand
\begin{equation}\label{eq:rdom-Ii-def}
    I_i(x) := \E_{b^{(i)}}\Bigl[\Bigl(\gain_i(x) - \loss_i\bigl(x, 1 - F_S(b^{(i)})\bigr)\Bigr)\I\bigl[b^{(i)} \ge x\bigr]\Bigr].
\end{equation}
Compared to Section~\ref{sec:upper-bound}, three differences arise, which we capture in Lemmas~\ref{lem:rdom-layer-cake}--\ref{lem:rdom-prefix-bounds} below (proved in Appendix~\ref{app:rdom-Pj-proof}).

First, we extend the GFT decomposition of Lemma~\ref{prop:layer-cake-main} and the quantile bound of Lemma~\ref{lem:fsd-cdf-bound} from $c = 2$ and $F_B \FSD F_S$ to general $c \ge 2$ and $F_B \RDOM F_S$. Since Lemmas~\ref{lem:str-indicator}--\ref{lem:prob-diff} in Subsection~\ref{sec:layer-cake} already hold for $k = n + c$ sellers, adding the $c$ new sellers one by one ($s = n, \dots, n + c - 1$) expresses $\gain_i(x)$ as a sum of $c$ binomial terms. In quantile space, replacing $F_B \FSD F_S$ with $F_B \RDOM F_S$ weakens the upper bound $1 - F_S(b^{(i)}) \le \min\{w, 1 - p\}$ from Subsection~\ref{sec:generalization} to $1 - F_S(b^{(i)}) \le \min\{w/r, 1 - p\}$ on $\{b^{(i)} \ge x\}$.

\begin{lemma}\label{lem:rdom-layer-cake}
Let $m \ge n \ge 1$, $c \in \Z_{\ge 2}$, and let $F_B, F_S$ be any distributions on $\R$ with $\OPT(m, n) < \infty$. Then
\begin{equation}\label{eq:rdom-surplus-diff-sum}
    \STR(m, n + c) - \OPT(m, n) \ge \sum_{i=1}^n \int_{-\infty}^{\infty} I_i(x) \, dx,
\end{equation}
and for every $i \in \{1, \dots, n\}$ and $x \in \R$ with $p := F_S(x) \in [0, 1]$,
\begin{equation}\label{eq:rdom-gain-binom}
    \gain_i(x) = \sum_{s=n}^{n+c-1} \binom{s}{i-1} p^i (1 - p)^{s+1-i}.
\end{equation}
Moreover, if $F_B \RDOM F_S$ ($r \in (0, 1]$), then for every $i \in \{1, \dots, n\}$ and $x \in \R$ with $p := F_S(x)$ and $R := 1 - F_B(x^-)$,
\begin{equation}\label{eq:rdom-Ii-quantile}
    I_i(x) \ge \int_0^R \Bigl(\gain_i(x) - \loss_i\bigl(x, \min\{w/r, 1 - p\}\bigr)\Bigr) f_{i, m}(w) \, dw,
\end{equation}
where $f_{i, m}(w) := m \binom{m-1}{i-1} w^{i-1}(1 - w)^{m-i}$ is the $\mathrm{Beta}(i, m - i + 1)$ density on $(0, 1)$.
\end{lemma}

Second, unlike the $r = 1$ case where $I_i(x) \ge 0$ holds for every $i$, when $r < 1$ and $n$ is large, the rescaled argument $w/r$ in $\loss_i(x, \min\{w/r, 1 - p\})$ can make $I_i(x)$ negative at larger indices $i$, so we must establish $\sum_{i=1}^n I_i(x) \ge 0$ jointly. Define the $m = n$, full-interval analogue
\begin{equation}\label{eq:rdom-Itilde-def}
\begin{aligned}
    \widetilde{I}_i(x) &:= \int_0^1 \Bigl(\gain_i(x) - \loss_i\bigl(x, \min\{w/r, 1 - p\}\bigr)\Bigr) f_{i, n}(w) \, dw \\
    &\phantom{:}= \gain_i(x) - \E_{w \sim f_{i, n}}\bigl[\loss_i(x, \min\{w/r, 1 - p\})\bigr].
\end{aligned}
\end{equation}
When $\Pr[b^{(i)} \ge x] = \int_0^R f_{i, m}(w)\,dw > 0$, because $w \mapsto \gain_i(x) - \loss_i(x, \min\{w/r, 1 - p\})$ and $f_{i, m}(w)/f_{i, n}(w) \propto (1 - w)^{m-n}$ are both non-increasing, dividing the right-hand side of~\eqref{eq:rdom-Ii-quantile} by $\Pr[b^{(i)} \ge x]$ is at least the conditional expectation over $w \le R$ under $f_{i, n}$, which in turn is at least the unconditional expectation $\widetilde{I}_i(x)$ over $w \in [0, 1]$ regardless of sign. Multiplying by $\Pr[b^{(i)} \ge x]$ gives $I_i(x) \ge \Pr[b^{(i)} \ge x] \, \widetilde{I}_i(x)$. Since $b^{(1)} \ge \dots \ge b^{(n)}$ makes the weights $\Pr[b^{(i)} \ge x]$ non-increasing in $i$, summation by parts reduces $\sum_{i=1}^n I_i(x) \ge 0$ to showing that every prefix sum $\sum_{i=1}^j \widetilde{I}_i(x)$ is non-negative.

\begin{lemma}\label{lem:rdom-sbp-reduction}
Suppose $r \in (0, 1]$, $m \ge n \ge 1$, $c \in \Z_{\ge 2}$, and $F_B \RDOM F_S$. Then for every $x \in \R$ and $i \in \{1, \dots, n\}$, we have $I_i(x) \ge \Pr[b^{(i)} \ge x] \, \widetilde{I}_i(x)$, and consequently, setting $\Pr[b^{(n+1)} \ge x] := 0$,
\begin{equation}\label{eq:rdom-sbp-main}
    \sum_{i=1}^n I_i(x) \ge \sum_{j=1}^n \Bigl(\Pr\bigl[b^{(j)} \ge x\bigr] - \Pr\bigl[b^{(j+1)} \ge x\bigr]\Bigr) \sum_{i=1}^j \widetilde{I}_i(x),
\end{equation}
where $\Pr[b^{(j)} \ge x] - \Pr[b^{(j+1)} \ge x] \ge 0$ for all $j \in \{1, \dots, n\}$.
\end{lemma}

Third, even $\widetilde{I}_i(x)$ can be negative at larger indices $i$, so for $p \in (0, 1)$ we lower bound the prefix sum $\sum_{i=1}^j \gain_i(x)$ and upper bound $\sum_{i=1}^j \E_{w \sim f_{i, n}}[\loss_i(x, \min\{w/r, 1 - p\})]$. On the gain side, for each $s \in \{n, \dots, n + c - 1\}$, factoring out one $p$ in~\eqref{eq:rdom-gain-binom} and summing over $i = 1, \dots, j$ gives $p \sum_{i=1}^j \binom{s}{i-1} p^{i-1}(1 - p)^{s-(i-1)} = p \Pr[s_s^{(j)} > x]$; since having fewer sellers makes $s_s^{(j)} > x$ more likely, $\Pr[s_s^{(j)} > x] \ge \Pr[s_{n+c-2}^{(j)} > x]$ for all $c - 1$ indices $s \in \{n, \dots, n + c - 2\}$. On the loss side, rather than bounding $\min\{w/r, 1 - p\}^{n+c-i}$ by $(w/r)(1 - p)^{n+c-1-i}$ as in Subsection~\ref{sec:warmup} (which is too loose when $r$ is small), we split according to whether $w/r \le 1 - p$ or $w/r > 1 - p$. Viewing $w \sim f_{i, n}$ as the $i$-th smallest of $n$ i.i.d.\ $\mathrm{Uniform}(0, 1)$ draws, conditioning on the subset of draws in $[0, r(1 - p)]$ bounds the expectation of $(w/r)^{n+c-i}$ when $w/r \le 1 - p$ by $\frac{i}{n+c}(1 - p)^{n+c-i}$, while $\Pr_{w \sim f_{i, n}}[w/r > 1 - p] \le \frac{i}{n r(1 - p) + 1}$. Summing over $i = 1, \dots, j$ expresses both prefix sums in terms of $\Pr[s_{n+c-1}^{(j)} > x]$ and $\Pr[s_{n+c-2}^{(j)} > x]$.

\begin{lemma}\label{lem:rdom-prefix-bounds}
For any integers $n \ge 1$ and $c \ge 2$, any $r \in (0, 1]$, any $j \in \{1, \dots, n\}$, and any $x \in \R$ with $p := F_S(x) \in (0, 1)$,
\begin{align}
    \sum_{i=1}^j \gain_i(x) &\ge p \Pr\bigl[s_{n+c-1}^{(j)} > x\bigr] + (c - 1) \, p \Pr\bigl[s_{n+c-2}^{(j)} > x\bigr], \label{eq:rdom-gain-sum-lb} \\
    \sum_{i=1}^j \E_{w \sim f_{i, n}}\bigl[\loss_i(x, \min\{w/r, 1 - p\})\bigr] &\le p \Pr\bigl[s_{n+c-1}^{(j)} > x\bigr] + \frac{(n + c - 1)(1 - p)}{n r(1 - p) + 1} \, p \Pr\bigl[s_{n+c-2}^{(j)} > x\bigr]. \label{eq:rdom-loss-sum-ub}
\end{align}
\end{lemma}

Given the lemmas above, we are now ready to prove Theorem~\ref{thm:rdom-upper}.

\begin{proof}[Proof of Theorem~\ref{thm:rdom-upper}]
Fix $r \in (0, 1]$, market sizes $m \ge n \ge 1$, and distributions $F_B \RDOM F_S$ with $\OPT(m, n) < \infty$, and set $c := \lceil 1/r \rceil + 1 \in \Z_{\ge 2}$. For every $x \in \R$ with $p := F_S(x) \in [0, 1]$: if $p = 0$, then $\gain_i(x) = 0$ by~\eqref{eq:rdom-gain-binom} and $\loss_i(x, \cdot) = 0$ by~\eqref{eq:rdom-gain-loss-def}, so $I_i(x) = 0$; if $p = 1$, then $\gain_i(x) = 0$ by~\eqref{eq:rdom-gain-binom} (as $s + 1 - i \ge 1$) and on $\{b^{(i)} \ge x\}$ we have $1 - F_S(b^{(i)}) \le 1 - p = 0$, so again $I_i(x) = 0$. Now suppose $p \in (0, 1)$. Since $c - 1 = \lceil 1/r \rceil \ge 1/r$, we have
\[
    (c - 1)\bigl(n r(1 - p) + 1\bigr) - (n + c - 1)(1 - p) = n\bigl((c - 1)r - 1\bigr)(1 - p) + (c - 1)p > 0,
\]
and hence $\frac{(n + c - 1)(1 - p)}{n r(1 - p) + 1} \le c - 1$. For every $j \in \{1, \dots, n\}$, subtracting~\eqref{eq:rdom-loss-sum-ub} from~\eqref{eq:rdom-gain-sum-lb} in Lemma~\ref{lem:rdom-prefix-bounds} cancels $p \Pr[s_{n+c-1}^{(j)} > x]$ and gives $\sum_{i=1}^j \widetilde{I}_i(x) \ge 0$. Lemma~\ref{lem:rdom-sbp-reduction} therefore yields $\sum_{i=1}^n I_i(x) \ge 0$ for all $x \in \R$, and integrating over $x \in (-\infty, \infty)$ via Lemma~\ref{lem:rdom-layer-cake} establishes $\STR(m, n + \lceil 1/r \rceil + 1) \ge \OPT(m, n)$.
\end{proof}

\bibliographystyle{plainnat}
\bibliography{references}

\appendix

\section{Omitted Proofs from Section~\ref{sec:upper-bound}}\label{app:upper-bound}

\subsection{Omitted Proofs from Subsection~\ref{sec:layer-cake}}\label{app:layer-cake}

In this subsection, we provide the complete proofs of Lemma~\ref{lem:str-indicator}, Lemma~\ref{lem:layer-cake-indicators}, Lemma~\ref{lem:prob-diff}, and Lemma~\ref{lem:cdf-diff-n+2} from Subsection~\ref{sec:layer-cake}.

\begin{replemma}{lem:str-indicator}
For any $m \ge n \ge 1$, any $k \ge n + 1$, and any realization $(\mathbf{b}, \mathbf{s}_k) \in \R^m \times \R^k$,
\begin{equation*}
    \STR(\mathbf{b}, \mathbf{s}_k) \ge \sum_{i=1}^n \bigl(b^{(i)} - s_k^{(i)}\bigr) \I\bigl[b^{(i)} \ge s_k^{(i+1)}\bigr]. \tag{\ref{eq:str-rank-lower-bound}}
\end{equation*}
\end{replemma}

\begin{proof}
Let $r' := r(\mathbf{b}, \mathbf{s}_k) \in \{0, \dots, \min(m, k)\}$ and let $t := t(\mathbf{b}, \mathbf{s}_k) \in \{0, \dots, \min(m, k - 1)\}$ be the number of pairs traded in $\STR(\mathbf{b}, \mathbf{s}_k)$ (Definition~\ref{def:str}; note that if $r' = k$, then $s_k^{(k+1)} = +\infty$ gives $t = k - 1$). For every $i \in \{1, \dots, \min(m, k - 1)\}$:
\begin{itemize}[leftmargin=1.5em, topsep=2pt, itemsep=2pt]
    \item If $i \le r' - 1$, then $i \le t$ and $b^{(i)} \ge b^{(r')} \ge s_k^{(r')} \ge s_k^{(i+1)}$, so $\I[b^{(i)} \ge s_k^{(i+1)}] = 1 = \I[i \le t]$.
    \item If $i = r'$, then $r' \ge 1$ and $t = r' - \I[b^{(r')} < s_k^{(r'+1)}]$ by Definition~\ref{def:str}, so $\I[b^{(r')} \ge s_k^{(r'+1)}] = \I[r' \le t] = \I[i \le t]$.
    \item If $i > r'$, then $i > t$ and by definition of $r'$, $b^{(i)} < s_k^{(i)} \le s_k^{(i+1)}$, so $\I[b^{(i)} \ge s_k^{(i+1)}] = 0 = \I[i \le t]$.
\end{itemize}
Hence $\STR(\mathbf{b}, \mathbf{s}_k) = \sum_{i=1}^t (b^{(i)} - s_k^{(i)}) = \sum_{i=1}^{\min(m, k-1)} (b^{(i)} - s_k^{(i)}) \I[b^{(i)} \ge s_k^{(i+1)}]$. Since $\min(m, k - 1) \ge n$ (as $m \ge n$ and $k \ge n + 1$) and every summand satisfies $(b^{(i)} - s_k^{(i)})\I[b^{(i)} \ge s_k^{(i+1)}] \ge (s_k^{(i+1)} - s_k^{(i)})\I[b^{(i)} \ge s_k^{(i+1)}] \ge 0$, dropping the non-negative summands for $i > n$ yields~\eqref{eq:str-rank-lower-bound}.
\end{proof}

\begin{replemma}{lem:layer-cake-indicators}
Let $m \ge n \ge i \ge 1$ and $k \ge n + 1$. If $\OPT(m, n) < \infty$, then
\begin{equation*}
    \E\Bigl[ \bigl(b^{(i)} - s_k^{(i)}\bigr) \I\bigl[b^{(i)} \ge s_k^{(i+1)}\bigr] - \bigl(b^{(i)} - s_n^{(i)}\bigr)^+ \Bigr] = \int_{-\infty}^{\infty} \E_{b^{(i)}}\Bigl[ \Delta_{n,k,i}\bigl(x, b^{(i)}\bigr) \I\bigl[b^{(i)} \ge x\bigr] \Bigr] dx, \tag{\ref{eq:layer-cake-diff}}
\end{equation*}
where
\begin{equation*}
    \Delta_{n,k,i}(x, b^{(i)}) := \Pr_{\mathbf{s}_k}\bigl[s_k^{(i)} \le x \textnormal{ and } s_k^{(i+1)} \le b^{(i)}\bigr] - \Pr_{\mathbf{s}_n}\bigl[s_n^{(i)} \le x\bigr]. \tag{\ref{eq:delta-def}}
\end{equation*}
\end{replemma}

\begin{proof}
Notice that (since $\I[b^{(i)} \ge s_k^{(i+1)}] = 1$ implies $b^{(i)} \ge s_k^{(i+1)} \ge s_k^{(i)}$)
\[
    \bigl(b^{(i)} - s_k^{(i)}\bigr) \I\bigl[b^{(i)} \ge s_k^{(i+1)}\bigr] = \int_{-\infty}^{\infty}\I\bigl[s_k^{(i)}\le x\le b^{(i)}\bigr] \cdot \I\bigl[b^{(i)} \ge s_k^{(i+1)}\bigr] \, dx,
\]
\[
    \bigl(b^{(i)} - s_n^{(i)}\bigr)^+ = \int_{-\infty}^{\infty} \I\bigl[s_n^{(i)} \le x\bigr] \cdot \I\bigl[b^{(i)} \ge x\bigr] \, dx.
\]
We take expectations of the above equations over $(\mathbf{b}, \mathbf{s}_n, \mathbf{s}_k)$ (with $\mathbf{b} \sim F_B^m$ independent of $(\mathbf{s}_n, \mathbf{s}_k)$ having marginals $F_S^n$ and $F_S^k$). By Tonelli's theorem and independence, the expectations of the two non-negative left-hand sides equal the integrals over $x \in (-\infty, \infty)$ of $g_{k,i}(x) := \E_{b^{(i)}}[\Pr_{\mathbf{s}_k}[s_k^{(i)} \le x \text{ and } s_k^{(i+1)} \le b^{(i)}]\I[b^{(i)} \ge x]]$ and $h_{n,i}(x) := \E_{b^{(i)}}[\Pr_{\mathbf{s}_n}[s_n^{(i)} \le x]\I[b^{(i)} \ge x]]$, respectively, in $[0, \infty]$. Moreover, since $\OPT(m, n) < \infty$, choosing thresholds $x_1 \le x_0$ with $\Pr[b^{(1)} \ge x_1] \ge 1/2$ and $\Pr[s_n^{(1)} \le x_0] \ge 1/2$ gives $\int_{-\infty}^{x_1} F_S(x)\,dx \le 2\E[(b^{(1)} - s_n^{(1)})^+] \le 2\OPT(m, n) < \infty$ and $\int_{x_0}^{\infty} \Pr[b^{(1)} \ge x]\,dx \le 2\OPT(m, n) < \infty$; bounding $g_{k,i}(x)$ by $\Pr[s_k^{(1)} \le x] \le k F_S(x)$ on $(-\infty, x_1]$, by $1$ on $[x_1, x_0]$, and by $\Pr[b^{(1)} \ge x]$ on $[x_0, \infty)$ shows that both expectations are finite, justifying subtracting the two integrals to obtain~\eqref{eq:layer-cake-diff}.
\end{proof}

\begin{replemma}{lem:prob-diff}
Let $m \ge n \ge i \ge 1$, $k \ge n + 1$, and $x \in \R$, and set $p := F_S(x)$ and $q := F_S(b^{(i)})$. When $b^{(i)}\geq x$,
\begin{equation*}
    \Delta_{n,k,i}(x, b^{(i)}) = \Bigl( \Pr_{\mathbf{s}_k}\bigl[s_k^{(i)} \le x\bigr] - \Pr_{\mathbf{s}_n}\bigl[s_n^{(i)} \le x\bigr] \Bigr) - \binom{k}{i} p^i (1 - q)^{k-i}. \tag{\ref{eq:prob-diff}}
\end{equation*}
\end{replemma}

\begin{proof}
Since $s_k^{(i)} \le s_k^{(i+1)}$, the event $\{s_k^{(i)} \le x\}$ partitions into two disjoint events:
\[
    \bigl\{s_k^{(i)} \le x\bigr\} = \bigl\{s_k^{(i)} \le x \text{ and } s_k^{(i+1)} \le b^{(i)}\bigr\} \;\sqcup\; \bigl\{s_k^{(i)} \le x \text{ and } s_k^{(i+1)} > b^{(i)}\bigr\}.
\]
The second event occurs if and only if \emph{exactly} $i$ of the $k$ i.i.d.\ sellers fall in $(-\infty, x]$ (each with probability $F_S(x) = p$) and the remaining $k - i$ sellers fall in $(b^{(i)}, \infty)$ (each with probability $1 - F_S(b^{(i)}) = 1 - q$). Since $x \le b^{(i)}$, the intervals $(-\infty, x]$ and $(b^{(i)}, \infty)$ are strictly disjoint. Hence $\Pr_{\mathbf{s}_k}[s_k^{(i)} \le x \text{ and } s_k^{(i+1)} > b^{(i)}] = \binom{k}{i} p^i (1 - q)^{k-i}$. Subtracting $\Pr_{\mathbf{s}_n}[s_n^{(i)} \le x]$ yields~\eqref{eq:prob-diff}.
\end{proof}

\begin{replemma}{lem:cdf-diff-n+2}
For any $n \ge i \ge 1$ and threshold $x \in \R$ with $p := F_S(x) \in [0, 1]$,
\begin{equation*}
    \Pr_{\mathbf{s}_{n+2}}\bigl[s_{n+2}^{(i)} \le x\bigr] - \Pr_{\mathbf{s}_n}\bigl[s_n^{(i)} \le x\bigr] = p^i A_{n,i}(p). \tag{\ref{eq:cdf-diff-n+2}}
\end{equation*}
\end{replemma}

\begin{proof}
Let $X_n \sim \mathrm{Binomial}(n, p)$ be the number of sellers among the first $n$ with cost at most $x$, and let $Y_2 \sim \mathrm{Binomial}(2, p)$ be the number among the two additional sellers. Then $\{s_n^{(i)} \le x\} = \{X_n \ge i\}$ and $\{s_{n+2}^{(i)} \le x\} = \{X_n + Y_2 \ge i\}$. The difference event $\{X_n + Y_2 \ge i > X_n\}$ occurs in two disjoint ways:
\begin{itemize}[leftmargin=1.5em, topsep=2pt, itemsep=2pt]
    \item $X_n = i - 1$ and $Y_2 \ge 1$: probability $\binom{n}{i-1} p^{i-1}(1-p)^{n-i+1} \cdot [1 - (1-p)^2] = \binom{n}{i-1} p^i (1-p)^{n-i+1}(2-p)$.
    \item For $i \ge 2$, $X_n = i - 2$ and $Y_2 = 2$: probability $\binom{n}{i-2} p^{i-2}(1-p)^{n-i+2} \cdot p^2 = \binom{n}{i-2} p^i (1-p)^{n-i+2}$ (for $i = 1$, this second case cannot occur and the right-hand side evaluates to $0$ under $\binom{n}{-1} := 0$).
\end{itemize}
Summing these two probabilities gives $p^i A_{n,i}(p)$.
\end{proof}

\subsection{Omitted Proof from Subsection~\ref{sec:warmup}}\label{app:warmup}

In this subsection, we provide the complete proof of Lemma~\ref{lem:warmup-reduction} from Subsection~\ref{sec:warmup}.

\begin{replemma}{lem:warmup-reduction}
For any integers $n \ge i \ge 1$ and any $p \in [0, 1]$, define the truncated integrand on $[0, 1]$ by
\begin{equation*}
    \psi_{n,i}(w; p) := A_{n,i}(p) - \binom{n+2}{i} \min(1 - p, w)^{n+2-i}. \tag{\ref{eq:psi-def}}
\end{equation*}
Then for every $r \in [0, 1]$,
\begin{equation*}
    \int_0^r \psi_{n,i}(w; p) \, f_{i,n}(w) \, dw \ge 0. \tag{\ref{eq:truncated-integral-nonneg}}
\end{equation*}
\end{replemma}

\begin{proof}
Let $g(w) := A_{n,i}(p) - \binom{n+2}{i} w^{n+2-i}$, so that the truncated integrand $\psi_{n,i}(w; p)$ from~\eqref{eq:psi-def} satisfies $\psi_{n,i}(w; p) = g(w)$ for $w \in [0, 1 - p]$ and $\psi_{n,i}(w; p) = g(1 - p)$ for $w \in [1 - p, 1]$. Because $n + 2 - i \ge 2$ and $\binom{n+2}{i} > 0$, the function $w \mapsto \min(1 - p, w)^{n+2-i}$ is non-decreasing on $[0, 1]$, which implies that $w \mapsto \psi_{n,i}(w; p)$ is non-increasing on $[0, 1]$.

Fix any $r \in [0, 1]$, so that $\psi_{n,i}(w; p) \ge \psi_{n,i}(r; p)$ for all $w \in [0, r]$ and $\psi_{n,i}(w; p) \le \psi_{n,i}(r; p)$ for all $w \in [r, 1]$. If $\psi_{n,i}(r; p) \ge 0$, then $\psi_{n,i}(w; p) \ge \psi_{n,i}(r; p) \ge 0$ for all $w \in [0, r]$, so $\int_0^r \psi_{n,i}(w; p) \, f_{i,n}(w) \, dw \ge 0$ immediately. If $\psi_{n,i}(r; p) < 0$, then $\psi_{n,i}(w; p) \le \psi_{n,i}(r; p) < 0$ for all $w \in [r, 1]$, so $\int_r^1 \psi_{n,i}(w; p) \, f_{i,n}(w) \, dw \le 0$; extending the upper limit of integration from $r$ to $1$ by adding this non-positive tail integral yields
\begin{equation}\label{eq:warmup-extension}
    \int_0^r \psi_{n,i}(w; p) \, f_{i,n}(w) \, dw \ge \int_0^1 \psi_{n,i}(w; p) \, f_{i,n}(w) \, dw =: \Phi_{n,i}(p).
\end{equation}
Thus it suffices to show that $\Phi_{n,i}(p) \ge 0$.

To that end, let $z := 1 - p \in [0, 1]$ and $d := n + 2 - i \ge 2$, and let $W \sim \mathrm{Beta}(i, n - i + 1)$ be a random variable with density $f_{i,n}(w) = n \binom{n-1}{i-1} w^{i-1}(1 - w)^{n-i}$ on $(0, 1)$. By linearity of expectation and $\int_0^1 f_{i,n}(w) \, dw = 1$, we have
\[
    \Phi_{n,i}(p) = A_{n,i}(p) - \binom{n+2}{i} \E\bigl[\min(z, W)^d\bigr],
\]
where $\E[W] = \frac{i}{n+1}$. Substituting $1 - p = z$ and $2 - p = 1 + z$ (with $n - i + 1 = d - 1$ and $n - i + 2 = d$) into~\eqref{eq:A-def} and applying Pascal's identity $\binom{n}{i-1} + \binom{n}{i-2} = \binom{n+1}{i-1}$ gives
\[
    A_{n,i}(1 - z) = \binom{n}{i-1} z^{d-1} (1 + z) + \binom{n}{i-2} z^d = \binom{n+1}{i-1} z^d + \binom{n}{i-1} z^{d-1}.
\]
Using the identities
\begin{align*}
    \binom{n+1}{i-1} &= \frac{i}{n + 2} \binom{n+2}{i}, \\
    \binom{n}{i-1} &= \frac{n + 2 - i}{n + 2} \cdot \frac{i}{n + 1} \binom{n+2}{i} = \frac{d}{n + 2} \, \E[W] \binom{n+2}{i},
\end{align*}
we can express $A_{n,i}(1 - z)$ as
\begin{equation}\label{eq:convex-comb-identity}
    A_{n,i}(1 - z) = \binom{n+2}{i} \left[ \frac{i}{n + 2} \, z^d \;+\; \frac{d}{n + 2} \, \E[W] \, z^{d-1} \right],
\end{equation}
where the positive weights $\frac{i}{n+2} > 0$ and $\frac{d}{n+2} = \frac{n+2-i}{n+2} > 0$ sum to $1$. Since $d \ge 2$ implies $d - 1 \ge 1$, for all $z, w \ge 0$ we have $\min(z, w)^d \le z^d$ and $\min(z, w)^d \le w \cdot z^{d-1}$. Taking expectations over $W \sim \mathrm{Beta}(i, n - i + 1)$ gives $\E[\min(z, W)^d] \le z^d$ and $\E[\min(z, W)^d] \le \E[W] \, z^{d-1}$. Averaging these two inequalities with the weights $\frac{i}{n+2}$ and $\frac{d}{n+2}$ yields
\[
    \E\bigl[\min(z, W)^d\bigr] = \frac{i}{n + 2} \E\bigl[\min(z, W)^d\bigr] + \frac{d}{n + 2} \E\bigl[\min(z, W)^d\bigr] \le \frac{i}{n + 2} \, z^d + \frac{d}{n + 2} \, \E[W] \, z^{d-1}.
\]
Multiplying both sides by $\binom{n+2}{i}$ and comparing with~\eqref{eq:convex-comb-identity} gives $A_{n,i}(1 - z) \ge \binom{n+2}{i}\E[\min(z, W)^d]$, so $\Phi_{n,i}(p) \ge 0$. Combining this with~\eqref{eq:warmup-extension} completes the proof of Lemma~\ref{lem:warmup-reduction}.
\end{proof}

\subsection{Omitted Proofs from Subsection~\ref{sec:generalization}}\label{app:generalization}

In this subsection, we provide the complete proofs of Lemmas~\ref{lem:fsd-cdf-bound} and~\ref{lem:mlrp-reduction} from Subsection~\ref{sec:generalization}.

\subsubsection{Proof of Lemma~\ref{lem:fsd-cdf-bound}}

When $F_B$ is an arbitrary CDF (possibly with jump discontinuities or flat regions), we couple buyer valuations using the generalized inverse CDF $F_B^{-1}(u) = \inf\{z \in \R : F_B(z) \ge u\}$ from Definition~\ref{def:setting}.

\begin{lemma}\label{lem:quantile-order-stats}
Let $U_1, \dots, U_m \stackrel{\mathrm{i.i.d.}}{\sim} \mathrm{Uniform}(0, 1)$ with sorted order statistics $U^{(1)} \ge \dots \ge U^{(m)}$, and let $W_{i,m} := 1 - U^{(i)} \sim \mathrm{Beta}(i, m - i + 1)$ have density $f_{i,m}(w)$. Then:
\begin{enumerate}[label=(\roman*), leftmargin=2em, topsep=2pt, itemsep=2pt]
    \item The vector $\bigl(F_B^{-1}(U^{(1)}), \dots, F_B^{-1}(U^{(m)})\bigr)$ has the exact joint distribution of the sorted buyer valuations $\bigl(b^{(1)}, \dots, b^{(m)}\bigr)$.
    \item For any threshold $x \in \R$, let $q_x := F_B(x^-) \in [0, 1]$. Then for all $u \in (0, 1)$,
    \begin{equation}\label{eq:threshold-sandwich}
        u > q_x \implies F_B^{-1}(u) \ge x \implies u \ge q_x.
    \end{equation}
    Consequently, since $\Pr[U^{(i)} = q_x] = 0$, the active buyer indicator satisfies $\I[F_B^{-1}(U^{(i)}) \ge x] = \I[W_{i,m} < 1 - q_x]$ almost surely.
\end{enumerate}
\end{lemma}

\begin{proof}
By standard inverse-transform sampling, $F_B^{-1}(U_1), \dots, F_B^{-1}(U_m)$ are i.i.d.\ draws from $F_B$. Since $u \mapsto F_B^{-1}(u)$ is non-decreasing, applying $F_B^{-1}(\cdot)$ to $U^{(1)} \ge \dots \ge U^{(m)}$ preserves their order and yields the order statistics $b^{(1)} \ge \dots \ge b^{(m)}$, proving~(i). For~(ii), if $u > q_x = \lim_{z \uparrow x} F_B(z)$, then $F_B(z) < u$ for all $z < x$, so $\inf\{z : F_B(z) \ge u\} \ge x$, i.e., $F_B^{-1}(u) \ge x$. Conversely, if $u < q_x$, there exists $z < x$ with $F_B(z) \ge u$, so $F_B^{-1}(u) \le z < x$. Thus $\{u : F_B^{-1}(u) \ge x\}$ differs from $(q_x, 1)$ by at most the singleton $\{q_x\}$, which has probability zero under the continuous Beta distribution of $U^{(i)}$.
\end{proof}

\begin{replemma}{lem:fsd-cdf-bound}
Suppose $m \ge n \ge 1$ and $F_B \FSD F_S$. Then for every rank $i \in \{1, \dots, n\}$ and every threshold $x \in \R$ with $p := F_S(x)$ and $q_x := F_B(x^-)$, the quantity $I_i(x)$ from~\eqref{eq:I-i-def} satisfies
\begin{equation*}
    I_i(x) \ge p^i \int_0^{1-q_x} \psi_{n,i}(w; p) \, f_{i,m}(w) \, dw, \tag{\ref{eq:general-integral-lower-bound}}
\end{equation*}
where $\psi_{n,i}(w; p) = A_{n,i}(p) - \binom{n+2}{i} \min(1 - p, w)^{n+2-i}$ is the truncated integrand from~\eqref{eq:psi-def} and $f_{i,m}(w) := m \binom{m-1}{i-1} w^{i-1}(1 - w)^{m-i}$ is the $\mathrm{Beta}(i, m - i + 1)$ density on $(0, 1)$.
\end{replemma}

\begin{proof}
Fix any threshold $x \in \R$, and let $p := F_S(x)$ and $q_x := F_B(x^-)$. For any quantile $u \in (0, 1)$ satisfying $F_B^{-1}(u) \ge x$, right-continuity of $F_B$ and stochastic dominance give $F_S(F_B^{-1}(u)) \ge F_B(F_B^{-1}(u)) \ge u$, which implies $1 - F_S(F_B^{-1}(u)) \le 1 - u$. Furthermore, since $F_B^{-1}(u) \ge x$ and $F_S$ is non-decreasing, we have $F_S(F_B^{-1}(u)) \ge F_S(x) = p$, which gives $1 - F_S(F_B^{-1}(u)) \le 1 - p$. Combining these two inequalities yields
\begin{equation}\label{eq:fsd-survival-bound}
    1 - F_S\bigl(F_B^{-1}(u)\bigr) \le \min(1 - p, 1 - u).
\end{equation}
Since the exponent $n + 2 - i \ge 2$ is positive, raising both sides of~\eqref{eq:fsd-survival-bound} to the power $n + 2 - i$ gives $\bigl(1 - F_S(F_B^{-1}(u))\bigr)^{n+2-i} \le \min(1 - p, 1 - u)^{n+2-i}$. Substituting this bound along with $b^{(i)} = F_B^{-1}(U^{(i)})$ and $w = 1 - U^{(i)}$ into~\eqref{eq:I-i-def},
\[
    I_i(x) = p^i \, \E_{b^{(i)}}\left[ \left( A_{n,i}(p) - \binom{n+2}{i} \bigl(1 - F_S(b^{(i)})\bigr)^{n+2-i} \right) \I\bigl[b^{(i)} \ge x\bigr] \right],
\]
and applying Lemma~\ref{lem:quantile-order-stats} ($\I[F_B^{-1}(U^{(i)}) \ge x] = \I[W_{i,m} < 1 - q_x]$ a.s.) yields~\eqref{eq:general-integral-lower-bound}.
\end{proof}

\subsubsection{Proof of Lemma~\ref{lem:mlrp-reduction}}

\begin{replemma}{lem:mlrp-reduction}
Let $m \ge n \ge i \ge 1$ and $p \in [0, 1]$. Then for every $r \in [0, 1]$,
\begin{equation*}
    \int_0^r \psi_{n,i}(w; p) \, f_{i,m}(w) \, dw \ge 0. \tag{\ref{eq:general-integral-nonneg}}
\end{equation*}
\end{replemma}

\begin{proof}
The inequality holds trivially when $r = 0$. 
Now fix any $r \in (0, 1]$, and let $P_k(r) := \int_0^r f_{i,k}(w) \, dw > 0$ for $k \in \{m, n\}$. Since $\int_0^r \psi_{n,i}(w; p) \, f_{i,n}(w) \, dw \ge 0$ by Lemma~\ref{lem:warmup-reduction}, to establish~\eqref{eq:general-integral-nonneg} it suffices to prove the conditional expectation inequality
\begin{equation}\label{eq:mlrp-conditional}
    \frac{\int_0^r \psi_{n,i}(w; p) \, f_{i,m}(w) \, dw}{\int_0^r f_{i,m}(w) \, dw} \ge \frac{\int_0^r \psi_{n,i}(w; p) \, f_{i,n}(w) \, dw}{\int_0^r f_{i,n}(w) \, dw}.
\end{equation}
Clearing denominators in~\eqref{eq:mlrp-conditional} by multiplying both sides by $P_m(r) P_n(r) > 0$, we must show that
\begin{align*}
    \Delta &:= \int_0^r \int_0^r \psi_{n,i}(u; p) \, f_{i,m}(u) f_{i,n}(v) \, du \, dv - \int_0^r \int_0^r \psi_{n,i}(v; p) \, f_{i,m}(u) f_{i,n}(v) \, du \, dv \\
    &= \int_0^r \int_0^r \bigl(\psi_{n,i}(u; p) - \psi_{n,i}(v; p)\bigr) f_{i,m}(u) f_{i,n}(v) \, du \, dv \ge 0.
\end{align*}
Since the domain of integration $[0, r] \times [0, r]$ is symmetric in $u$ and $v$, relabeling the dummy variables $(u, v) \leftrightarrow (v, u)$ gives an equivalent expression for $\Delta$:
\begin{align*}
    \Delta &= \int_0^r \int_0^r \bigl(\psi_{n,i}(v; p) - \psi_{n,i}(u; p)\bigr) f_{i,m}(v) f_{i,n}(u) \, du \, dv \\
    &= -\int_0^r \int_0^r \bigl(\psi_{n,i}(u; p) - \psi_{n,i}(v; p)\bigr) f_{i,m}(v) f_{i,n}(u) \, du \, dv.
\end{align*}
Averaging these two expressions for $\Delta$ and factoring out $\bigl(\psi_{n,i}(u; p) - \psi_{n,i}(v; p)\bigr)$ yields
\begin{equation}\label{eq:double-integral-sym}
    \Delta = \frac{1}{2} \int_0^r \int_0^r \bigl(\psi_{n,i}(u; p) - \psi_{n,i}(v; p)\bigr) \bigl(f_{i,m}(u) f_{i,n}(v) - f_{i,m}(v) f_{i,n}(u)\bigr) \, du \, dv.
\end{equation}
We now check that the two factors in the integrand of~\eqref{eq:double-integral-sym} always have the same sign:
\begin{itemize}[leftmargin=1.5em, topsep=2pt, itemsep=2pt]
    \item If $u \le v$, then since $w \mapsto \psi_{n,i}(w; p) = A_{n,i}(p) - \binom{n+2}{i} \min(1 - p, w)^{n+2-i}$ is non-increasing, the first factor satisfies $\psi_{n,i}(u; p) - \psi_{n,i}(v; p) \ge 0$. Moreover, since $f_{i,m}(w) = C \cdot (1 - w)^{m-n} f_{i,n}(w)$ with $C = m\binom{m-1}{i-1} / (n\binom{n-1}{i-1}) > 0$ and $m \ge n$, we have $(1 - u)^{m-n} \ge (1 - v)^{m-n}$, so the second factor also satisfies
    \[
        f_{i,m}(u) f_{i,n}(v) - f_{i,m}(v) f_{i,n}(u) = C \, f_{i,n}(u) f_{i,n}(v) \bigl[(1 - u)^{m-n} - (1 - v)^{m-n}\bigr] \ge 0.
    \]
    \item If $u \ge v$, the same monotonicity relations give $\psi_{n,i}(u; p) - \psi_{n,i}(v; p) \le 0$ and $f_{i,m}(u) f_{i,n}(v) - f_{i,m}(v) f_{i,n}(u) \le 0$.
\end{itemize}
In both cases, the product of the two factors is non-negative for all $(u, v) \in (0, r)^2$, and therefore $\Delta \ge 0$.
\end{proof}

\section{Omitted Proofs from Section~\ref{sec:lower-bound}}\label{app:lower-bound}

In this appendix, we record the standard symmetrization reduction for i.i.d.\ markets (used in Section~\ref{sec:lower-bound} and Appendix~\ref{app:rdom-lower}) and provide the proofs of Lemmas~\ref{lem:sim-opt-bound} and~\ref{lem:sim-gft-bound}.

\begin{lemma}\label{lem:seller-sym-reduction}
Let $m, k \in \Z_{\ge 1}$, and let $\mathfrak{S}_m$ and $\mathfrak{S}_k$ denote the symmetric groups on $\{1, \dots, m\}$ and $\{1, \dots, k\}$, respectively. Given any DSIC, IR, and WBB mechanism $M$ for $m$ buyers and $k$ sellers, let $\bar{M}$ be the mechanism that draws independent uniform permutations $(\pi, \sigma) \sim \mathrm{Uniform}(\mathfrak{S}_m \times \mathfrak{S}_k)$, runs $M$ on the permuted profile $(\mathbf{b}^\pi, \mathbf{s}_k^\sigma)$ where $\mathbf{b}^\pi := (b_{\pi(1)}, \dots, b_{\pi(m)})$ and $\mathbf{s}_k^\sigma := (s_{\sigma(1)}, \dots, s_{\sigma(k)})$, and assigns buyer $j$ (resp.\ seller $\ell$) the allocation and payment assigned by $M$ to buyer slot $\pi^{-1}(j)$ (resp.\ seller slot $\sigma^{-1}(\ell)$). Then $\bar{M}$ is DSIC, IR, WBB, and buyer-and-seller symmetric---satisfying $\bar{x}_j^B(\mathbf{b}^\pi, \mathbf{s}_k^\sigma) = \bar{x}_{\pi(j)}^B(\mathbf{b}, \mathbf{s}_k)$ and $\bar{x}_\ell^S(\mathbf{b}^\pi, \mathbf{s}_k^\sigma) = \bar{x}_{\sigma(\ell)}^S(\mathbf{b}, \mathbf{s}_k)$ for all $(\pi, \sigma) \in \mathfrak{S}_m \times \mathfrak{S}_k$---with $\mathrm{GFT}_{\bar{M}}(m, k; F_B, F_S) = \mathrm{GFT}_M(m, k; F_B, F_S)$ for all distributions $F_B, F_S$.
\end{lemma}

\begin{proof}
By construction, $\bar{M}$'s expected allocation rules are
\[
    \bar{x}_j^B(\mathbf{b}, \mathbf{s}_k) = \frac{1}{m! \, k!}\sum_{\pi \in \mathfrak{S}_m}\sum_{\sigma \in \mathfrak{S}_k} x_{\pi^{-1}(j)}^B(\mathbf{b}^\pi, \mathbf{s}_k^\sigma), \qquad
    \bar{x}_\ell^S(\mathbf{b}, \mathbf{s}_k) = \frac{1}{m! \, k!}\sum_{\pi \in \mathfrak{S}_m}\sum_{\sigma \in \mathfrak{S}_k} x_{\sigma^{-1}(\ell)}^S(\mathbf{b}^\pi, \mathbf{s}_k^\sigma),
\]
with analogous expected payments $\bar{p}_j^B, \bar{p}_\ell^S$. For each fixed $(\pi, \sigma) \in \mathfrak{S}_m \times \mathfrak{S}_k$, permuting agent labels before and after running $M$ preserves feasibility ($\sum_{j=1}^m x_{\pi^{-1}(j)}^B(\mathbf{b}^\pi, \mathbf{s}_k^\sigma) = \sum_{\ell=1}^k x_{\sigma^{-1}(\ell)}^S(\mathbf{b}^\pi, \mathbf{s}_k^\sigma)$), DSIC, IR, and WBB; averaging over $(\pi, \sigma) \sim \mathrm{Uniform}(\mathfrak{S}_m \times \mathfrak{S}_k)$ preserves all four properties and makes $\bar{M}$ invariant under relabeling buyers and sellers. Finally, reindexing $j' = \pi^{-1}(j)$ (so $b_j = (\mathbf{b}^\pi)_{j'}$) and $\ell' = \sigma^{-1}(\ell)$ (so $s_\ell = (\mathbf{s}_k^\sigma)_{\ell'}$) gives
\[
    \mathrm{gft}_{\bar{M}}(\mathbf{b}, \mathbf{s}_k) = \frac{1}{m! \, k!} \sum_{\pi \in \mathfrak{S}_m} \sum_{\sigma \in \mathfrak{S}_k} \mathrm{gft}_M(\mathbf{b}^\pi, \mathbf{s}_k^\sigma).
\]
Since $(\mathbf{b}, \mathbf{s}_k) \sim F_B^m \times F_S^k$ has i.i.d.\ coordinates on each side, $(\mathbf{b}^\pi, \mathbf{s}_k^\sigma) \stackrel{d}{=} (\mathbf{b}, \mathbf{s}_k)$ for every $(\pi, \sigma)$, so taking expectations yields $\mathrm{GFT}_{\bar{M}}(m, k; F_B, F_S) = \mathrm{GFT}_M(m, k; F_B, F_S)$.
\end{proof}

\begin{replemma}{lem:sim-opt-bound}
For any $p, \rho, \tau \in (0, 1)$,
\begin{equation*}
    \OPT\bigl(1, 1; F_B^{(\tau)}, F_S^{(\tau)}\bigr) = p(1 - p\rho)(1 - \tau) \ge p(1 - \rho)(1 - \tau) > 0. \tag{\ref{eq:sim-opt-exact}}
\end{equation*}
\end{replemma}

\begin{proof}
In $\OPT(1, 1; F_B^{(\tau)}, F_S^{(\tau)})$, positive GFT occurs only when the single seller has cost $\tau$ (probability $p$; when $s_1 = 1$, $b \le 1$ gives $(b - s_1)^+ = 0$). Conditional on $s_1 = \tau$, a buyer $b \sim F_B^{(\tau)}$ has valuation $1$ with probability $1 - p$ (surplus $1 - \tau$) and valuation $b_\tau = 1 - \rho(1 - \tau)$ with probability $p$ (surplus $b_\tau - \tau = (1 - \rho)(1 - \tau)$), yielding
\[
    \OPT\bigl(1, 1; F_B^{(\tau)}, F_S^{(\tau)}\bigr) = p\bigl[(1 - p)(1 - \tau) + p(1 - \rho)(1 - \tau)\bigr] = p(1 - p\rho)(1 - \tau) \ge p(1 - \rho)(1 - \tau) > 0. \qedhere
\]
\end{proof}

\begin{replemma}{lem:sim-gft-bound}
Let $M$ be any seller-symmetric DSIC, IR, and WBB mechanism satisfying~\eqref{eq:feasibility}. For any $p, \rho, \tau \in (0, 1)$,
\begin{equation*}
    \mathrm{GFT}_M\bigl(1, 2; F_B^{(\tau)}, F_S^{(\tau)}\bigr) \le p^2(1 - \tau) + 2p(1 - p)(1 - \tau)\Bigl[(1 - p) + p \, x_1\bigl(b_\tau, \tau, 1\bigr)\Bigr]. \tag{\ref{eq:sim-gft-upper}}
\end{equation*}
\end{replemma}

\begin{proof}
Depending on the realization of $(s_1, s_2) \sim (F_S^{(\tau)})^2$, there are three cases:
\begin{itemize}[leftmargin=1.5em, topsep=2pt, itemsep=2pt]
    \item \emph{Both sellers have cost $1$ ($s_1 = s_2 = 1$, probability $(1 - p)^2$):} Since $b \le 1$ and $s_1 = s_2 = 1$,~\eqref{eq:gft-profile} gives $\mathrm{gft}_M(b, 1, 1) \le 0$.
    \item \emph{Both sellers have cost $\tau$ ($s_1 = s_2 = \tau$, probability $p^2$):} Since $b \le 1$ and $x_B(b, \tau, \tau) = x_1(b, \tau, \tau) + x_2(b, \tau, \tau) \le 1$,~\eqref{eq:gft-profile} gives $\mathrm{gft}_M(b, \tau, \tau) = (b - \tau) x_B(b, \tau, \tau) \le 1 - \tau$, so this case contributes at most $p^2(1 - \tau)$.
    \item \emph{One seller has cost $\tau$ and one has cost $1$ ($(\tau, 1)$ or $(1, \tau)$, total probability $2p(1 - p)$):} By seller symmetry, $\E_b[\mathrm{gft}_M(b, 1, \tau)] = \E_b[\mathrm{gft}_M(b, \tau, 1)]$. Under $b \sim F_B^{(\tau)}$, with probability $1 - p$ we have $b = 1$ and~\eqref{eq:mixed-gft-ub} gives $\mathrm{gft}_M(1, \tau, 1) = (1 - \tau)x_1(1, \tau, 1) \le 1 - \tau$, while with probability $p$ we have $b = b_\tau$ and~\eqref{eq:mixed-gft-ub} gives $\mathrm{gft}_M(b_\tau, \tau, 1) \le (b_\tau - \tau) x_1(b_\tau, \tau, 1) \le (1 - \tau) x_1(b_\tau, \tau, 1)$. Thus $\E_{b \sim F_B^{(\tau)}}[\mathrm{gft}_M(b, \tau, 1)] \le (1 - p)(1 - \tau) + p(1 - \tau) x_1(b_\tau, \tau, 1)$.
\end{itemize}
Summing the three cases weighted by $(1 - p)^2$, $p^2$, and $2p(1 - p)$ establishes~\eqref{eq:sim-gft-upper}.
\end{proof}

\section{Omitted Proofs for \texorpdfstring{$r$}{r}-Dominance}\label{app:r-dominance}

\subsection{Omitted Proofs from Section~\ref{sec:r-dominance}}\label{app:rdom-Pj-proof}\label{subsec:rdom-Pj-proof}

In this subsection, we provide the complete proofs of Lemmas~\ref{lem:rdom-layer-cake}, \ref{lem:rdom-sbp-reduction}, and~\ref{lem:rdom-prefix-bounds} from Section~\ref{sec:r-dominance}, recalling the definitions of $\gain_i(x)$, $\loss_i(x, w)$, $I_i(x)$, and $\widetilde{I}_i(x)$ from~\eqref{eq:rdom-gain-loss-def}, \eqref{eq:rdom-Ii-def}, and~\eqref{eq:rdom-Itilde-def}.

\begin{replemma}{lem:rdom-layer-cake}
Let $m \ge n \ge 1$, $c \in \Z_{\ge 2}$, and let $F_B, F_S$ be any distributions on $\R$ with $\OPT(m, n) < \infty$. Then
\begin{equation*}
    \STR(m, n + c) - \OPT(m, n) \ge \sum_{i=1}^n \int_{-\infty}^{\infty} I_i(x) \, dx, \tag{\ref{eq:rdom-surplus-diff-sum}}
\end{equation*}
and for every $i \in \{1, \dots, n\}$ and $x \in \R$ with $p := F_S(x) \in [0, 1]$,
\begin{equation*}
    \gain_i(x) = \sum_{s=n}^{n+c-1} \binom{s}{i-1} p^i (1 - p)^{s+1-i}. \tag{\ref{eq:rdom-gain-binom}}
\end{equation*}
Moreover, if $F_B \RDOM F_S$ ($r \in (0, 1]$), then for every $i \in \{1, \dots, n\}$ and $x \in \R$ with $p := F_S(x)$ and $R := 1 - F_B(x^-)$,
\begin{equation*}
    I_i(x) \ge \int_0^R \Bigl(\gain_i(x) - \loss_i\bigl(x, \min\{w/r, 1 - p\}\bigr)\Bigr) f_{i, m}(w) \, dw, \tag{\ref{eq:rdom-Ii-quantile}}
\end{equation*}
where $f_{i, m}(w) := m \binom{m-1}{i-1} w^{i-1}(1 - w)^{m-i}$ is the $\mathrm{Beta}(i, m - i + 1)$ density on $(0, 1)$.
\end{replemma}

\begin{proof}
Set $k := n + c \ge n + 2 \ge n + 1$. Taking expectations in Lemma~\ref{lem:str-indicator}, subtracting $\OPT(m, n) = \sum_{i=1}^n \E[(b^{(i)} - s_n^{(i)})^+]$, and substituting Lemmas~\ref{lem:layer-cake-indicators} and~\ref{lem:prob-diff} from Subsection~\ref{sec:layer-cake} (together with $\binom{n+c}{i} p^i (1 - F_S(b^{(i)}))^{n+c-i} = \loss_i(x, 1 - F_S(b^{(i)}))$ from~\eqref{eq:rdom-gain-loss-def}) yields~\eqref{eq:rdom-surplus-diff-sum}.

To verify~\eqref{eq:rdom-gain-binom}, draw $\mathbf{s}_{n+c} = (s_1, \dots, s_{n+c}) \sim F_S^{n+c}$ and, for $s \in \{n, \dots, n + c\}$, define $N_s := \sum_{\ell=1}^s \I[s_\ell \le x] \sim \mathrm{Binomial}(s, p)$, where $\Pr[s_\ell \le x] = F_S(x) = p$ by right-continuity of $F_S$. Since $\{s_s^{(i)} \le x\} = \{N_s \ge i\}$ and $1 \le i \le n \le s$, telescoping over $s = n, \dots, n + c - 1$ gives
\[
\begin{aligned}
    \gain_i(x) &= \sum_{s=n}^{n+c-1} \Bigl(\Pr[N_{s+1} \ge i] - \Pr[N_s \ge i]\Bigr) \\
    &= \sum_{s=n}^{n+c-1} \Pr\bigl[N_s = i - 1 \text{ and } s_{s+1} \le x\bigr] = \sum_{s=n}^{n+c-1} \binom{s}{i-1} p^i (1 - p)^{s+1-i},
\end{aligned}
\]
where both exponents $i \ge 1$ and $s + 1 - i \ge 1$ are positive integers.

Finally, suppose $F_B \RDOM F_S$ ($r \in (0, 1]$), fix $x \in \R$, and set $p := F_S(x)$, $q_x := F_B(x^-)$, and $R := 1 - q_x \in [0, 1]$. Let $U_1, \dots, U_m \stackrel{\mathrm{i.i.d.}}{\sim} \mathrm{Uniform}(0, 1)$ with order statistics $U^{(1)} \ge \dots \ge U^{(m)}$ and $W_{i, m} := 1 - U^{(i)} \sim \mathrm{Beta}(i, m - i + 1)$ as in Lemma~\ref{lem:quantile-order-stats} (Appendix~\ref{app:generalization}). For any $u \in (0, 1)$ satisfying $F_B^{-1}(u) \ge x$, Definition~\ref{def:setting} gives $F_B(F_B^{-1}(u)) \ge u$, so applying $F_B \RDOM F_S$ (Definition~\ref{def:r-dominance}) at $F_B^{-1}(u)$ gives
\[
    r \bigl(1 - F_S(F_B^{-1}(u))\bigr) \le 1 - F_B\bigl(F_B^{-1}(u)\bigr) \le 1 - u \qquad \implies \qquad 1 - F_S\bigl(F_B^{-1}(u)\bigr) \le \frac{1 - u}{r}.
\]
Since $F_B^{-1}(u) \ge x$ and $F_S$ is non-decreasing, we also have $1 - F_S(F_B^{-1}(u)) \le 1 - F_S(x) = 1 - p$, so $0 \le 1 - F_S(F_B^{-1}(u)) \le \min\{(1 - u)/r, 1 - p\}$. Since $n + c - i \ge c \ge 2$ and $p^i \ge 0$, $w \mapsto \loss_i(x, w)$ is non-decreasing on $[0, \infty)$, so $\loss_i(x, 1 - F_S(F_B^{-1}(u))) \le \loss_i(x, \min\{(1 - u)/r, 1 - p\})$. Substituting $b^{(i)} = F_B^{-1}(U^{(i)})$, $W_{i, m} = 1 - U^{(i)}$, and $\I[F_B^{-1}(U^{(i)}) \ge x] = \I[W_{i, m} < R]$ a.s.\ (Lemma~\ref{lem:quantile-order-stats}) into~\eqref{eq:rdom-Ii-def} yields~\eqref{eq:rdom-Ii-quantile}.
\end{proof}

\begin{replemma}{lem:rdom-sbp-reduction}
Suppose $r \in (0, 1]$, $m \ge n \ge 1$, $c \in \Z_{\ge 2}$, and $F_B \RDOM F_S$. Then for every $x \in \R$ and $i \in \{1, \dots, n\}$, we have $I_i(x) \ge \Pr[b^{(i)} \ge x] \, \widetilde{I}_i(x)$, and consequently, setting $\Pr[b^{(n+1)} \ge x] := 0$,
\begin{equation*}
    \sum_{i=1}^n I_i(x) \ge \sum_{j=1}^n \Bigl(\Pr\bigl[b^{(j)} \ge x\bigr] - \Pr\bigl[b^{(j+1)} \ge x\bigr]\Bigr) \sum_{i=1}^j \widetilde{I}_i(x), \tag{\ref{eq:rdom-sbp-main}}
\end{equation*}
where $\Pr[b^{(j)} \ge x] - \Pr[b^{(j+1)} \ge x] \ge 0$ for all $j \in \{1, \dots, n\}$.
\end{replemma}

\begin{proof}
Fix $x \in \R$, and set $p := F_S(x) \in [0, 1]$ and $R := 1 - F_B(x^-) \in [0, 1]$. By Lemma~\ref{lem:quantile-order-stats}, $(b^{(1)}, \dots, b^{(m)}) \stackrel{d}{=} (F_B^{-1}(U^{(1)}), \dots, F_B^{-1}(U^{(m)}))$ and $\I[F_B^{-1}(U^{(i)}) \ge x] = \I[W_{i, m} < R]$ almost surely for $W_{i, m} \sim \mathrm{Beta}(i, m - i + 1)$, so taking expectations gives $\Pr[b^{(i)} \ge x] = \int_0^R f_{i, m}(w) \, dw$ for every $i \in \{1, \dots, n\}$.

Fix $i \in \{1, \dots, n\}$ and set $g_i(w) := \gain_i(x) - \loss_i(x, \min\{w/r, 1 - p\})$ on $[0, 1]$. Since $r > 0$, $p^i \ge 0$, and $n + c - i \ge 2$, the function $w \mapsto g_i(w)$ is non-increasing on $[0, 1]$, so
\begin{equation}\label{eq:rdom-psi-monotone}
    g_i(u) \ge g_i(R) \ge g_i(v) \qquad \forall \, 0 \le u \le R \le v \le 1.
\end{equation}
If $R = 0$, then $\Pr[b^{(i)} \ge x] = 0$, so $\I[b^{(i)} \ge x] = 0$ a.s.\ and $I_i(x) = 0 = \Pr[b^{(i)} \ge x] \, \widetilde{I}_i(x)$. Suppose $R \in (0, 1]$, so $\Pr[b^{(i)} \ge x] = \int_0^R f_{i, m}(w) \, dw > 0$ and $P_n(R) := \int_0^R f_{i, n}(w) \, dw > 0$. Since $g_i$ is non-increasing on $[0, R]$ and $f_{i, m}(w) / f_{i, n}(w) \propto (1 - w)^{m-n}$ is well-defined and non-increasing on $(0, R)$ for $m \ge n$, the double-integral calculation in the proof of Lemma~\ref{lem:mlrp-reduction} (Equation~\eqref{eq:double-integral-sym} in Appendix~\ref{app:generalization}, with $g_i$ in place of $\psi_{n, i}$) yields
\[
    \frac{\int_0^R g_i(w) \, f_{i, m}(w) \, dw}{\int_0^R f_{i, m}(w) \, dw} \ge \frac{\int_0^R g_i(w) \, f_{i, n}(w) \, dw}{\int_0^R f_{i, n}(w) \, dw}.
\]
Furthermore, using $\int_0^1 = \int_0^R + \int_R^1$ and $1 - P_n(R) = \int_R^1 f_{i, n}(v) \, dv$, we have
\begin{align*}
    &\int_0^R g_i(u) \, f_{i, n}(u) \, du - P_n(R) \int_0^1 g_i(v) \, f_{i, n}(v) \, dv \\
    &\quad = \bigl(1 - P_n(R)\bigr) \int_0^R g_i(u) \, f_{i, n}(u) \, du - P_n(R) \int_R^1 g_i(v) \, f_{i, n}(v) \, dv \\
    &\quad = \int_0^R \int_R^1 \bigl(g_i(u) - g_i(v)\bigr) f_{i, n}(u) f_{i, n}(v) \, dv \, du \ge 0,
\end{align*}
where the inequality holds by~\eqref{eq:rdom-psi-monotone}. Dividing by $P_n(R) > 0$ gives $\frac{1}{P_n(R)} \int_0^R g_i(w) \, f_{i, n}(w) \, dw \ge \int_0^1 g_i(w) \, f_{i, n}(w) \, dw = \widetilde{I}_i(x)$. Combining these two inequalities with~\eqref{eq:rdom-Ii-quantile} and multiplying by $\Pr[b^{(i)} \ge x] = \int_0^R f_{i, m}(w) \, dw > 0$ establishes $I_i(x) \ge \Pr[b^{(i)} \ge x] \, \widetilde{I}_i(x)$.

Finally, since $b^{(1)} \ge b^{(2)} \ge \dots \ge b^{(n)}$ pointwise, $\{b^{(n)} \ge x\} \subseteq \dots \subseteq \{b^{(1)} \ge x\}$, which gives $\Pr[b^{(1)} \ge x] \ge \Pr[b^{(2)} \ge x] \ge \dots \ge \Pr[b^{(n)} \ge x] \ge \Pr[b^{(n+1)} \ge x] := 0$. Writing $\Pr[b^{(i)} \ge x] = \sum_{j=i}^n \bigl(\Pr[b^{(j)} \ge x] - \Pr[b^{(j+1)} \ge x]\bigr)$, summing $I_i(x) \ge \Pr[b^{(i)} \ge x] \, \widetilde{I}_i(x)$ over $i = 1, \dots, n$, and exchanging the order of summation over $1 \le i \le j \le n$ yields~\eqref{eq:rdom-sbp-main}.
\end{proof}

To prove Lemma~\ref{lem:rdom-prefix-bounds}, we first record an elementary upper bound on binomial lower tails and use it to bound $\E_{w \sim f_{i, n}}[\loss_i(x, \min\{w/r, 1 - p\})]$ for each $i \in \{1, \dots, n\}$.

\begin{lemma}\label{lem:binom-tail-stein}
For any integers $n \ge 1$ and $u \ge 0$ and any probability $p_0 \in (0, 1]$, let $X \sim \mathrm{Binomial}(n, p_0)$. Then
\begin{equation}\label{eq:binom-tail-stein}
    \Pr[X \le u] \le \frac{u + 1}{n p_0 + 1}.
\end{equation}
\end{lemma}

\begin{proof}
Write $X = \sum_{\ell=1}^n \xi_\ell$ with $\xi_1, \dots, \xi_n \stackrel{\mathrm{i.i.d.}}{\sim} \mathrm{Bernoulli}(p_0)$. For each $\ell \in \{1, \dots, n\}$, $X_{-\ell} := X - \xi_\ell$ is independent of $\xi_\ell$; moreover, $\xi_\ell \ge 0$ implies $\{X \le u\} \subseteq \{X_{-\ell} \le u\}$, and $\xi_\ell \in \{0, 1\}$ implies $\xi_\ell \, \I[X_{-\ell} \le u] = \xi_\ell \, \I[X \le u + 1]$ pointwise. Summing over $\ell = 1, \dots, n$ yields
\begin{align*}
    n p_0 \Pr[X \le u]
    &\le \sum_{\ell=1}^n p_0 \Pr[X_{-\ell} \le u] = \sum_{\ell=1}^n \E\bigl[\xi_\ell \, \I[X_{-\ell} \le u]\bigr] = \sum_{\ell=1}^n \E\bigl[\xi_\ell \, \I[X \le u + 1]\bigr] \\
    &= \E\bigl[X \, \I[X \le u + 1]\bigr] = \E\bigl[X \, \I[X \le u]\bigr] + (u + 1)\Pr[X = u + 1] \\
    &\le u \Pr[X \le u] + (u + 1)\Pr[X \ge u + 1].
\end{align*}
Adding $\Pr[X \le u]$ to both sides gives $(n p_0 + 1)\Pr[X \le u] \le (u + 1)\bigl(\Pr[X \le u] + \Pr[X \ge u + 1]\bigr) = u + 1$, and dividing by $n p_0 + 1 > 0$ establishes~\eqref{eq:binom-tail-stein}.
\end{proof}

\begin{lemma}\label{lem:rdom-moment-ub}
For any integers $n \ge 1$ and $c \ge 2$, any $r \in (0, 1]$, any $i \in \{1, \dots, n\}$, and any $x \in \R$ with $p := F_S(x) \in (0, 1)$,
\begin{equation}\label{eq:rdom-convex-comb-bound}
    \E_{w \sim f_{i, n}}\bigl[\loss_i(x, \min\{w/r, 1 - p\})\bigr] \le \binom{n+c}{i} p^i (1 - p)^{n+c-i} \left(\frac{i}{n + c} + \frac{n + c - i}{n + c} \cdot \frac{i}{n r(1 - p) + 1}\right).
\end{equation}
\end{lemma}

\begin{proof}
Set $d := n + c - i \ge c \ge 2$ and $p_0 := r(1 - p) \in (0, 1)$, and let $U_1, \dots, U_n \stackrel{\mathrm{i.i.d.}}{\sim} \mathrm{Uniform}(0, 1)$ with order statistics $U_{(1)} \le \dots \le U_{(n)}$, so $w \sim f_{i, n}$ has the distribution of $U_{(i)}$~\cite{DavidNagaraja03}. Define $S_{r, p} := \{\ell \in \{1, \dots, n\} : U_\ell / r \le 1 - p\} = \{\ell \in \{1, \dots, n\} : U_\ell \le p_0\}$, so that $M_{r, p} := |S_{r, p}| \sim \mathrm{Binomial}(n, p_0)$.
We bound $\E\bigl[\min\{U_{(i)} / r, 1 - p\}^d \mid S_{r, p}\bigr]$ according to $|S_{r, p}|$:
\begin{itemize}[leftmargin=1.5em, topsep=2pt, itemsep=2pt]
    \item If $|S_{r, p}| \le i - 1$ (equivalently, $U_{(i)} / r > 1 - p$), then $\min\{U_{(i)} / r, 1 - p\}^d = (1 - p)^d$.
    \item If $|S_{r, p}| \ge i$ (equivalently, $U_{(i)} / r \le 1 - p$), then conditional on $S_{r, p} = A$ with $|A| \ge i$, the variables $(U_\ell)_{\ell \in A}$ are mutually independent $\mathrm{Uniform}(0, p_0)$ draws, so $Y_\ell := U_\ell / p_0 = \frac{U_\ell}{r(1 - p)}$ for $\ell \in A$ are i.i.d.\ $\mathrm{Uniform}(0, 1)$. Let $A_0 \subseteq A$ consist of the $i$ smallest indices in $A$, and set $\widetilde{W} := \max_{\ell \in A_0} Y_\ell \in [0, 1]$. Since $|A_0| = i$ and $U_\ell / r = (1 - p) Y_\ell \le (1 - p) \widetilde{W}$ for all $\ell \in A_0$, at least $i$ of $U_1 / r, \dots, U_n / r$ are at most $(1 - p)\widetilde{W}$, so $0 \le U_{(i)} / r \le (1 - p)\widetilde{W} \le 1 - p$ and $\min\{U_{(i)} / r, 1 - p\}^d \le (1 - p)^d \widetilde{W}^d$. Because $\widetilde{W}$ has density $i v^{i-1}$ on $(0, 1)$ conditional on $S_{r, p} = A$,
    \[
        \E\bigl[\min\{U_{(i)} / r, 1 - p\}^d \;\big|\; S_{r, p} = A\bigr] \le (1 - p)^d \int_0^1 v^d \cdot i v^{i-1} \, dv = \frac{i}{i + d} \, (1 - p)^d = \frac{i}{n + c} \, (1 - p)^d.
    \]
\end{itemize}
Taking expectations over $S_{r, p}$ and multiplying by $\binom{n+c}{i} p^i$ gives
\[
    \E_{w \sim f_{i, n}}\bigl[\loss_i(x, \min\{w/r, 1 - p\})\bigr] \le \binom{n+c}{i} p^i (1 - p)^{n+c-i} \left(\frac{i}{n + c} + \frac{n + c - i}{n + c} \Pr[M_{r, p} \le i - 1]\right).
\]
Applying Lemma~\ref{lem:binom-tail-stein} to $M_{r, p} \sim \mathrm{Binomial}(n, r(1 - p))$ at $u = i - 1 \ge 0$ gives $\Pr[M_{r, p} \le i - 1] \le \frac{i}{n r(1 - p) + 1}$, establishing~\eqref{eq:rdom-convex-comb-bound}.
\end{proof}

\begin{replemma}{lem:rdom-prefix-bounds}
For any integers $n \ge 1$ and $c \ge 2$, any $r \in (0, 1]$, any $j \in \{1, \dots, n\}$, and any $x \in \R$ with $p := F_S(x) \in (0, 1)$,
\begin{align*}
    \sum_{i=1}^j \gain_i(x) &\ge p \Pr\bigl[s_{n+c-1}^{(j)} > x\bigr] + (c - 1) \, p \Pr\bigl[s_{n+c-2}^{(j)} > x\bigr], \tag{\ref{eq:rdom-gain-sum-lb}} \\
    \sum_{i=1}^j \E_{w \sim f_{i, n}}\bigl[\loss_i(x, \min\{w/r, 1 - p\})\bigr] &\le p \Pr\bigl[s_{n+c-1}^{(j)} > x\bigr] + \frac{(n + c - 1)(1 - p)}{n r(1 - p) + 1} \, p \Pr\bigl[s_{n+c-2}^{(j)} > x\bigr]. \tag{\ref{eq:rdom-loss-sum-ub}}
\end{align*}
\end{replemma}

\begin{proof}
For any integer $s \ge n$ and $\mathbf{s}_s = (s_1, \dots, s_s) \sim F_S^s$, each seller satisfies $s_\ell \le x$ independently with probability $F_S(x) = p \in (0, 1)$, and since $1 \le j \le n \le s$, the event $\{s_s^{(j)} > x\}$ holds if and only if the number of sellers with cost at most $x$ lies in $\{0, 1, \dots, j - 1\}$:
\begin{equation}\label{eq:rdom-seller-tail-binom}
    \Pr\bigl[s_s^{(j)} > x\bigr] = \sum_{i=1}^j \binom{s}{i - 1} p^{i-1}(1 - p)^{s - (i - 1)}.
\end{equation}
Summing~\eqref{eq:rdom-gain-binom} over $i = 1, \dots, j$ and substituting~\eqref{eq:rdom-seller-tail-binom} gives
\[
    \sum_{i=1}^j \gain_i(x) = \sum_{s=n}^{n+c-1} p \sum_{i=1}^j \binom{s}{i - 1} p^{i-1}(1 - p)^{s - (i - 1)} = \sum_{s=n}^{n+c-1} p \Pr\bigl[s_s^{(j)} > x\bigr].
\]
For each of the $c - 1$ indices $s \in \{n, \dots, n + c - 2\}$ (noting $n + c - 2 \ge n \ge j$ since $c \ge 2$), coupling $\mathbf{s}_s$ as the first $s$ coordinates of $\mathbf{s}_{n+c-2} \sim F_S^{n+c-2}$ gives $s_s^{(j)} \ge s_{n+c-2}^{(j)}$ pointwise and hence $\Pr[s_s^{(j)} > x] \ge \Pr[s_{n+c-2}^{(j)} > x]$, proving~\eqref{eq:rdom-gain-sum-lb}.

For~\eqref{eq:rdom-loss-sum-ub}, observe that for every $i \in \{1, \dots, j\}$ (with $j \le n \le n + c - 2$), we have $\binom{n+c}{i}\frac{i}{n+c} = \binom{n+c-1}{i-1}$ and $\binom{n+c}{i}\frac{n+c-i}{n+c} \cdot i = i \binom{n+c-1}{i} = (n + c - 1)\binom{n+c-2}{i-1}$. Summing~\eqref{eq:rdom-convex-comb-bound} of Lemma~\ref{lem:rdom-moment-ub} over $i = 1, \dots, j$ and applying~\eqref{eq:rdom-seller-tail-binom} at $s = n + c - 1$ and $s = n + c - 2$ yields
\begin{align*}
    &\sum_{i=1}^j \E_{w \sim f_{i, n}}\bigl[\loss_i(x, \min\{w/r, 1 - p\})\bigr] \\
    &\quad \le p \sum_{i=1}^j \binom{n+c-1}{i-1} p^{i-1}(1 - p)^{(n+c-1)-(i-1)} \\
    &\qquad + \frac{(n + c - 1)(1 - p)}{n r(1 - p) + 1} \, p \sum_{i=1}^j \binom{n+c-2}{i-1} p^{i-1}(1 - p)^{(n+c-2)-(i-1)} \\
    &\quad = p \Pr\bigl[s_{n+c-1}^{(j)} > x\bigr] + \frac{(n + c - 1)(1 - p)}{n r(1 - p) + 1} \, p \Pr\bigl[s_{n+c-2}^{(j)} > x\bigr]. \qedhere
\end{align*}
\end{proof}

\subsection{Lower Bounds under \texorpdfstring{$r$}{r}-Dominance}\label{app:rdom-lower}\label{subsec:rdom-lower}\label{sec:rdom-lower-bounds}

In this subsection, we adapt the lower bound construction of Section~\ref{sec:lower-bound} from $r = 1$ to $r$-dominance across all prior-independent DSIC, IR, and WBB mechanisms, proving Theorems~\ref{thm:rdom-mech-lb} and~\ref{thm:general-two-sided-mech-lb}.

\paragraph{Smaller-side recruitment (Theorem~\ref{thm:rdom-mech-lb}).}
When $c = 1$, Theorem~\ref{thm:lower-bound} already gives distributions $F_B \FSD F_S$---and hence $F_B \RDOM F_S$---on $[0, 1]$ with $\mathrm{GFT}_M(1, 2) < \OPT(1, 1)$. Assume henceforth that $c \ge 2$ and $c < \max\bigl(2, \lceil 1/r \rceil - 1\bigr)$; since $c \ge 2$ is an integer, this implies $c + 1 \le \lceil 1/r \rceil - 1 < 1/r$ and hence $(c + 1)r < 1$. As in Section~\ref{sec:lower-bound}, we may assume without loss of generality that $M$ is restricted to $[0, 1] \times [0, 1]^{1+c}$ and is symmetric across the $1 + c$ sellers (see Lemma~\ref{lem:seller-sym-reduction} in Appendix~\ref{app:lower-bound}).

Parallel to Section~\ref{sec:lower-bound}, we first consider the mixed seller profile $(\tau, \mathbf{1}_c) := (\tau, 1, \dots, 1) \in [0, 1]^{1+c}$, where the $c$ sellers at cost $1$ cannot profitably trade for $b \le 1$. Thus the proof of Lemma~\ref{lem:sim-inf-zero} again forces Seller~1's trade probability at $b_\tau := 1 - \rho(1 - \tau)$ to have infimum $0$ over $\tau \in (0, 1)$.

\begin{lemma}\label{lem:rdom-smaller-inf-zero}
Let $c \in \Z_{\ge 1}$ and let $M$ be any DSIC, IR, and WBB mechanism for $1$ buyer and $1 + c$ sellers on $[0, 1] \times [0, 1]^{1+c}$. For every $\rho \in (0, 1)$, letting $b_\tau := 1 - \rho(1 - \tau)$,
\begin{equation}\label{eq:rdom-smaller-inf-zero}
    \inf_{\tau \in (0, 1)} x_1^S\bigl(b_\tau, (\tau, \mathbf{1}_c)\bigr) = 0.
\end{equation}
\end{lemma}

\begin{proof}
Set $k := 1 + c \ge 2$. On $[0, 1] \times [0, 1]^k$, Subsection~\ref{subsec:mechanisms} gives $0 \le x_1^B(b, \mathbf{s}_k) = \sum_{\ell=1}^k x_\ell^S(b, \mathbf{s}_k)$ with $x_1^B, x_\ell^S \in [0, 1]$, and $\mathrm{gft}_M(b, \mathbf{s}_k) = b \, x_1^B(b, \mathbf{s}_k) - \sum_{\ell=1}^k s_\ell \, x_\ell^S(b, \mathbf{s}_k) = \sum_{\ell=1}^k (b - s_\ell) x_\ell^S(b, \mathbf{s}_k)$. Letting $u_1^B := b \, x_1^B - p_1^B \ge 0$ and $u_\ell^S := p_\ell^S - s_\ell \, x_\ell^S \ge 0$ for $\ell \in \{1, \dots, k\}$ denote the agents' expected utilities under IR, WBB ($p_1^B \ge \sum_{\ell=1}^k p_\ell^S$) implies $u_1^B(b, \mathbf{s}_k) + \sum_{\ell=1}^k u_\ell^S(b, \mathbf{s}_k) \le \mathrm{gft}_M(b, \mathbf{s}_k)$.

Fix any $0 < \tau < b \le 1$ and consider $\mathbf{s}_k = (\tau, \mathbf{1}_c)$. Since $s_\ell = 1 \ge b$ and $x_\ell^S \ge 0$ for all $\ell \ge 2$, we have
\begin{equation}\label{eq:rdom-smaller-profile-ub}
    \mathrm{gft}_M\bigl(b, (\tau, \mathbf{1}_c)\bigr) = (b - \tau) x_1^S\bigl(b, (\tau, \mathbf{1}_c)\bigr) - \sum_{\ell=2}^k (1 - b) x_\ell^S\bigl(b, (\tau, \mathbf{1}_c)\bigr) \le (b - \tau) x_1^S\bigl(b, (\tau, \mathbf{1}_c)\bigr).
\end{equation}
Moreover, since $u_1^B(\tau, (\tau, \mathbf{1}_c)) \ge 0$, $u_1^S(b, (b, \mathbf{1}_c)) \ge 0$, and $u_\ell^S(b, (\tau, \mathbf{1}_c)) \ge 0$ for $\ell \ge 2$ by IR, DSIC and Myerson's payment identity give
\[
    \int_\tau^b x_1^B\bigl(t, (\tau, \mathbf{1}_c)\bigr) \, dt + \int_\tau^b x_1^S\bigl(b, (u, \mathbf{1}_c)\bigr) \, du \le u_1^B\bigl(b, (\tau, \mathbf{1}_c)\bigr) + u_1^S\bigl(b, (\tau, \mathbf{1}_c)\bigr) \le (b - \tau) x_1^S\bigl(b, (\tau, \mathbf{1}_c)\bigr).
\]
Since $t \mapsto x_1^B(t, (\tau, \mathbf{1}_c))$ is non-decreasing, $x_1^B \ge x_1^S$, and $u \mapsto x_1^S(b, (u, \mathbf{1}_c))$ is non-increasing by DSIC, setting $b = 1$ and applying the proof of Lemma~\ref{lem:sim-inf-zero} in Section~\ref{sec:lower-bound} verbatim yields~\eqref{eq:rdom-smaller-inf-zero}.
\end{proof}

Next, fix $p, \rho \in (0, 1)$. For each $\tau \in (0, 1)$ with $b_\tau := 1 - \rho(1 - \tau) \in (\tau, 1)$, let $F_S^{(\tau)}$ place mass $p$ at $\tau$ and $1 - p$ at $1$, while under $F_B \RDOM F_S$ let the buyer distribution $F_B^{(\tau)}$ place mass $r(1 - p)$ at $1$ instead of $1 - p$ and the remaining mass $1 - r(1 - p)$ at $b_\tau$. Straightforward expectation calculations analogous to Lemmas~\ref{lem:sim-opt-bound} and~\ref{lem:sim-gft-bound} give the following bounds on $\OPT(1, 1)$ and $\mathrm{GFT}_M(1, 1 + c)$.

\begin{lemma}\label{lem:rdom-smaller-opt}
For any $r \in (0, 1]$ and $p, \rho, \tau \in (0, 1)$, we have $F_B^{(\tau)} \RDOM F_S^{(\tau)}$ and
\begin{equation}\label{eq:rdom-smaller-opt}
    \OPT\bigl(1, 1; F_B^{(\tau)}, F_S^{(\tau)}\bigr) \ge p(1 - \rho)(1 - \tau) > 0.
\end{equation}
\end{lemma}

\begin{proof}
Set $\alpha := r(1 - p) \in (0, r)$. For $x < b_\tau$, $1 - F_B^{(\tau)}(x) = 1 \ge r \ge r(1 - F_S^{(\tau)}(x))$; for $b_\tau \le x < 1$, since $b_\tau > \tau$, we have $1 - F_B^{(\tau)}(x) = \alpha = r(1 - p) = r(1 - F_S^{(\tau)}(x))$; and for $x \ge 1$, both sides equal $0$, establishing $F_B^{(\tau)} \RDOM F_S^{(\tau)}$. Conditional on $s_1 = \tau$, which occurs with probability $p$, a buyer $b \sim F_B^{(\tau)}$ has $b - \tau = 1 - \tau$ with probability $\alpha$ and $b_\tau - \tau = (1 - \rho)(1 - \tau)$ with probability $1 - \alpha$, giving $\OPT(1, 1; F_B^{(\tau)}, F_S^{(\tau)}) = p[\alpha(1 - \tau) + (1 - \alpha)(1 - \rho)(1 - \tau)] \ge p(1 - \rho)(1 - \tau) > 0$.
\end{proof}

\begin{lemma}\label{lem:rdom-smaller-bounds}
Let $c \in \Z_{\ge 2}$ and let $M$ be any seller-symmetric DSIC, IR, and WBB mechanism for $1$ buyer and $1 + c$ sellers on $[0, 1] \times [0, 1]^{1+c}$. For any $r \in (0, 1]$ and $p, \rho, \tau \in (0, 1)$,
\begin{equation}\label{eq:rdom-ml-step3-ub}
\begin{aligned}
    \mathrm{GFT}_M\bigl(1, 1 + c; F_B^{(\tau)}, F_S^{(\tau)}\bigr)
    &\le \binom{c+1}{2} p^2(1 - \tau) \\
    &\quad + (c + 1)p(1 - p)^c(1 - \tau)\Bigl[r(1 - p) + \bigl(1 - r(1 - p)\bigr)x_1^S\bigl(b_\tau, (\tau, \mathbf{1}_c)\bigr)\Bigr].
\end{aligned}
\end{equation}
\end{lemma}

\begin{proof}
Set $k := 1 + c \ge 3$ and $\alpha := r(1 - p) \in (0, r)$. In the $(1, k)$ market under $(b, \mathbf{s}_k) \sim F_B^{(\tau)} \times (F_S^{(\tau)})^k$, let $N_\tau := |\{\ell \le k : s_\ell = \tau\}| \sim \mathrm{Binomial}(k, p)$:
\begin{itemize}[leftmargin=1.5em, topsep=2pt, itemsep=2pt]
    \item \emph{$N_\tau = 0$:} Since $b \le 1$ and $s_\ell = 1$ for all $\ell$, we have $\mathrm{gft}_M(b, \mathbf{s}_k) = (b - 1) x_1^B(b, \mathbf{s}_k) \le 0$.
    \item \emph{$N_\tau \ge 2$:} By a union bound over unordered pairs, $\Pr[N_\tau \ge 2] \le \binom{k}{2} p^2 = \binom{c+1}{2} p^2$; since $b \le 1$ and $s_\ell \ge \tau$ for all $\ell$, we have $\mathrm{gft}_M(b, \mathbf{s}_k) = \sum_{\ell=1}^k (b - s_\ell) x_\ell^S(b, \mathbf{s}_k) \le (1 - \tau) x_1^B(b, \mathbf{s}_k) \le 1 - \tau$, contributing expected GFT at most $\binom{c+1}{2} p^2(1 - \tau)$.
    \item \emph{$N_\tau = 1$:} This event occurs with probability $k p(1 - p)^{k-1} = (c + 1)p(1 - p)^c$. By seller symmetry, condition on $\mathbf{s}_k = (\tau, \mathbf{1}_c)$. Under $b \sim F_B^{(\tau)}$, with probability $\alpha = r(1 - p)$ we have $b = 1$ and $\mathrm{gft}_M \le 1 - \tau$, while with probability $1 - \alpha = 1 - r(1 - p)$ we have $b = b_\tau$ and~\eqref{eq:rdom-smaller-profile-ub} gives $\mathrm{gft}_M(b_\tau, (\tau, \mathbf{1}_c)) \le (b_\tau - \tau) x_1^S(b_\tau, (\tau, \mathbf{1}_c)) \le (1 - \tau) x_1^S(b_\tau, (\tau, \mathbf{1}_c))$.
\end{itemize}
Summing the three cases establishes~\eqref{eq:rdom-ml-step3-ub}.
\end{proof}

\begin{proof}[Proof of Theorem~\ref{thm:rdom-mech-lb}]
Having handled $c = 1$ above, assume $c \ge 2$ and $(c + 1)r < 1$. Dividing~\eqref{eq:rdom-ml-step3-ub} from Lemma~\ref{lem:rdom-smaller-bounds} by $p(1 - \tau) > 0$, taking $\inf_{\tau \in (0, 1)}$, and applying~\eqref{eq:rdom-smaller-inf-zero} from Lemma~\ref{lem:rdom-smaller-inf-zero} gives
\[
    \inf_{\tau \in (0, 1)} \frac{\mathrm{GFT}_M\bigl(1, 1 + c; F_B^{(\tau)}, F_S^{(\tau)}\bigr)}{p(1 - \tau)} \le \binom{c+1}{2} p + (c + 1)(1 - p)^{c+1} r,
\]
whereas Lemma~\ref{lem:rdom-smaller-opt} gives $\frac{\OPT(1, 1; F_B^{(\tau)}, F_S^{(\tau)})}{p(1 - \tau)} \ge 1 - \rho$ for all $\tau \in (0, 1)$. As $p, \rho \downarrow 0$, the upper bound converges to $(c + 1)r < 1$ while $1 - \rho \to 1$, so choosing $p, \rho \in (0, 1)$ sufficiently small yields $\tau \in (0, 1)$ with $\mathrm{GFT}_M\bigl(1, 1 + c; F_B^{(\tau)}, F_S^{(\tau)}\bigr) < \OPT\bigl(1, 1; F_B^{(\tau)}, F_S^{(\tau)}\bigr)$.
\end{proof}

\paragraph{Two-sided recruitment (Theorem~\ref{thm:general-two-sided-mech-lb}).}
Now let $r \in (0, 1)$ and $1 \le c < \lceil \frac{1}{64 r} \rceil$ (so $r \le c r < \frac{1}{64}$), and set $n := \lceil \frac{2}{5 r} \rceil$ and $k := n + c$. We may again assume without loss of generality that $M$ is restricted to $[0, 1]^k \times [0, 1]^k$ and is buyer-and-seller symmetric (Lemma~\ref{lem:seller-sym-reduction}).

When a single seller has cost $s_1 \in (0, 1)$ and a single buyer has positive value $b \in (0, 1]$ while the remaining $k - 1$ buyers have value $0$ and the remaining $k - 1$ sellers have cost $1$, symmetry reduces the profile to $\bigl((b, \mathbf{0}_{k-1}), (s_1, \mathbf{1}_{k-1})\bigr)$, where $\mathbf{0}_{k-1} := (0, \dots, 0)$ and $\mathbf{1}_{k-1} := (1, \dots, 1)$ in $\R^{k-1}$. Because the remaining $k - 1$ buyers at $0$ and $k - 1$ sellers at $1$ cannot profitably trade, the argument of Lemma~\ref{lem:sim-inf-zero} forces the minimum of Buyer~1's and Seller~1's trade probabilities at $(b_\tau, \tau)$ to have infimum $0$ over $\tau \in (0, 1)$.

\begin{lemma}\label{lem:two-sided-inf-zero}
Let $k \in \Z_{\ge 2}$ and let $M$ be any DSIC, IR, and WBB mechanism for $k$ buyers and $k$ sellers on $[0, 1]^k \times [0, 1]^k$, and define
\[
    q_M(b, s_1) := \min\Bigl(x_1^B\bigl((b, \mathbf{0}_{k-1}), (s_1, \mathbf{1}_{k-1})\bigr), \, x_1^S\bigl((b, \mathbf{0}_{k-1}), (s_1, \mathbf{1}_{k-1})\bigr)\Bigr) \in [0, 1].
\]
For every $\rho \in (0, 1)$, letting $b_\tau := 1 - \rho(1 - \tau)$, we have
\begin{equation}\label{eq:two-sided-inf-zero}
    \inf_{\tau \in (0, 1)} q_M(b_\tau, \tau) = 0.
\end{equation}
\end{lemma}

\begin{proof}
For any profile $(\mathbf{b}_k, \mathbf{s}_k) \in [0, 1]^k \times [0, 1]^k$ with $\mathbf{b}_k := (b_1, \dots, b_k)$, Subsection~\ref{subsec:mechanisms} gives $\sum_{j=1}^k x_j^B(\mathbf{b}_k, \mathbf{s}_k) = \sum_{\ell=1}^k x_\ell^S(\mathbf{b}_k, \mathbf{s}_k)$ with $x_j^B, x_\ell^S \in [0, 1]$; letting $u_j^B := b_j \, x_j^B - p_j^B \ge 0$ and $u_\ell^S := p_\ell^S - s_\ell \, x_\ell^S \ge 0$ denote the agents' expected utilities under IR, WBB implies $\sum_{j=1}^k u_j^B(\mathbf{b}_k, \mathbf{s}_k) + \sum_{\ell=1}^k u_\ell^S(\mathbf{b}_k, \mathbf{s}_k) \le \mathrm{gft}_M(\mathbf{b}_k, \mathbf{s}_k)$.

Define the shorthand
\[
    \tilde{x}_1^B(b, s_1) := x_1^B\bigl((b, \mathbf{0}_{k-1}), (s_1, \mathbf{1}_{k-1})\bigr), \qquad
    \tilde{x}_1^S(b, s_1) := x_1^S\bigl((b, \mathbf{0}_{k-1}), (s_1, \mathbf{1}_{k-1})\bigr),
\]
so $q_M(b, s_1) = \min(\tilde{x}_1^B(b, s_1), \tilde{x}_1^S(b, s_1))$. On the profile $((b, \mathbf{0}_{k-1}), (\tau, \mathbf{1}_{k-1}))$ for any $0 < \tau < b \le 1$, we have $b_j = 0$ for $j \ge 2$ and $s_\ell = 1$ for $\ell \ge 2$. Feasibility gives $\tilde{x}_1^B(b, \tau) \le \sum_{j=1}^k x_j^B = \tilde{x}_1^S(b, \tau) + \sum_{\ell=2}^k x_\ell^S$ with $x_\ell^S \ge 0$; since $\tau > 0$ and $1 - b \ge 0$, we have
\[
    \mathrm{gft}_M\bigl((b, \mathbf{0}_{k-1}), (\tau, \mathbf{1}_{k-1})\bigr) = b \, \tilde{x}_1^B(b, \tau) - \tau \, \tilde{x}_1^S(b, \tau) - \sum_{\ell=2}^k x_\ell^S,
\]
which is bounded above both by $b \, \tilde{x}_1^B(b, \tau) - \tau(\tilde{x}_1^S(b, \tau) + \sum_{\ell=2}^k x_\ell^S) \le (b - \tau)\tilde{x}_1^B(b, \tau)$ and by $(b - \tau)\tilde{x}_1^S(b, \tau) - (1 - b)\sum_{\ell=2}^k x_\ell^S \le (b - \tau)\tilde{x}_1^S(b, \tau)$, giving
\begin{equation}\label{eq:two-sided-profile-ub}
    \mathrm{gft}_M\bigl((b, \mathbf{0}_{k-1}), (\tau, \mathbf{1}_{k-1})\bigr) \le (b - \tau) q_M(b, \tau).
\end{equation}
Moreover, since $u_j^B, u_\ell^S \ge 0$ by IR, DSIC, Myerson's payment identity, and WBB yield
\[
    \int_\tau^b \tilde{x}_1^B(t, \tau) \, dt + \int_\tau^b \tilde{x}_1^S(b, u) \, du \le (b - \tau) q_M(b, \tau).
\]
Since $q_M(1, \tau) \le \tilde{x}_1^S(1, \tau)$, $u \mapsto \tilde{x}_1^S(1, u)$ is non-increasing by DSIC, and DSIC monotonicity of $t \mapsto \tilde{x}_1^B(t, \tau)$ gives $\tilde{x}_1^B(t, \tau) \ge \tilde{x}_1^B(b_\tau, \tau) \ge q_M(b_\tau, \tau)$ for all $\tau \in (0, 1)$ and $t \in [b_\tau, 1]$, the proof of Lemma~\ref{lem:sim-inf-zero} in Section~\ref{sec:lower-bound} (at $b = 1$) applies verbatim to establish~\eqref{eq:two-sided-inf-zero}.
\end{proof}

To handle $c$ additional traders on \emph{both} sides, we thin out the buyer distribution above: since $F_B \RDOM F_S$ on $[0, b_\tau)$ only requires a buyer to have positive value with probability at least $r$, we set $p_B := 2r \in (0, \frac{1}{32})$ and place mass $1 - p_B$ at $0$. Specifically, fixing $p_S, \rho \in (0, 1)$, for each $\tau \in (0, 1)$ with $b_\tau := 1 - \rho(1 - \tau)$, let $\widetilde{F}_S^{(\tau)}$ place mass $p_S$ at $\tau$ and $1 - p_S$ at $1$, while $\widetilde{F}_B^{(\tau)}$ places mass $1 - p_B$ at $0$, mass $p_B - r(1 - p_S) > 0$ at $b_\tau$, and mass $r(1 - p_S)$ at $1$. Evaluating expectations over the realizations of $(\widetilde{F}_B^{(\tau)}, \widetilde{F}_S^{(\tau)})$ yields the following bounds on $\OPT(n, n)$ and $\mathrm{GFT}_M(k, k)$.

\begin{lemma}\label{lem:two-sided-opt}
For every $p_S, \rho, \tau \in (0, 1)$, we have $\widetilde{F}_B^{(\tau)} \RDOM \widetilde{F}_S^{(\tau)}$ and
\begin{equation}\label{eq:gen-lb-opt}
    \OPT\bigl(n, n; \widetilde{F}_B^{(\tau)}, \widetilde{F}_S^{(\tau)}\bigr) \ge n p_S(1 - p_S)^{n-1}\Bigl[1 - (1 - p_B)^n\Bigr](1 - \rho)(1 - \tau) > 0.
\end{equation}
\end{lemma}

\begin{proof}
For $x < 0$, $1 - \widetilde{F}_B^{(\tau)}(x) = 1 \ge r = r(1 - \widetilde{F}_S^{(\tau)}(x))$; for $0 \le x < b_\tau$, $1 - \widetilde{F}_B^{(\tau)}(x) = p_B = 2r > r \ge r(1 - \widetilde{F}_S^{(\tau)}(x))$; for $b_\tau \le x < 1$, since $b_\tau > \tau$, we have $1 - \widetilde{F}_B^{(\tau)}(x) = r(1 - p_S) = r(1 - \widetilde{F}_S^{(\tau)}(x))$; and for $x \ge 1$, both sides equal $0$. Thus $\widetilde{F}_B^{(\tau)} \RDOM \widetilde{F}_S^{(\tau)}$. In the $(n, n)$ market, the independent events $\{|\{\ell \le n : s_\ell = \tau\}| = 1\}$ and $\{|\{j \le n : b_j \ge b_\tau\}| \ge 1\}$ have probabilities $n p_S(1 - p_S)^{n-1}$ and $1 - (1 - p_B)^n$, respectively; on their intersection, trading the unique low-cost seller at $\tau$ with a positive-valuation buyer at $b_j \ge b_\tau$ yields surplus $b_j - \tau \ge b_\tau - \tau = (1 - \rho)(1 - \tau)$, establishing~\eqref{eq:gen-lb-opt}.
\end{proof}

\begin{lemma}\label{lem:two-sided-gft-ub}
Let $M$ be any buyer-and-seller symmetric DSIC, IR, and WBB mechanism for $k = n + c$ buyers and $k$ sellers on $[0, 1]^k \times [0, 1]^k$. For every $p_S, \rho, \tau \in (0, 1)$,
\begin{equation}\label{eq:two-sided-gft-upper}
\begin{aligned}
    &\mathrm{GFT}_M\bigl(k, k; \widetilde{F}_B^{(\tau)}, \widetilde{F}_S^{(\tau)}\bigr) \\
    &\quad \le (1 - \tau) k(k - 1) p_S^2 + k p_S(1 - p_S)^{k-1}(1 - \tau)\Bigl[1 - (1 - p_B)^k - k p_B(1 - p_B)^{k-1}\Bigr] \\
    &\qquad + k^2 p_S(1 - p_S)^{k-1}(1 - p_B)^{k-1}(1 - \tau)\Bigl[r(1 - p_S) + \bigl(p_B - r(1 - p_S)\bigr) q_M(b_\tau, \tau)\Bigr].
\end{aligned}
\end{equation}
\end{lemma}

\begin{proof}
Set $\alpha := r(1 - p_S) \in (0, r) \subset (0, p_B)$. Let $A_B := \{j \le k : b_j \ge b_\tau\}$ and $A_S := \{\ell \le k : s_\ell = \tau\}$, so $N_B := |A_B| \sim \mathrm{Binomial}(k, p_B)$ and $N_S := |A_S| \sim \mathrm{Binomial}(k, p_S)$. Setting $X_B := \sum_{j \in A_B} x_j^B \in [0, N_B]$ and $X_S := \sum_{\ell \in A_S} x_\ell^S \in [0, N_S]$, since $b_j = 0$ for $j \notin A_B$, $b_j \le 1$ for $j \in A_B$, and $s_\ell = 1$ for $\ell \notin A_S$, feasibility ($\sum_{j=1}^k x_j^B = \sum_{\ell=1}^k x_\ell^S$ with $x_j^B, x_\ell^S \ge 0$) forces $\sum_{\ell \notin A_S} x_\ell^S \ge \max(0, X_B - X_S)$, so
\begin{equation}\label{eq:two-sided-pointwise}
    \mathrm{gft}_M(\mathbf{b}_k, \mathbf{s}_k) \le X_B - \tau X_S - \max(0, X_B - X_S) = \min(X_B, X_S) - \tau X_S \le (1 - \tau)\min(N_B, N_S).
\end{equation}
We partition according to the realization of $N_S$:
\begin{itemize}[leftmargin=1.5em, topsep=2pt, itemsep=2pt]
    \item \emph{$N_S = 0$:} Here~\eqref{eq:two-sided-pointwise} gives $\mathrm{gft}_M(\mathbf{b}_k, \mathbf{s}_k) \le 0$.
    \item \emph{$N_S \ge 2$:} Since $\min(N_B, N_S) \le N_S$ and Bernoulli's inequality gives $\E[N_S \I[N_S \ge 2]] = \E[N_S] - \Pr[N_S = 1] = k p_S\bigl(1 - (1 - p_S)^{k-1}\bigr) \le k(k - 1) p_S^2$,~\eqref{eq:two-sided-pointwise} shows that this case contributes expected GFT at most $(1 - \tau) k(k - 1) p_S^2$.
    \item \emph{$N_S = 1$:} This event occurs with probability $k p_S(1 - p_S)^{k-1}$. By seller symmetry, condition on $\mathbf{s}_k = (\tau, \mathbf{1}_{k-1})$. When $N_B = 0$,~\eqref{eq:two-sided-pointwise} gives $\mathrm{gft}_M \le 0$; when $N_B \ge 2$, which occurs with probability $1 - (1 - p_B)^k - k p_B(1 - p_B)^{k-1}$,~\eqref{eq:two-sided-pointwise} with $N_S = 1$ gives $\mathrm{gft}_M \le 1 - \tau$. When $N_B = 1$, which occurs with probability $k p_B(1 - p_B)^{k-1}$, buyer symmetry reduces the profile to $\bigl((b_1, \mathbf{0}_{k-1}), (\tau, \mathbf{1}_{k-1})\bigr)$ with $b_1 \in \{b_\tau, 1\}$. Conditional on $b_1 \in \{b_\tau, 1\}$, we have $b_1 = 1$ with probability $\frac{\alpha}{p_B}$, where $\mathrm{gft}_M \le 1 - \tau$, and $b_1 = b_\tau$ with probability $\frac{p_B - \alpha}{p_B}$, where~\eqref{eq:two-sided-profile-ub} gives $\mathrm{gft}_M \le (b_\tau - \tau) q_M(b_\tau, \tau) \le (1 - \tau) q_M(b_\tau, \tau)$, yielding $\E[\mathrm{gft}_M \mid N_B = 1, N_S = 1] \le (1 - \tau)\bigl[\frac{\alpha}{p_B} + \frac{p_B - \alpha}{p_B} q_M(b_\tau, \tau)\bigr]$.
\end{itemize}
Summing the three cases and substituting $\alpha = r(1 - p_S)$ establishes~\eqref{eq:two-sided-gft-upper}.
\end{proof}

Dividing~\eqref{eq:gen-lb-opt} and~\eqref{eq:two-sided-gft-upper} by $p_S(1 - \tau)$ and letting $p_S, \rho \downarrow 0$ (after taking $\inf_{\tau \in (0, 1)}$) compares $n[1 - (1 - 2r)^n]$ against $k[1 - (1 - 2r)^k] - k^2 r(1 - 2r)^{k-1}$. Because $k r = \Theta(1)$ while $c r < 1/64$, for $k = n + c$, $k[1 - (1 - 2r)^k]$ exceeds $n[1 - (1 - 2r)^n]$ additively by at most $O(c k r) = O(c)$, which is strictly smaller than the subtracted term $k^2 r(1 - 2r)^{k-1} = \Theta(1/r)$.

\begin{lemma}\label{lem:two-sided-gap}
For any $r \in (0, 1)$ and $c \in \Z_{\ge 1}$ with $c < \lceil \frac{1}{64 r} \rceil$, $n := \lceil \frac{2}{5 r} \rceil$, and $k := n + c$,
\begin{equation}\label{eq:two-sided-limit-gap}
    k\bigl[1 - (1 - 2r)^k\bigr] - k^2 r(1 - 2r)^{k-1} < n\bigl[1 - (1 - 2r)^n\bigr].
\end{equation}
\end{lemma}

\begin{proof}
Since $c \ge 1$ is an integer satisfying $c < \lceil \frac{1}{64 r} \rceil$, we have $1 \le c \le \lceil \frac{1}{64 r} \rceil - 1 < \frac{1}{64 r}$, so $r \le c r < \frac{1}{64}$ and $2r < \frac{1}{32} < 1$. Moreover, $n = \lceil \frac{2}{5 r} \rceil$ satisfies $\frac{2}{5 r} \le n < \frac{2}{5 r} + 1$, so $k r > n r \ge \frac{2}{5}$ and
\[
    k r = (n - 1) r + r + c r < \frac{2}{5} + \frac{1}{64} + \frac{1}{64} = \frac{69}{160} < \frac{7}{16}.
\]
For any integer $s \ge 1$, Bernoulli's inequality gives $(1 - 2r)^s \ge 1 - 2 s r$, i.e., $1 - (1 - 2r)^s \le 2 s r$. Applying this for $s \in \{k, c\}$ along with $(1 - 2r)^n \le 1$ and $n < k$ gives
\[
    k\bigl[1 - (1 - 2r)^k\bigr] - n\bigl[1 - (1 - 2r)^n\bigr] = c\bigl[1 - (1 - 2r)^k\bigr] + n(1 - 2r)^n\bigl[1 - (1 - 2r)^c\bigr] \le 2 c k r + 2 c n r < 4 c k r.
\]
Meanwhile, writing $k - 1 = m_1 + m_2$ with $m_1 := \lfloor \frac{k - 1}{2} \rfloor \ge 0$ and $m_2 := \lceil \frac{k - 1}{2} \rceil \ge 0$, note that $2 m_1 \le 2 m_2 \le k$ and hence $0 \le 2 m_1 r \le 2 m_2 r \le k r < \frac{7}{16} < 1$. Applying Bernoulli's inequality $(1 - 2r)^{m_j} \ge 1 - 2 m_j r \ge 1 - k r > \frac{9}{16} > 0$ for $j \in \{1, 2\}$ gives
\[
    (1 - 2r)^{k-1} = (1 - 2r)^{m_1}(1 - 2r)^{m_2} \ge (1 - k r)^2 > \left(\frac{9}{16}\right)^2 = \frac{81}{256}.
\]
Combining these bounds with $c r < \frac{1}{64}$ and $k r > \frac{2}{5}$ yields
\[
    \frac{k\bigl[1 - (1 - 2r)^k\bigr] - n\bigl[1 - (1 - 2r)^n\bigr]}{k^2 r(1 - 2r)^{k-1}} < \frac{4 c k r}{k^2 r(1 - 2r)^{k-1}} = \frac{4(c r)}{(k r)(1 - 2r)^{k-1}} < \frac{4 \cdot \frac{1}{64}}{\frac{2}{5} \cdot \frac{81}{256}} = \frac{40}{81} < 1.
\]
Multiplying by $k^2 r(1 - 2r)^{k-1} > 0$ and rearranging establishes~\eqref{eq:two-sided-limit-gap}.
\end{proof}

\begin{proof}[Proof of Theorem~\ref{thm:general-two-sided-mech-lb}]
Fix $p_S, \rho \in (0, 1)$. Dividing~\eqref{eq:two-sided-gft-upper} from Lemma~\ref{lem:two-sided-gft-ub} by $p_S(1 - \tau) > 0$, taking $\inf_{\tau \in (0, 1)}$, and applying~\eqref{eq:two-sided-inf-zero} from Lemma~\ref{lem:two-sided-inf-zero} yields
\[
\begin{aligned}
    &\inf_{\tau \in (0, 1)} \frac{\mathrm{GFT}_M\bigl(k, k; \widetilde{F}_B^{(\tau)}, \widetilde{F}_S^{(\tau)}\bigr)}{p_S(1 - \tau)} \\
    &\quad \le k(k - 1) p_S + k(1 - p_S)^{k-1}\Bigl[1 - (1 - p_B)^k - k\bigl(p_B - r(1 - p_S)\bigr)(1 - p_B)^{k-1}\Bigr],
\end{aligned}
\]
whereas Lemma~\ref{lem:two-sided-opt} gives $\frac{\OPT(n, n; \widetilde{F}_B^{(\tau)}, \widetilde{F}_S^{(\tau)})}{p_S(1 - \tau)} \ge n(1 - p_S)^{n-1}[1 - (1 - p_B)^n](1 - \rho)$ for all $\tau \in (0, 1)$. Since $p_B = 2r$, as $p_S, \rho \downarrow 0$ these upper and lower bounds converge to the left- and right-hand sides of~\eqref{eq:two-sided-limit-gap}, respectively. By Lemma~\ref{lem:two-sided-gap}, choosing $p_S, \rho \in (0, 1)$ sufficiently small therefore yields $\tau \in (0, 1)$ with $\mathrm{GFT}_M\bigl(n + c, n + c; \widetilde{F}_B^{(\tau)}, \widetilde{F}_S^{(\tau)}\bigr) < \OPT\bigl(n, n; \widetilde{F}_B^{(\tau)}, \widetilde{F}_S^{(\tau)}\bigr)$.
\end{proof}

\end{document}